\documentclass[aps,pra,superscriptaddress,notitlepage, amsfonts,longbibliography, twocolumn]{revtex4-2}

\usepackage{qcircuit}
\usepackage[dvipdfmx]{graphicx}
\usepackage{amsmath,amssymb,amsthm,mathrsfs,amsfonts,dsfont}
\usepackage{subfigure, epsfig}
\usepackage{braket}
\usepackage{bm}
\usepackage{enumerate}
\usepackage{physics}
\usepackage{comment}
\newtheorem{thm}{Theorem}
\newtheorem{lem}{Lemma}
\newtheorem{prop}{Proposition}
\newtheorem{cor}{Corollary}
\usepackage[colorlinks,linkcolor=blue,citecolor=blue]{hyperref}
\usepackage{stmaryrd}
\usepackage{booktabs}
\usepackage[dvipsnames]{xcolor}

\makeatletter
\renewcommand*{\numberline}[1]{\hb@xt@1em{#1\hfil}} 
\makeatother

\graphicspath{{./fig/}}

\begin{document}

\title{Quantum advantages for syndrome-aware noisy logical observable estimation}

\author{Kento Tsubouchi}
\email{tsubouchi@noneq.t.u-tokyo.ac.jp}
\affiliation{\mbox{Department of Applied Physics, University of Tokyo, 7-3-1 Hongo, Bunkyo-ku, Tokyo 113-8656, Japan}}

\author{Hyukgun Kwon}
\affiliation{\mbox{Department of Physics and Astronomy, Sejong University, 209 Neungdong-ro Gwangjin-gu, Seoul 05006, Republic of Korea}}

\author{Liang Jiang}
\affiliation{Pritzker School of Molecular Engineering, University of Chicago, Chicago, Illinois 60637, USA}

\author{Nobuyuki Yoshioka}
\email{ny.nobuyoshioka@gmail.com}
\affiliation{\mbox{International Center for Elementary Particle Physics, University of Tokyo, 7-3-1 Hongo, Bunkyo-ku, Tokyo 113-0033, Japan}}

\begin{abstract}
Recent progress in fault-tolerant quantum computing suggests that leveraging error-syndrome information at the logical layer can substantially improve performance, including the estimation of logical observables from noisy states.
In this work, based on quantum estimation theory, we develop an information-theoretic framework to quantify the utility of error syndromes for noisy logical observable estimation.
We distinguish two operational regimes of such syndrome-aware protocols: \emph{classical} protocols, in which the logical measurement basis is fixed and syndrome information is used only in classical post-processing, and \emph{quantum} protocols, in which the logical quantum control can be tailored to depend on the observed error syndrome.
For classical syndrome-aware protocols, we prove a universal limitation: on average, syndrome information can improve the effective logical error rate by at most a factor of two, implying at most a quadratic reduction in sampling overhead.
In contrast, once syndrome-conditioned quantum control is permitted, we demonstrate that the effective logical error rate decays exponentially with the number of code blocks.
These findings provide fundamental guidance for designing future fault-tolerant architectures that actively exploit syndrome records rather than discarding them after decoding.
\end{abstract}

\maketitle

\section{Introduction}
\label{sec_introduction}
Fault-tolerant quantum computing based on quantum error correction provides a powerful framework to suppress errors in quantum processors~\cite{shor1995scheme, knill1996threshold, aharonov1997fault, lidar2013quantum}.
In quantum error correction, logical quantum information is encoded into a larger Hilbert space of physical qubits by introducing redundancy.
Although physical qubits are subject to noise, one repeatedly measures redundant degrees of freedom to extract error syndromes, and then infers and corrects the underlying errors via decoding.
With continuous experimental progress demonstrating quantum error correction~\cite{ofek2016extending, krinner2022realizing, google2023suppressing, sivak2023real, bluvstein2024logical, ai2024quantum}, it is expected that sufficiently large-scale fault-tolerant devices will eventually enable reliable quantum computation with negligible logical errors.

However, especially in the early stages of fault-tolerant hardware development, limited resources make the implementation of full quantum error correction challenging.
As a result, logical errors may remain in the logical state obtained after decoding, and it becomes necessary to design computational protocols that explicitly account for residual logical noise.
In other words, one must estimate noiseless logical properties of interest from noisy logical states produced after decoding.
For example, one can apply quantum error mitigation protocols~\cite{temme2017error, endo2021hybrid, cai2023quantum, kim2023evidence} at the logical layer to further suppress logical errors~\cite{piveteau2021error, lostaglio2021error, suzuki2022quantum, tsubouchi2024symmetric}.

\begin{figure}[t]
    \begin{center}
        \includegraphics[width=0.99\linewidth]{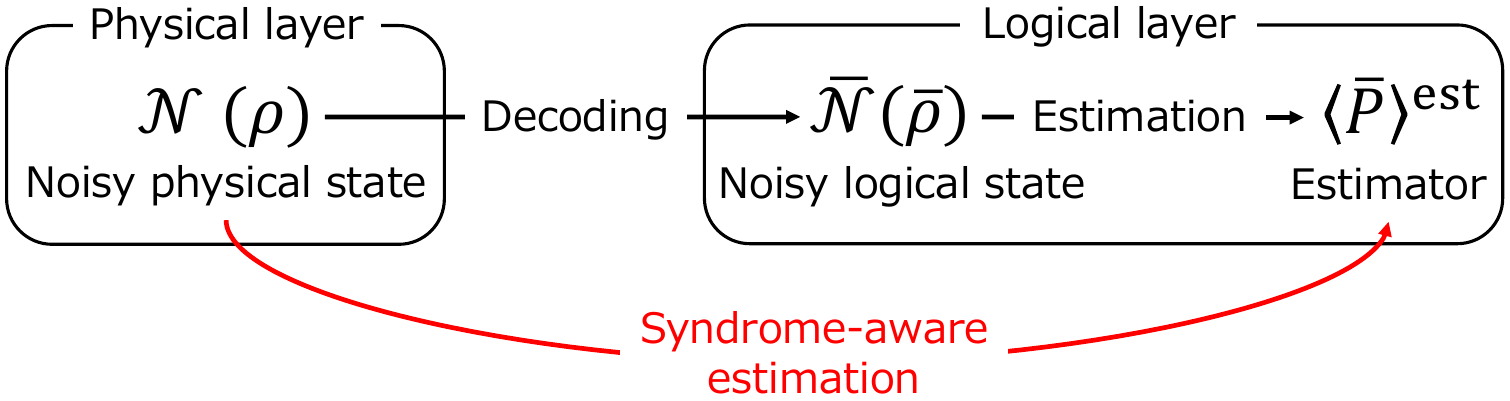}
        \caption{Schematic illustration of syndrome-aware estimation. In conventional estimation protocols performed at the logical layer, one first obtains a noisy logical state by decoding a noisy physical state and then performs estimation using only the decoded logical state, while the error syndromes are not explicitly used at the estimation stage. In contrast, syndrome-aware estimation protocols explicitly incorporate the error syndromes into the estimation procedure, and can be viewed as a joint implementation of decoding and estimation.}
        \label{fig_decoding_estimation}
    \end{center}
\end{figure}

\begin{figure*}[t]
    \begin{center}
        \includegraphics[width=0.99\linewidth]{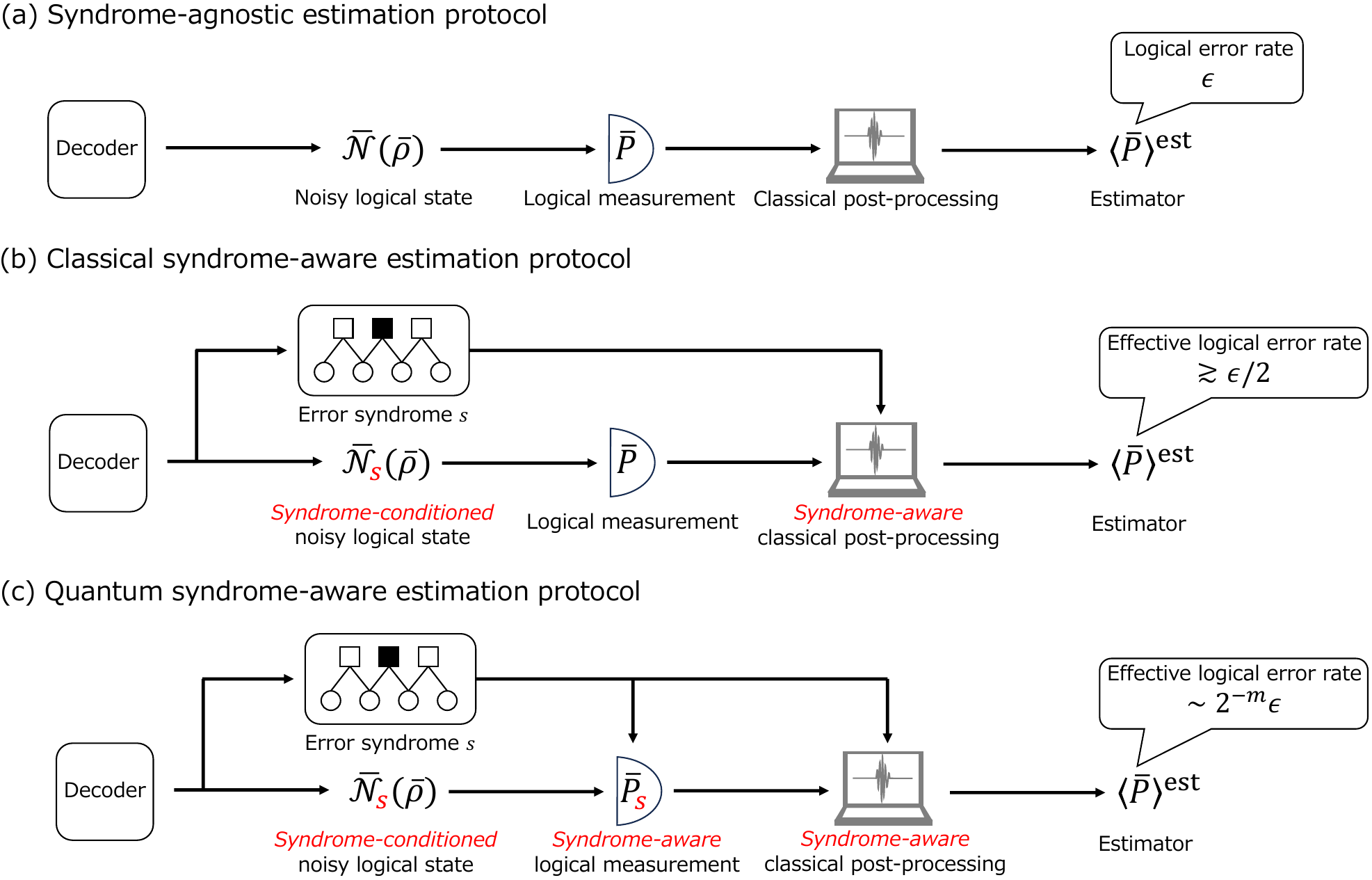}
        \caption{Schematic illustration of syndrome-agnostic and syndrome-aware estimation protocols. (a) In the syndrome-agnostic estimation protocol, the decoder outputs an averaged noisy logical state, and we perform a fixed logical measurement on this state followed by classical post-processing to estimate a noiseless observable. The error-syndrome information is not used in either the measurement or the classical post-processing. The effect of the error on the resulting estimator is characterized by a logical error rate $\epsilon$. (b) In the classical syndrome-aware estimation protocol, we assume that the decoder outputs the observed error syndrome $s$ together with a syndrome-conditioned noisy logical state. The error-syndrome information is used only at the classical post-processing stage, while the logical measurement basis does not depend on the observed syndrome $s$. We characterize the effect of the error in this scenario by the \textit{effective logical error rate} $\epsilon^{\mathrm{cSynd}}$, which is lower-bounded by one half of the original logical error rate $\epsilon$ on average. (c) In the quantum syndrome-aware estimation protocol, we again assume that the decoder outputs the observed error syndrome $s$ and a syndrome-conditioned noisy logical state. Unlike the classical protocol, we additionally allow the logical measurement basis to depend on $s$, followed by syndrome-aware classical post-processing. This allows the effective logical error rate $\epsilon^{\mathrm{qSynd}}$ to be exponentially smaller than the original logical error rate $\epsilon$ as a function of the number of code blocks $m$.}
        \label{fig_syndrome_aware_protocols}
    \end{center}
\end{figure*}

Recently, there has been growing interest in protocols that go beyond using only the decoded logical state.
By exploiting redundant information available at the physical layer---namely, the observed error syndromes---and constructing protocols conditioned on the syndrome record, one can further enhance the performance of the estimation (Fig.~\ref{fig_decoding_estimation}).
For example, by estimating the logical error probability conditioned on the observed syndromes and post-selecting runs with low estimated error probability, one can further suppress the logical error rate and improve the estimation performance~\cite{bombin2024fault, meister2024efficient, smith2024mitigating, gidney2025yoked, lee2025efficient, xie2026simple, kishi2026even}.
Moreover, syndrome-aware quantum error mitigation~\cite{bluvstein2024logical, zhou2025error, dincua2025error, aharonov2025syndrome}, which performs post-processing conditioned on the observed syndromes, can reduce sampling overhead and improve estimation precision compared to error-mitigation protocols that do not use syndrome information.

In such a setting, a natural question is how far syndrome information can improve logical-layer estimation.
While the syndrome record provides additional information about residual logical noise, it is not obvious what ultimate performance is achievable once we allow the estimation procedure to be conditioned on the observed syndrome.
Existing information-theoretic limitations for estimation from noisy quantum states~\cite{takagi2022fundamental, takagi2023universal, tsubouchi2023universal, quek2024exponentially} do not directly answer this question, because they treat the estimator as acting only on a single noisy quantum state and do not incorporate the extra information carried by the syndrome record.
Consequently, it remains unclear how to fully exploit syndrome information and what the fundamental limits are in this setting.
Establishing fundamental limits in this syndrome-aware setting is therefore essential for assessing the true value of syndrome data and for designing fault-tolerant architectures and mitigation protocols that achieve reliable estimation with substantially reduced resources.

In this work, based on quantum estimation theory~\cite{helstrom1969quantum, holevo2011probabilistic, hayashi2006quantum}, we develop a theoretical framework for quantifying how useful error-syndrome information is for estimating logical observables from noisy states.
We first define \emph{syndrome-agnostic} estimation protocols as those that perform estimation directly on the decoded noisy logical state without using the syndrome record (Fig.~\ref{fig_syndrome_aware_protocols}(a)).
We then consider two operational regimes for \emph{syndrome-aware} protocols: (i) \emph{classical} syndrome-aware protocols, where the logical measurement basis is fixed for all syndromes and syndrome information is used only in classical post-processing (Fig.~\ref{fig_syndrome_aware_protocols}(b)); and (ii) \emph{quantum} syndrome-aware protocols, where the logical measurement (and more generally, logical quantum control) can be tailored to depend on the observed syndrome (Fig.~\ref{fig_syndrome_aware_protocols}(c)).
We analyze and compare the Fisher information across these three settings.
To quantify the gain from syndrome information, we introduce an \emph{effective logical error rate}, defined as the logical error rate of a syndrome-agnostic protocol that would yield the same Fisher information as a given syndrome-aware protocol.

Our main message is that the impact of syndrome information depends qualitatively on whether syndrome-conditioned quantum control is available.
We first establish a universal limitation for classical syndrome-aware protocols.
While syndrome-aware classical post-processing can reduce the impact of errors compared to syndrome-agnostic protocols, we prove that the effective logical error rate can be improved, on average, by at most a factor of two compared to the logical error rate of an optimal decoder.
This implies that the sampling overhead can be improved at most quadratically by using the syndrome record classically, and hence the exponential overhead inherent to quantum error mitigation~\cite{takagi2022fundamental, takagi2023universal, tsubouchi2023universal, quek2024exponentially} cannot be avoided.
In contrast, once quantum syndrome-aware protocols are permitted---i.e., the measurement basis can be adapted to the observed syndromes---the above limitation of classical protocol can be violated.
Concretely, we demonstrate that the effective logical error rate can be made exponentially smaller than the logical error rate as the number of code blocks increases.

These results establish a clear separation between what can and cannot be achieved by leveraging error-syndrome information for noisy logical observable estimation.
On the one hand, they provide a general no-go statement showing that syndrome-dependent classical post-processing alone offers only limited improvements.
On the other hand, they identify syndrome-conditioned quantum control---in particular, adapting the measurement basis to the observed syndrome record---as the key ingredient that enables genuine asymptotic advantages, including exponential improvements in the effective logical error rate.
More broadly, our framework offers a unified, information-theoretic way to quantify the utility of error-syndrome information, and provides fundamental guidance for designing future fault-tolerant protocols that actively exploit error syndromes rather than discarding them after decoding.

The remainder of the paper is organized as follows.
In Sec.~\ref{sec_preliminaries}, we review basic tools from classical and quantum estimation theory that underpin our framework.
In Sec.~\ref{sec_problem_setup}, we introduce the problem setup and define syndrome-agnostic as well as classical and quantum syndrome-aware estimation protocols.
In Sec.~\ref{sec_syndrome_agnostic}, we analyze syndrome-agnostic protocols by deriving the corresponding Fisher information, which serves as a baseline for comparison.
In Sec.~\ref{sec_classical}, we establish fundamental limitations of classical syndrome-aware protocols by deriving a lower bound on the effective logical error rate.
In Sec.~\ref{sec_quantum}, we demonstrate exponential advantages of quantum syndrome-aware protocols.
Finally, we summarize our contributions and discuss future directions in Sec.~\ref{sec_discussion}.

\section{Preliminaries}
\label{sec_preliminaries}
In this section, we summarize basic tools from classical and quantum estimation theory~\cite{helstrom1969quantum, holevo2011probabilistic, hayashi2006quantum} that underpin our framework.
In Sec.~\ref{sec_classical_estimation_theory}, we review classical estimation theory for a classical probability distribution parameterized by a single unknown parameter.
In Sec.~\ref{sec_quantum_estimation_theory}, we review quantum estimation theory for a quantum state parameterized by multiple unknown parameters.

\subsection{Single-parameter classical estimation theory}
\label{sec_classical_estimation_theory}
Consider a classical probability distribution $\Pr(x|\theta)$ parameterized by an unknown parameter $\theta$.
Given $N$ samples $x_{1},\ldots,x_{N}$ drawn independently from $\Pr(x|\theta)$, we aim to construct an unbiased estimator $\theta^{\mathrm{est}}$ of $\theta$.
Here, an estimator is unbiased if its expectation value equals the true parameter value, i.e., $\mathbb{E}[\theta^{\mathrm{est}}]=\theta$.
In this setting, the precision of the unbiased estimation is characterized by the \emph{classical Fisher information}
\begin{equation}
    \label{eq_CFI}
    F_\theta = \sum_x \frac{1}{\Pr(x|\theta)}\qty(\partial_{\theta}\Pr(x|\theta))^2.
\end{equation}
Indeed, the classical Fisher information provides the achievable lower bound on the variance $\mathrm{Var}[\theta^{\mathrm{est}}]$ of the unbiased estimator $\theta^{\mathrm{est}}$ through the \textit{Cram\'er-Rao bound}
\begin{equation}
    \label{eq_CCRB}
    \mathrm{Var}[\theta^{\mathrm{est}}]\geq \frac{1}{N}F_\theta^{-1}.
\end{equation}
In the asymptotic regime $N\to\infty$, the asymptotically unbiased estimator that attains the Cram\'er--Rao bound is given by the maximum-likelihood estimator,
\begin{equation}
    \label{eq_MLE}
    \theta^{\mathrm{est}}=
    \arg\max_{\theta}\prod_{i=1}^N \Pr(x_i|\theta)
    = \arg\max_{\theta}\sum_{i=1}^N \log \Pr(x_i|\theta).
\end{equation}

Next, consider a joint distribution of the form $\Pr(x,s|\theta)=p_s\Pr(x|s,\theta)$, where one first observes $s$ with probability $p_s$ and then observes $x$ drawn from the conditional distribution $\Pr(x|s,\theta)$.
If $p_s$ is independent of $\theta$, the classical Fisher information of $\Pr(x,s|\theta)$ decomposes as
\begin{equation}
    \label{eq_CFI_joint}
    F_\theta = \sum_s p_s F_\theta^{(s)},
\end{equation}
where $F_\theta^{(s)}$ denotes the Fisher information of the conditional distribution $\Pr(x|s,\theta)$.
In our setting, $s$ corresponds to the error syndrome, and $\Pr(x|s,\theta)$ is the outcome distribution of a logical measurement performed on the syndrome-conditioned noisy logical state.

\subsection{Multi-parameter quantum estimation theory}
\label{sec_quantum_estimation_theory}
Consider a quantum state $\bar{\rho}(\vb*{\theta})$ parameterized by $M$ unknown parameters $\vb*{\theta}=(\theta_1,\ldots,\theta_M)$.
Since we mainly deal with logical operators, we denote operators acting on the logical Hilbert space by $\bar{\cdot}$, while ordinary (non-operator) quantities such as probability distributions and Fisher information matrices are written without $\bar{\cdot}$.
Given $N$ copies of $\bar{\rho}(\vb*{\theta})$, we aim to construct an unbiased estimator $\theta_i^{\mathrm{est}}$ of $\theta_i$.
We define the \emph{symmetric logarithmic derivative} (SLD) operator $\bar{L}_i$ by 
\begin{equation}
    \label{eq_SLD}
    \partial_{\theta_i}\bar{\rho}(\vb*{\theta})
    = \frac{1}{2}\{\bar{\rho}(\vb*{\theta}),\bar{L}_i\},
\end{equation}
and the \emph{(SLD) quantum Fisher information matrix} $J\in\mathbb{R}^{M\times M}$ by
\begin{equation}
    \label{eq_QFIM}
    J_{ij}
    = \mathrm{tr}\qty[\bar{\rho}(\vb*{\theta})\frac{1}{2}\{\bar{L}_i,\bar{L}_j\}].
\end{equation}
For a single parameter $\theta_i$ in the multi-parameter setting, we define the corresponding \emph{quantum Fisher information} as
\begin{equation}
    \label{eq_QFI}
    J_{\theta_i} = \frac{1}{(J^{-1})_{ii}}.
\end{equation}
Then, the achievable lower bound on the variance $\mathrm{Var}[\theta_i^{\mathrm{est}}]$ of the unbiased estimator $\theta_i^{\mathrm{est}}$ is given by the \emph{quantum Cram\'er--Rao bound}
\begin{equation}
    \label{eq_QCRB}
    \mathrm{Var}[\theta_i^{\mathrm{est}}]\geq \frac{1}{N}J_{\theta_i}^{-1}.
\end{equation}
In the asymptotic regime $N\to\infty$, an asymptotically unbiased estimator that attains the quantum Cram\'er--Rao bound can be constructed by first performing a measurement in the eigenbasis of the operator
\begin{equation}
    \label{eq_optimal_measurement_basis}
    \sum_{j}(J^{-1})_{ij}\bar{L}_j,
\end{equation}
and then applying a maximum-likelihood estimator to the resulting classical data~\cite{suzuki2020quantum}.
If the operator $\sum_{j}(J^{-1})_{ij}\bar{L}_j$ depends on the unknown parameters $\vb*{\theta}$, one may use an adaptive two-step strategy: measure $\sqrt{N}$ copies to obtain a rough estimate of $\vb*{\theta}$, and then measure the remaining $N-\sqrt{N}$ copies in the eigenbasis of $\sum_{j}(J^{-1})_{ij}\bar{L}_j$ evaluated at the rough estimate~\cite{yang2019attaining, zhou2020saturating}.

The discussion above assumes that all parameters $\vb*{\theta}=(\theta_1,\ldots,\theta_M)$ are unknown, so that the parameters other than $\theta_i$ act as nuisance parameters~\cite{suzuki2020quantum}.
If, instead, all parameters except $\theta_i$ are known, the achievable lower bound is given by
\begin{equation}
    \mathrm{Var}[\theta_i^{\mathrm{est}}]\geq \frac{1}{N}(J_{ii})^{-1},
\end{equation}
and the optimal measurement is the one that diagonalizes the SLD operator $\bar{L}_i$.
An important point is that the optimal measurement generally depends on whether the other parameters are treated as known or unknown.

Finally, consider a classical--quantum state of the form
\begin{equation}
    \sum_s p_s\ketbra{s}\otimes\bar{\rho}^{(s)}(\vb*{\theta}),
\end{equation}
where $p_s$ is independent of $\vb*{\theta}$.
Let $\bar{L}^{(s)}_i$ and $J^{(s)}$ denote the SLD operators and quantum Fisher information matrices associated with $\bar{\rho}^{(s)}(\vb*{\theta})$.
Then, the SLD operators and quantum Fisher information matrix for the classical--quantum state can be written as
\begin{equation}
    \label{eq_QFIM_classical_quantum}
    \bar{L}_i = \sum_s \ketbra{s}\otimes \bar{L}^{(s)}_i, 
    \qquad
    J = \sum_s p_s J^{(s)},
\end{equation}
and hence
\begin{equation}
    J_{\theta_i}
    = \frac{1}{\bigl((\sum_s p_s J^{(s)})^{-1}\bigr)_{ii}}.
\end{equation}
Accordingly, an asymptotically optimal measurement for estimating $\theta_i$ is given by measuring $s$ first and then, conditioned on $s$, measuring $\bar{\rho}^{(s)}(\vb*{\theta})$ in the eigenbasis of $\sum_j (J^{-1})_{ij}\bar{L}^{(s)}_j$.
Note that if one had access only to copies of a fixed conditional state $\bar{\rho}^{(s)}(\vb*{\theta})$, the corresponding optimal measurement would instead involve $\sum_j ((J^{(s)})^{-1})_{ij}\bar{L}^{(s)}_j$.
Thus, the optimal measurement for a given branch $s$ generally depends on the presence of the other branches $\bar{\rho}^{(s')}(\vb*{\theta})$ and their weights.
In the main text, we interpret $s$ as the error syndrome, and $\bar{\rho}^{(s)}(\vb*{\theta})$ as the noisy logical state conditioned on the syndrome $s$.

\section{Problem setup}
\label{sec_problem_setup}
In this section, we describe the problem setup considered throughout this paper.
Consider the task of preparing a $k$-qubit logical state
\begin{equation}
    \bar{\rho}(\vb*{\theta}) = \frac{1}{2^k}\qty(\bar{I} + \sum_{i=1}^{4^k-1}\theta_i\bar{P}_i),
\end{equation}
parameterized by an (unnormalized) generalized Bloch vector $\vb*{\theta} = (\theta_1,\ldots, \theta_{4^k-1})$~\cite{kimura2003bloch}, and estimating the expectation value of a logical Pauli operator $\bar{P}_i$.
Here, operators with $\bar{\cdot}$ represent operators of a $k$-qubit logical system: $\{\bar{P}_i\}_{i=0}^{4^k-1}$ denotes the set of $k$-qubit Pauli operators, with $\bar{P}_0=\bar{I}$ the $k$-qubit identity operator.
Since $\mathrm{tr}[\bar{\rho}(\vb*{\theta})\bar{P}_i]=\theta_i$, this task is equivalent to estimating the parameter $\theta_i$.
For simplicity, we restrict to unbiased estimation of $\theta_i$, and denote the unbiased estimator of $\theta_i$ as $\theta_i^{\mathrm{est}}$.

When a quantum device is affected by noise, we encode this logical state into an $n$-qubit physical state using an $[[n,k,d]]$ stabilizer code.
After the physical noise affects the physical state, we perform syndrome measurements and then decode the noisy physical state, resulting in a noisy logical state
\begin{equation}
    \overline{\mathcal{N}}(\bar{\rho}(\vb*{\theta})) = \frac{1}{2^k}\qty(\bar{I} + \sum_{i=1}^{4^k-1}(1-2\epsilon_i)\theta_i\bar{P}_i).
\end{equation}
Here, we assume that the residual logical noise is Pauli noise, and $\epsilon_i$ is the logical error rate defined as the probability that a logical Pauli error anti-commuting with $\bar{P}_i$ occurs.
To estimate $\theta_i$, we perform a logical measurement of $\bar{P}_i$ on $\overline{\mathcal{N}}(\bar{\rho}(\vb*{\theta}))$ and classically post-process the measurement outcomes to mitigate the effect of logical errors.

In such an estimation protocol, one performs identical operations on the decoded noisy logical state regardless of the observed error syndrome.
In other words, the protocol does not use syndrome information, and we call it a \textit{syndrome-agnostic estimation protocol} (Fig.~\ref{fig_syndrome_aware_protocols}(a)).
Meanwhile, recent works have pointed out that using error-syndrome information at the logical-layer estimation stage can further boost performance, for example by post-selecting runs with low estimated logical error probability~\cite{bombin2024fault, meister2024efficient, smith2024mitigating, gidney2025yoked, lee2025efficient, xie2026simple, kishi2026even} or by performing classical post-processing conditioned on the observed syndrome~\cite{bluvstein2024logical, zhou2025error, dincua2025error, aharonov2025syndrome}.
We call such protocols \textit{syndrome-aware estimation protocols}.
By incorporating syndrome information at the estimation stage, one can exploit information from the physical layer that would otherwise be discarded after decoding to improve the estimation of logical observables.
Indeed, in Appendix~\ref{sec_physical_layer}, we show that estimation performed directly on the noisy physical state is equivalent to a syndrome-aware estimation protocol, which justifies the schematic in Fig.~\ref{fig_decoding_estimation}.

To model syndrome-aware estimation protocols, we represent the decoder output as the observed error syndrome $s$ together with a syndrome-conditional noisy logical state (Fig.~\ref{fig_syndrome_aware_protocols}(b),(c)).
Concretely, we model the decoder output as a classical--quantum state
\begin{equation}
    \label{eq_classical_quantum_state}
    \sum_{s\in\mathcal{S}} p_s \ketbra{s} \otimes \overline{\mathcal{N}}_s(\bar{\rho}(\vb*{\theta})).
\end{equation}
Here, $s$ labels the syndrome, $\mathcal{S}$ is the set of possible error syndromes, $p_s$ is the probability of obtaining syndrome $s$, and $\overline{\mathcal{N}}_s$ is the logical noise channel conditioned on $s$, defined as
\begin{equation}
    \overline{\mathcal{N}}_s(\bar{\rho}(\vb*{\theta})) = \frac{1}{2^k}\qty(\bar{I} + \sum_{i=1}^{4^k-1}(1-2\epsilon_{i,s})\theta_i\bar{P}_i).
\end{equation}
We assume that each syndrome-conditioned logical noise channel $\overline{\mathcal{N}}_s$ is also Pauli, and $\epsilon_{i,s}$ denotes the probability that a logical Pauli error anti-commuting with $\bar{P}_i$ occurs when conditioned on the error syndrome $s$.

Under this definition, a syndrome-aware estimation protocol is an estimation procedure performed on the classical--quantum state $\sum_{s} p_s \ketbra{s}\otimes \overline{\mathcal{N}}_s(\bar{\rho}(\vb*{\theta}))$.
In contrast, a syndrome-agnostic estimation protocol is an estimation procedure performed on the averaged noisy logical state $\overline{\mathcal{N}}(\bar{\rho}(\vb*{\theta}))$, where the syndrome information is traced out:
\begin{equation}
    \label{eq_conditional_logical_error}
    \overline{\mathcal{N}} = \sum_{s\in\mathcal{S}} p_s \overline{\mathcal{N}}_s, \qquad
    \epsilon_i = \sum_{s\in\mathcal{S}} p_s \epsilon_{i,s}.
\end{equation}
Therefore, to quantify the information gain from utilizing error-syndrome information in estimation, it suffices to compare the information contained in the classical--quantum state $\sum_{s} p_s \ketbra{s}\otimes \overline{\mathcal{N}}_s(\bar{\rho}(\vb*{\theta}))$ with that in the averaged state $\overline{\mathcal{N}}(\bar{\rho}(\vb*{\theta}))$.

In the following, we consider two operational regimes for syndrome-aware protocols: (i) \emph{classical} syndrome-aware protocols, where a fixed logical measurement of $\bar{P}_i$ is applied to $\overline{\mathcal{N}}_s(\bar{\rho}(\vb*{\theta}))$ for all $s$ regardless of the observed syndrome and the syndrome information is used only in classical post-processing (Fig.~\ref{fig_syndrome_aware_protocols}(b)); and (ii) \emph{quantum} syndrome-aware protocols, where the logical measurement basis is allowed to depend on the observed syndrome $s$ (Fig.~\ref{fig_syndrome_aware_protocols}(c)).
For the classical regime, we analyze the classical Fisher information of the outcome distribution obtained by measuring the classical--quantum state $\sum_{s} p_s \ketbra{s}\otimes \overline{\mathcal{N}}_s(\bar{\rho}(\vb*{\theta}))$.
For the quantum regime, we analyze the quantum Fisher information of the same classical--quantum state.
From these Fisher informations, we introduce an \emph{effective logical error rate} to quantify the gain from syndrome information.
The effective logical error rate is defined as the logical error rate of a syndrome-agnostic protocol that would yield the same classical/quantum Fisher information as a given classical/quantum syndrome-aware protocol.
Our goal is to characterize how small the effective logical error rate can be, depending on which class of syndrome-aware protocol is allowed.

We conclude this section with several remarks on our assumptions.
First, for the target operator $\bar{P}_i$, we assume $0 \leq \epsilon_{i,s} \leq 1/2$ for all $s\in\mathcal{S}$.
This can be enforced by, for each syndrome $s$, choosing whether a logical Pauli error commuting with $\bar{P}_i$ or anti-commuting with $\bar{P}_i$ is more likely under $s$.
This choice corresponds to using a degenerate quantum maximum-likelihood decoder~\cite{iyer2015hardness, fuentes2021degeneracy, demarti2024decoding} with respect to the logical observable $\bar{P}_i$, and then $\epsilon_i=\sum_s p_s \epsilon_{i,s}$ represents the logical error rate induced by this decoder.
In the following, we refer to the degenerate quantum maximum-likelihood decoder simply as the maximum-likelihood decoder.
Second, we assume a single-round code-capacity noise model, where the ideal physical state is affected by a single layer of physical noise, followed by a single round of noiseless syndrome measurement.
Nevertheless, our framework extends to more general circuit-level noise models on Clifford circuits by interpreting $s\in\mathcal{S}$ as a full syndrome record and $\overline{\mathcal{N}}_s(\bar{\rho}(\vb*{\theta}))$ as the logical state at the end of the circuit, as long as the circuit noise can be modeled as Pauli noise.

\section{Fisher information analysis for syndrome-agnostic protocols}
\label{sec_syndrome_agnostic}
In this section, we analyze the Fisher information for syndrome-agnostic estimation protocols, in which error-syndrome information is not used at the logical-layer estimation stage (Fig.~\ref{fig_syndrome_aware_protocols}(a)).
As discussed in Sec.~\ref{sec_problem_setup}, a syndrome-agnostic protocol performs estimation on the averaged noisy logical state $\overline{\mathcal{N}}(\bar{\rho}(\vb*{\theta}))$.
Equivalently, given $N$ copies of $\overline{\mathcal{N}}(\bar{\rho}(\vb*{\theta}))$, the goal is to construct an unbiased estimator $\theta_i^{\mathrm{est}}$ of the parameter $\theta_i$.
Therefore, the ultimate precision limit and sampling overhead of the syndrome-agnostic protocol can be analyzed using the quantum Fisher information of $\overline{\mathcal{N}}(\bar{\rho}(\vb*{\theta}))$.

The Fisher-information analysis and optimal estimation strategy for this setting have been studied in Ref.~\cite{watanabe2010optimal}, which we review in Appendix~\ref{sec_estimation_theory_multiqubit}.
In particular, the optimal measurement basis (characterized by Eq.~\eqref{eq_optimal_measurement_basis}) is the one that diagonalizes the target Pauli operator $\bar{P}_i$.
Indeed, measuring $\overline{\mathcal{N}}(\bar{\rho}(\vb*{\theta}))$ with $\bar{P}_i$ yields the binary distribution
\begin{equation}
    \mathrm{Pr}^{\mathrm{L}}(x|\theta_i) = \frac{1+x(1-2\epsilon_i)\theta_i}{2},
\end{equation}
where $x=\pm1$ denotes the logical measurement outcome, and the superscript $\mathrm{L}$ indicates that the protocol is performed purely at the logical layer without using the error-syndrome record.
As reviewed in Sec.~\ref{sec_classical_estimation_theory}, an asymptotically efficient estimator is given by the maximum-likelihood estimator in Eq.~\eqref{eq_MLE}.
In this specific binary-outcome setting, however, one can write down a simple optimal unbiased estimator explicitly.
Given outcomes $x_1,\ldots,x_N$ obtained by measuring $N$ copies of $\overline{\mathcal{N}}(\bar{\rho}(\vb*{\theta}))$ with $\bar{P}_i$, an unbiased estimator can be constructed as
\begin{equation}
    \label{eq_estimator_agnostic}
    \theta_i^{\mathrm{est}} = \frac{1}{N}\sum_{j=1}^N \frac{x_j}{(1-2\epsilon_i)},
\end{equation}
i.e., by averaging the rescaled outcomes $x_j/(1-2\epsilon_i)$.
This estimator satisfies $\mathbb{E}[\theta_i^{\mathrm{est}}]=\theta_i$ and has variance
\begin{equation}
    \mathrm{Var}[\theta_i^{\mathrm{est}}]
    = \frac{1}{N}\frac{1-(1-2\epsilon_i)^2\theta_i^2}{(1-2\epsilon_i)^2},
\end{equation}
which attains the (quantum) Cram\'er--Rao bound in Eqs.~\eqref{eq_CCRB},~\eqref{eq_QCRB}.
This is because both the classical Fisher information $F^{\mathrm{L}}_{\theta_i}$ of the probability distribution $\Pr^{\mathrm{L}}(x|\theta_i)$ (defined in Eq.~\eqref{eq_CFI}) and the quantum Fisher information $J^{\mathrm{L}}_{\theta_i}$ of the noisy logical state $\overline{\mathcal{N}}(\bar{\rho}(\vb*{\theta}))$ for $\theta_i$ (defined in Eq.~\eqref{eq_QFI}) coincide and are given by
\begin{equation}
    \label{eq_QFI_agnostic}
    J^{\mathrm{L}}_{\theta_i} = F^{\mathrm{L}}_{\theta_i}
    = f_{\theta_i}(\epsilon_i)
    = f_{\theta_i}\qty(\sum_{s\in\mathcal{S}} p_s\epsilon_{i,s}).
\end{equation}
Here,
\begin{equation}
    \label{eq_def_f}
    f_{\theta_i}(\epsilon)
    = \frac{(1-2\epsilon)^2}{1-(1-2\epsilon)^2\theta_i^2}
\end{equation}
is a convex function of $\epsilon$ on $[0,1/2]$.
The equality $J^{\mathrm{L}}_{\theta_i}=F^{\mathrm{L}}_{\theta_i}$ implies that measuring $\bar{P}_i$ is optimal in this setting.
We provide a detailed derivation of the Fisher information in Appendix~\ref{sec_estimation_theory_multiqubit}.

In the following sections, we use the Fisher information of the syndrome-agnostic protocol, $J^{\mathrm{L}}_{\theta_i}=F^{\mathrm{L}}_{\theta_i}=f_{\theta_i}(\epsilon_i)$, as the baseline for comparison.
In particular, for a classical (quantum) syndrome-aware protocol, we define its effective logical error rate $\epsilon_i^{\mathrm{cSynd}}$ ($\epsilon_i^{\mathrm{qSynd}}$) as the logical error rate such that $f_{\theta_i}(\epsilon_i^{\mathrm{cSynd}})$ ($f_{\theta_i}(\epsilon_i^{\mathrm{qSynd}})$) matches the Fisher information achieved by that protocol.
By comparing $\epsilon_i^{\mathrm{cSynd}}$ ($\epsilon_i^{\mathrm{qSynd}}$) with the original logical error rate $\epsilon_i$, we quantify how much the impact of noise can be reduced by incorporating error-syndrome information at the estimation stage.

\section{Fundamental limitations on classical syndrome-aware protocols}
\label{sec_classical}
In this section, we study \emph{classical} syndrome-aware estimation protocols, in which the logical measurement basis is fixed to $\bar{P}_i$ regardless of the observed error syndrome $s$, and syndrome information is used only in classical post-processing (Fig.~\ref{fig_syndrome_aware_protocols}(b)).
This class includes well-known protocols such as post-selection based on an estimated conditional logical error rate~\cite{bombin2024fault, meister2024efficient, smith2024mitigating, gidney2025yoked, lee2025efficient, xie2026simple, kishi2026even}, as well as syndrome-conditioned variants of quantum error mitigation methods~\cite{bluvstein2024logical, zhou2025error, dincua2025error, aharonov2025syndrome} such as rescaling~\cite{tsubouchi2024symmetric}, zero-noise extrapolation~\cite{temme2017error, li2017efficient, kim2023evidence} (since, in our model, the noise strength can be increased classically after measurement), probabilistic error cancellation~\cite{temme2017error, endo2018practical, van2023probabilistic} (since, in our model, noise inversion can be implemented by classically flipping measurement outcomes according to an appropriate quasiprobability distribution), symmetry verification~\cite{bonet2018low, mcardle2019error}, and learning-based methods~\cite{czarnik2021error, strinkis2021learning}.

Our analysis is based on the classical Fisher information of the outcome distribution obtained by measuring the syndrome-conditioned logical state with $\bar{P}_i$, together with the \emph{effective logical error rate}, defined as the logical error rate of a syndrome-agnostic protocol that would yield the same classical Fisher information as a given classical syndrome-aware protocol.
In Sec.~\ref{sec_classical_1}, we derive the classical Fisher information and the corresponding effective logical error rate in our setup.
In Sec.~\ref{sec_classical_2}, we prove a universal lower bound on the effective logical error rate, showing that, on average, it cannot be smaller than one-half of the logical error rate.
In Sec.~\ref{sec_classical_3}, we further analyze the effective logical error rate, with particular emphasis on the low-error regime.
Finally, in Sec.~\ref{sec_classical_4}, we numerically investigate the effective logical error rate for several representative quantum error-correcting codes.

\subsection{Classical Fisher information and effective logical error rate}
\label{sec_classical_1}
As discussed in Sec.~\ref{sec_problem_setup}, a syndrome-aware protocol is an estimation protocol performed on the classical--quantum state $\sum_{s\in\mathcal{S}} p_s \ketbra{s} \otimes \overline{\mathcal{N}}_s(\bar{\rho}(\vb*{\theta}))$.
In particular, for the classical syndrome-aware protocol, the logical measurement applied to each syndrome-conditioned noisy logical state $\overline{\mathcal{N}}_s(\bar{\rho}(\vb*{\theta}))$ is fixed to $\bar{P}_i$, which is the optimal measurement basis for the syndrome-agnostic protocol.
Therefore, a classical syndrome-aware protocol can be modeled as an estimation protocol based on the classical joint distribution
\begin{equation}
    \mathrm{Pr}^{\mathrm{Synd}}(x,s|\theta_i) = p_s \mathrm{Pr}^{\mathrm{L}}(x|s,\theta_i),
\end{equation}
where
\begin{equation}
    \mathrm{Pr}^{\mathrm{L}}(x|s,\theta_i) = \frac{1+x(1-2\epsilon_{i,s})\theta_i}{2}
\end{equation}
is the outcome distribution conditioned on the syndrome $s$.
Equivalently, given $N$ samples $(x_1,s_1),\ldots,(x_N,s_N)$ drawn from $\Pr^{\mathrm{Synd}}(x,s|\theta_i)$, the goal is to construct an unbiased estimator $\theta_i^{\mathrm{est}}$ of $\theta_i$.
Hence, the ultimate precision and sampling overhead of classical syndrome-aware protocols can be analyzed via the classical Fisher information of $\Pr^{\mathrm{Synd}}(x,s|\theta_i)$.

From Eq.~\eqref{eq_CFI_joint} and Eq.~\eqref{eq_QFI_agnostic}, the classical Fisher information $F_{\theta_i}^{\mathrm{Synd}}$ of $\Pr^{\mathrm{Synd}}(x,s|\theta_i)$ is given by
\begin{equation}
    F_{\theta_i}^{\mathrm{Synd}} = \sum_{s\in\mathcal{S}} p_s f_{\theta_i}(\epsilon_{i,s}),
\end{equation}
where the function $f_{\theta_i}(\epsilon)$ is defined in Eq.~\eqref{eq_def_f}.
Since $f_{\theta_i}(\epsilon)$ is convex in $\epsilon$, Jensen's inequality implies
\begin{equation}
    \label{eq_cfi_inequality}
    J_{\theta_i}^{\mathrm{L}} = F_{\theta_i}^{\mathrm{L}} \leq F_{\theta_i}^{\mathrm{Synd}}.
\end{equation}
In other words, the syndrome-aware distribution $\mathrm{Pr}^{\mathrm{Synd}}(x,s|\theta_i)$ contains more information than the averaged logical distribution $\mathrm{Pr}^{\mathrm{L}}(x|\theta_i)=\sum_s \mathrm{Pr}^{\mathrm{Synd}}(x,s|\theta_i)$.
Thus, incorporating error-syndrome information can improve estimation precision (equivalently, reduce the required sampling overhead), compared to syndrome-agnostic estimation.
In other words, the effective impact of logical noise can be reduced by exploiting the syndrome information.

To quantify this improvement, we define the \emph{effective logical error rate} $\epsilon_i^{\mathrm{cSynd}}$ for classical syndrome-aware estimation protocols by
\begin{equation}
    \label{eq_effective_logical_error_rate_classical}
    f_{\theta_i}(\epsilon_i^{\mathrm{cSynd}}) = F_{\theta_i}^{\mathrm{Synd}} = \sum_{s\in\mathcal{S}} p_s f_{\theta_i}(\epsilon_{i,s}).
\end{equation}
Equivalently, $\epsilon_i^{\mathrm{cSynd}}$ is the logical error rate of a syndrome-agnostic protocol that would achieve the same Fisher information as the classical syndrome-aware protocol.
By comparing $\epsilon_i^{\mathrm{cSynd}}$ with the original logical error rate $\epsilon_i$, we quantify how much the impact of noise can be reduced when incorporating error-syndrome information at the estimation stage.

We end this section by commenting on an asymptotically optimal classical syndrome-aware estimator that attains the Cram\'er--Rao bound in Eq.~\eqref{eq_CCRB}.
Given $N$ samples $(x_1,s_1),\ldots,(x_N,s_N)$ from $\Pr^{\mathrm{Synd}}(x,s|\theta_i)$, we first construct a rough initial unbiased estimator $\theta_i^{\mathrm{init}}$ using the initial $N'=\lfloor\sqrt{N}\rfloor$ samples.
Using the remaining $N-N'$ samples, we then define
\begin{equation}
    \label{eq_estimator_classical_syndrome_aware}
    \theta_i^{\mathrm{est}} = \sum_{j=N'+1}^{N}
    \frac{f_{\theta_i^{\mathrm{init}}}(\epsilon_{i,s_j})}{\sum_{j'=N'+1}^N f_{\theta_i^{\mathrm{init}}}(\epsilon_{i,s_{j'}})}
    \frac{x_j}{(1-2\epsilon_{i,s_j})}.
\end{equation}
This estimator can be regarded as a weighted average of the syndrome-agnostic estimator in Eq.~\eqref{eq_estimator_agnostic}, where the weights are proportional to the conditional Fisher informations $f_{\theta_i^{\mathrm{init}}}(\epsilon_{i,s_j})$.
In the asymptotic regime $N\to\infty$, we have $\theta_i^{\mathrm{init}}\to\theta_i$ and $f_{\theta_i^{\mathrm{init}}}\to f_{\theta_i}$, and moreover $\sum_{j=N'+1}^N f_{\theta_i}(\epsilon_{i,s_j})\to N F_{\theta_i}^{\mathrm{Synd}}$.
Therefore, $\mathrm{Var}(\theta_i^{\mathrm{est}})\to (N F_{\theta_i}^{\mathrm{Synd}})^{-1}$, and the estimator attains the Cram\'er--Rao bound in Eq.~\eqref{eq_CCRB}.
We note that a similar unbiased estimator was considered in Ref.~\cite{aharonov2025syndrome}, although its optimality was not discussed there.

\subsection{Lower bounds on the effective logical error rate}
\label{sec_classical_2}
As defined in the previous subsection, the effective logical error rate $\epsilon_i^{\mathrm{cSynd}}$ quantifies the residual effect of errors inherent to classical syndrome-aware protocols.
By analyzing how much $\epsilon_i^{\mathrm{cSynd}}$ can be reduced relative to the original logical error rate $\epsilon_i$, we can assess how much incorporating syndrome information can mitigate the impact of noise.
However, as shown in the following theorem, $\epsilon_i^{\mathrm{cSynd}}$ cannot be made arbitrarily small: it can improve over $\epsilon_i$ only by a constant factor.
\begin{thm}
    \label{thm_1}
    The effective logical error rate $\epsilon_i^{\mathrm{cSynd}}$ of classical syndrome-aware protocols satisfies
    \begin{equation}
        \frac{1-\theta_i^2}{2}\epsilon_i \leq \epsilon_i^{\mathrm{cSynd}} \leq \epsilon_i,
    \end{equation}
    where $\epsilon_i$ is the logical error rate under the maximum-likelihood decoder and $\theta_i = \mathrm{tr}[\bar{P}_i\bar{\rho}(\vb*{\theta})]$ is the expectation value to be estimated.
    In particular, when the ideal quantum state $\bar{\rho}(\vb*{\theta})=\ketbra{\psi}$ is drawn from the $k$-qubit Haar measure, the average effective logical error rate satisfies
    \begin{equation}
        \frac{1}{2}\epsilon_i \lesssim \mathbb{E}_{\mathrm{Haar}}[\epsilon_i^{\mathrm{cSynd}}] \leq \epsilon_i,
    \end{equation}
    where $\lesssim$ indicates that the inequality holds up to the leading order in $k$.
\end{thm}

We provide a proof of Theorem~\ref{thm_1} in Appendix~\ref{sec_proof_classical_1}, which relies only on the convexity of the Fisher information function $f_{\theta_i}(\epsilon)$.
A key point is that the bound is derived without assuming any specific code family or noise model.
Therefore, within our setting of logical Pauli observable estimation, Theorem~\ref{thm_1} is universal: $\epsilon_i^{\mathrm{cSynd}}$ cannot be smaller than one half of the logical error rate $\epsilon_i$ achieved by the maximum-likelihood decoder on average, regardless of the code or the underlying noise model.
We note that the assumption that we use the maximum-likelihood decoder is necessary to ensure $0\leq\epsilon_{i,s}\leq1/2$, so that the function $f_{\theta_i}(\epsilon_{i,s})$ is invertible on the domain $0\leq\epsilon_{i,s}\leq1/2$.

Several recent works report that post-selecting syndromes with low estimated error probability can reduce the logical error rate by orders of magnitude, lower the apparent threshold, and even allow a smaller code distance~\cite{bombin2024fault, meister2024efficient, smith2024mitigating, gidney2025yoked, lee2025efficient, xie2026simple, kishi2026even}.
At first glance, such observations might seem to contradict Theorem~\ref{thm_1}.
However, these analyses focus on the \emph{post-selected} logical error rate $\sum_{s\in\mathcal{S}^{\mathrm{Post}}}p_s\epsilon_{i,s}/\sum_{s\in\mathcal{S}^{\mathrm{Post}}}p_s$, where $\mathcal{S}^{\mathrm{Post}}\subset\mathcal{S}$ is the set of error syndromes being post-selected, and typically do not account for the increased sampling overhead incurred by post-selection.
Our result can be interpreted as follows: once the sampling overhead induced by post-selection is properly included, the net reduction of logical errors (as quantified by the effective logical error rate) is bounded by a constant factor, at most $1/2$.
Consequently, quantities such as the code distance or the threshold cannot be improved in an asymptotic sense solely by syndrome-dependent classical post-processing at the estimation stage.
We formalize this intuition in the following corollary, whose proof is provided in Appendix~\ref{sec_proof_classical_2}.
\begin{cor}
    \label{cor_1}
    Assume that $\theta_i\neq \pm1$.
    Assume further that the logical error rate under the maximum-likelihood decoder scales as $\epsilon_i=\Theta(\eta^{l})$ for some constant $l$, where $\eta$ is a physical error rate.
    Then, the effective logical error rate of classical syndrome-aware protocols also scales as $\epsilon_i^{\mathrm{cSynd}}=\Theta(\eta^{l})$, i.e., the effective code distance does not change.
    Moreover, the threshold defined in terms of $\epsilon_i^{\mathrm{cSynd}}$ is identical to that defined in terms of the logical error rate $\epsilon_i$ under the maximum-likelihood decoder.
\end{cor}

Using the lower bound on $\epsilon_i^{\mathrm{cSynd}}$, we can also derive a lower bound on the sampling overhead for classical syndrome-aware estimation protocols.
In quantum error mitigation, the sampling overhead is known to grow exponentially with the total error rate $\epsilon_{\mathrm{tot}}$ of the circuit, typically as $e^{2\epsilon_{\mathrm{tot}}}$~\cite{tsubouchi2023universal, tsubouchi2024symmetric}.
This scaling also applies directly to syndrome-agnostic protocols when these results are applied to logical circuits.
However, it has been unclear how much this overhead can be reduced by incorporating syndrome information.
Theorem~\ref{thm_1} provides a rigorous answer: even when syndrome information is used at the logical-layer estimation stage, the average effective logical error rate can improve by at most a factor of two.
Accordingly, the sampling overhead can improve at best quadratically, reducing the scaling from $e^{2\epsilon_{\mathrm{tot}}}$ to $e^{\epsilon_{\mathrm{tot}}}$, and therefore the exponential overhead cannot be avoided.
We state this formally as the following corollary.
\begin{cor}
    \label{cor_2}
    The minimal sampling overhead for classical syndrome-aware protocols $N^{\mathrm{cSynd}} = (\sigma^2_{\mathrm{target}}F_{\theta_i}^{\mathrm{Synd}})^{-1}$ satisfies
    \begin{equation}
        \frac{(1-\theta_i^2)}{\sigma_{\mathrm{target}}}\sqrt{N^{\mathrm{L}}+\frac{\theta_i^2}{\sigma^2_{\mathrm{target}}}}\leq N^{\mathrm{cSynd}} \leq N^{\mathrm{L}},
    \end{equation}
    where $N^{\mathrm{L}} = (\sigma^2_{\mathrm{target}}F_{\theta_i}^{\mathrm{L}})^{-1}$ is the minimal sampling overhead for syndrome-agnostic protocols and $\sigma_{\mathrm{target}}$ is the target standard deviation.
    In particular, when the ideal logical state is drawn from the $k$-qubit Haar measure, the average sampling overhead satisfies
    \begin{equation}
        \frac{1}{\sigma_{\mathrm{target}}}\sqrt{\mathbb{E}_{\mathrm{Haar}}[N^{\mathrm{L}}]}\lesssim \mathbb{E}_{\mathrm{Haar}}[N^{\mathrm{cSynd}}] \leq \mathbb{E}_{\mathrm{Haar}}[N^{\mathrm{L}}].
    \end{equation}
\end{cor}

Here, the minimal sampling overhead in Corollary~\ref{cor_2} is defined as the sample size $N$ that attains equality in the Cram\'er--Rao bound in Eq.~\eqref{eq_CCRB}.
Corollary~\ref{cor_2} shows that $N^{\mathrm{cSynd}}$ can be at most quadratically smaller than the syndrome-agnostic overhead $N^{\mathrm{L}}$.
Since $N^{\mathrm{L}}$ grows exponentially with the total error rate as $e^{2\epsilon_{\mathrm{tot}}}$~\cite{tsubouchi2023universal, tsubouchi2024symmetric}, the exponential overhead in $N^{\mathrm{cSynd}}=\Omega(e^{\epsilon_{\mathrm{tot}}})$ cannot be avoided.
We provide a proof of Corollary~\ref{cor_2} in Appendix~\ref{sec_proof_classical_3}.

We emphasize that the fundamental limitations established here apply to a broad class of previously proposed classical syndrome-aware strategies, including post-selection based on syndrome records~\cite{bombin2024fault, meister2024efficient, smith2024mitigating, gidney2025yoked, lee2025efficient, xie2026simple, kishi2026even} and syndrome-conditioned variants of logical-layer quantum error mitigation~\cite{bluvstein2024logical, zhou2025error, dincua2025error, aharonov2025syndrome}.
For instance, post-selection based on syndrome records can be viewed as using syndrome information to convert a subset of leading logical faults into detectable erasure events.
When detecting such errors, the sampling overhead typically scales as $e^{p_{\mathrm{tot}}}$ (as in symmetry verification~\cite{bonet2018low, mcardle2019error}), which is quadratically smaller than the generic quantum-error-mitigation overhead $e^{2p_{\mathrm{tot}}}$ required to handle undetectable logical faults~\cite{tsubouchi2023universal, tsubouchi2024symmetric}.
Our Theorem~\ref{thm_1} and Corollary~\ref{cor_2} show that classical syndrome-aware protocols cannot yield a better scaling of the sampling overhead than the one achievable by such error-detection-based strategies.
Our bounds are also consistent with numerical observations for syndrome-conditioned mitigation protocols.
In particular, Ref.~\cite{aharonov2025syndrome} reported that the blowup rate $\lambda$, defined from the sampling overhead $e^{\lambda p_{\mathrm{tot}}}$ of a given quantum error mitigation method (which is closely related to our effective logical error rate), is reduced by a factor no larger than $1/2$, which is precisely explained by the universal limitation in Theorem~\ref{thm_1}.

\subsection{Analysis in the low-error regime}
\label{sec_classical_3}
In the previous subsections, we showed that exploiting syndrome information can reduce the effective impact of logical noise.
Here, we analyze when such an error reduction is possible in the low-error regime.
Specifically, we characterize when the effective logical error rate $\epsilon_i^{\mathrm{cSynd}}$ can be strictly smaller than the logical error rate $\epsilon_i$, i.e., when $\epsilon_i^{\mathrm{cSynd}}<\epsilon_i$, in the limit of small physical noise.

We recall that the logical error rate $\epsilon_i$ can be written as a weighted average of the syndrome-conditioned logical error rates:
\begin{equation}
    \epsilon_i = \sum_{s\in\mathcal{S}} p_s\epsilon_{i,s}.
\end{equation}
Moreover, in the low-error regime, the effective logical error rate admits an expansion in terms of syndrome contributions:
\begin{equation}
    \begin{aligned}
        \epsilon_i^{\mathrm{cSynd}} 
        &\sim\frac{f_{\theta_i}(\epsilon_i^{\mathrm{cSynd}})-f_{\theta_i}(0)}{f_{\theta_i}^{(1)}} \\
        &\sim \sum_{s\in\mathcal{S}} p_s\frac{f_{\theta_i}(\epsilon_{i,s})-f_{\theta_i}(0)}{f_{\theta_i}^{(1)}}.
    \end{aligned}
\end{equation}
Here, $\sim$ means equality up to terms that are $o(\epsilon_i^{\mathrm{cSynd}})$.
To obtain these relations, we expand $f_{\theta_i}(\epsilon)$ around $\epsilon=0$ as $f_{\theta_i}(\epsilon) = f_{\theta_i}(0) + f_{\theta_i}^{(1)}\epsilon + O(\epsilon^2)$, and use the definition of $\epsilon_i^{\mathrm{cSynd}}$ in Eq.~\eqref{eq_effective_logical_error_rate_classical}.
Therefore, to analyze the (effective) logical error rate in the low-error regime, it suffices to understand the contribution of each syndrome individually.

We now identify which syndromes can affect the leading-order behavior.
Consider an $[[n,k,d]]$ stabilizer code under a local noise model in which single-qubit Pauli errors occur independently on each physical qubit with probability $\eta$.
We assume that the minimum weight of a logical Pauli operator that anti-commutes with the target logical Pauli operator $\bar{P}_i$ is also $d$, so that $\epsilon_i = \Theta(\eta^{\lfloor (d+1)/2\rfloor})$.
Suppose a syndrome $s$ occurs with probability $p_s=\Theta(\eta^{l})$.
Then, by the definition of code distance, the conditional logical error rate scales as $\epsilon_{i,s}=O(\eta^{\max\{d-2l,0\}})$, and hence $p_s\epsilon_{i,s} = O(\eta^{\max\{d-l,l\}})$.
On the other hand, the logical error rate and the effective logical error rate scale as $\epsilon_i,\ \epsilon_i^{\mathrm{cSynd}} = \Theta(\eta^{\lfloor (d+1)/2\rfloor})$ in the low-error regime.
Therefore, syndromes with $p_s\epsilon_{i,s}=o(\eta^{\lfloor (d+1)/2\rfloor})$ do not contribute to the leading-order term of $\epsilon_i$ or $\epsilon_i^{\mathrm{cSynd}}$.
Consequently, the dominant syndromes are those satisfying
\begin{equation}
    p_s\epsilon_{i,s} = \Theta(\eta^{\lfloor (d+1)/2\rfloor}).
\end{equation}
In other words, $\max\{d-l,l\}=\lfloor (d+1)/2\rfloor$, so we have $l=d/2$ when $d$ is even and $l=(d-1)/2, (d+1)/2$ when $d$ is odd.
In the following, we analyze these dominant contributions separately depending on the parity of the code distance $d$.

\subsubsection{When the code distance $d$ is even}
When the code distance $d$ is even, the condition
$p_s\epsilon_{i,s} = \Theta(\eta^{\lfloor (d+1)/2\rfloor}) = \Theta(\eta^{d/2})$
is satisfied by syndromes for which $l=d/2$, i.e., $p_s = \Theta(\eta^{d/2}) = \Theta(\epsilon_i)$ and $\epsilon_{i,s}=\Theta(1)$.
Since the syndrome-conditioned logical error rate remains $\Theta(1)$ for such syndromes, we denote the set of these syndromes by $\mathcal{S}_{\Theta(1)}\subset\mathcal{S}$.
Then,
\begin{equation}
    \label{eq_eff_err_rate_classical_even}
    \begin{aligned}
        \epsilon_i &\sim \sum_{s\in\mathcal{S}_{\Theta(1)}} p_s\epsilon_{i,s}, \\
        \epsilon_i^{\mathrm{cSynd}} &\sim \sum_{s\in\mathcal{S}_{\Theta(1)}}p_s\frac{f_{\theta_i}(\epsilon_{i,s})-f_{\theta_i}(0)}{f_{\theta_i}^{(1)}},
    \end{aligned}
\end{equation}
where $\sim$ means equality up to terms that are $o(\eta^{d/2})$.

For $s\in\mathcal{S}_{\Theta(1)}$, the conditional logical error rate does not vanish as the physical error rate $\eta$ decreases, but instead converges to a constant, i.e., $\epsilon_{i,s}=\Theta(1)$.
Operationally, observing such a syndrome indicates a persistent decoding ambiguity: there exist at least two error configurations of the same weight of $d/2$ that induce different logical actions.
In syndrome-agnostic protocols, the contributions of these ambiguous syndromes $s\in\mathcal{S}_{\Theta(1)}$ are inseparable, because the syndrome record is discarded.
In contrast, syndrome-aware protocols can apply different classical post-processing conditioned on $s\in\mathcal{S}_{\Theta(1)}$, which can improve the effective logical error rate.
Indeed, for any fixed $\epsilon_{i,s}=\Theta(1)>0$, we have $(f_{\theta_i}(\epsilon_{i,s})-f_{\theta_i}(0))/f_{\theta_i}^{(1)} < \epsilon_{i,s}$ from the convexity of the function $f_{\theta_i}(\epsilon)$, so that $\epsilon_i^{\mathrm{cSynd}}$ becomes strictly smaller than $\epsilon_i$ in the low-error regime.
Because such ambiguous syndromes $s\in\mathcal{S}_{\Theta(1)}$ always exist when $d$ is even, $\epsilon_i^{\mathrm{cSynd}}$ is always strictly smaller than $\epsilon_i$ as $\eta\to0$:
\begin{prop}
    \label{prop_1}
    When the code distance $d$ is even, the effective logical error rate $\epsilon_i^{\mathrm{cSynd}}$ is strictly smaller than the logical error rate $\epsilon_i$ under the maximum-likelihood decoder in the limit $\eta\to0$:
    \begin{equation}
        \frac{1-\theta_i^2}{2}\leq\lim_{\eta\to0} \frac{\epsilon_i^{\mathrm{cSynd}}}{\epsilon_i} < 1.
    \end{equation}
    The equality for $(1-\theta_i^2)/2$ is achieved if and only if $\epsilon_{i,s}=1/2$ for all $s\in\mathcal{S}_{\Theta(1)}$.
\end{prop}

We provide a proof of Proposition~\ref{prop_1} in Appendix~\ref{sec_proof_classical_4}.
The strict inequality $\lim_{\eta\to0} \epsilon_i^{\mathrm{cSynd}}/\epsilon_i < 1$ means that, for even-distance codes, syndrome-aware classical post-processing always yields a nontrivial improvement in the small-error regime.
This occurs because ambiguous syndromes $s\in\mathcal{S}_{\Theta(1)}$, for which a finite conditional logical error persists, necessarily appear when $d$ is even.
Moreover, when the decoding ambiguity is maximal, i.e., $\epsilon_{i,s}=1/2$ for all $s\in\mathcal{S}_{\Theta(1)}$, the improvement from using syndrome information is maximized, and the lower bound in Theorem~\ref{thm_1} can be achieved.
An example is the distance-$2$ rotated surface code with parameters $[[4,1,2]]$.

\subsubsection{When the code distance $d$ is odd}
When the code distance $d$ is odd, the syndromes satisfying
$p_s\epsilon_{i,s} = \Theta(\eta^{\lfloor (d+1)/2\rfloor}) = \Theta(\eta^{(d+1)/2})$ can be classified into two classes corresponding to $l=(d+1)/2$ and $l=(d-1)/2$ in $p_s=\Theta(\eta^l)$: those with $p_s = \Theta(\eta^{(d+1)/2}) = \Theta(\epsilon_i)$ and conditional logical error rate $\epsilon_{i,s}=\Theta(1)$, and those with $p_s = \Theta(\eta^{(d-1)/2}) = \Theta(\epsilon_i/\eta)$ and conditional logical error rate $\epsilon_{i,s}=\Theta(\eta)$.
We denote the former set by $\mathcal{S}_{\Theta(1)}\subset\mathcal{S}$ and the latter by $\mathcal{S}_{\Theta(\eta)}\subset\mathcal{S}$.
Unlike the even-distance case, the syndromes $s\in\mathcal{S}_{\Theta(\eta)}$ with $\epsilon_{i,s}=\Theta(\eta)$ also contribute to the leading-order terms of the (effective) logical error rates $\epsilon_i$ and $\epsilon_i^{\mathrm{cSynd}}$.
However, since these syndromes have vanishing conditional logical error rate in the limit $\eta\to0$, they do not lead to an improvement from syndrome awareness.
Indeed, for $s\in\mathcal{S}_{\Theta(\eta)}$ we have $(f_{\theta_i}(\epsilon_{i,s})-f_{\theta_i}(0))/f_{\theta_i}^{(1)} = \epsilon_{i,s}+O(\epsilon_{i,s}^2)$, so their contributions to $\epsilon_i^{\mathrm{cSynd}}$ coincide with those to $\epsilon_i$ at leading order.
Therefore,
\begin{equation}
    \label{eq_eff_err_rate_classical_odd}
    \begin{aligned}
        \epsilon_i &\sim \sum_{s\in\mathcal{S}_{\Theta(\eta)}} p_s\epsilon_{i,s} + \sum_{s\in\mathcal{S}_{\Theta(1)}} p_s\epsilon_{i,s}, \\
        \epsilon_i^{\mathrm{cSynd}} &\sim \sum_{s\in\mathcal{S}_{\Theta(\eta)}} p_s\epsilon_{i,s} + \sum_{s\in\mathcal{S}_{\Theta(1)}}p_s\frac{f_{\theta_i}(\epsilon_{i,s})-f_{\theta_i}(0)}{f_{\theta_i}^{(1)}},
    \end{aligned}
\end{equation}
where $\sim$ denotes equality up to terms that are $o(\eta^{(d+1)/2})$. 
In other words, only the ambiguous syndromes $s\in\mathcal{S}_{\Theta(1)}$ with non-vanishing conditional logical error are responsible for any improvement.
However, such syndromes $s\in\mathcal{S}_{\Theta(1)}$ do not necessarily exist for odd-distance codes.
Therefore, an improvement from using syndrome information at the logical estimation stage is not guaranteed in general.
In fact, we obtain the following proposition.
\begin{prop}
    \label{prop_2}
    When the code distance $d$ is odd, the effective logical error rate $\epsilon_i^{\mathrm{cSynd}}$ is strictly smaller than the logical error rate $\epsilon_i$ under the maximum-likelihood decoder only if $\mathcal{S}_{\Theta(1)}\neq\emptyset$.
    Equivalently,
    \begin{equation}
        \lim_{\eta\to0}\frac{\epsilon_i^{\mathrm{cSynd}}}{\epsilon_i} < 1
    \end{equation}
    holds if and only if $\mathcal{S}_{\Theta(1)}\neq\emptyset$; otherwise, $\lim_{\eta\to0}\epsilon_i^{\mathrm{cSynd}}/\epsilon_i=1$.
\end{prop}

\begin{figure*}[t]
    \begin{center}
        \includegraphics[width=0.9\linewidth]{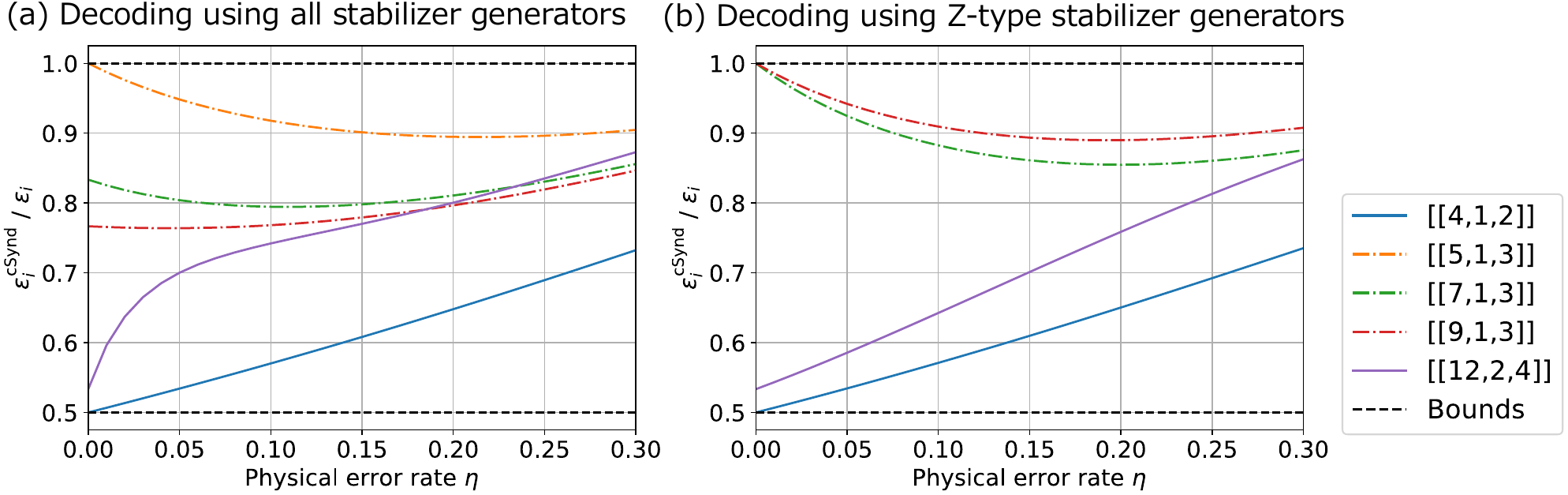}
        \caption{Ratio $\epsilon_i^{\mathrm{cSynd}}/\epsilon_i$ as a function of the physical error rate $\eta$ for several stabilizer codes, where $\epsilon_i$ is the logical error rate under the maximum-likelihood decoder and $\epsilon_i^{\mathrm{cSynd}}$ is the effective logical error rate of classical syndrome-aware protocols. We consider the $[[4,1,2]]$ rotated surface code, $[[5,1,3]]$ perfect code, $[[7,1,3]]$ Steane code, $[[9,1,3]]$ rotated surface code, and $[[12,2,4]]$ carbon code~\cite{paetznick2024demonstration}. Solid (dash-dotted) lines correspond to even-distance (odd-distance) codes, and dashed lines indicate the universal bounds $(1-\theta_i^2)/2\leq \epsilon_i^{\mathrm{cSynd}}/\epsilon_i \leq 1$ from Theorem~\ref{thm_1}, evaluated at $\theta_i=0$. (a) Results obtained using syndrome outcomes from all stabilizer generators. (b) Results obtained when only syndrome outcomes from $Z$-type stabilizer generators are used to correct bit-flip errors. Note that the $[[5,1,3]]$ perfect code is not shown in this figure, since it is not a CSS code.}
        \label{fig_numerics_ratio_smallcode}
    \end{center}
\end{figure*}

We provide the proof of Proposition~\ref{prop_2} in Appendix~\ref{sec_proof_classical_5}.
Proposition~\ref{prop_2} gives a sharp criterion for when classical syndrome awareness yields a nontrivial improvement in the low-error regime $\eta\to0$: $\epsilon_i^{\mathrm{cSynd}}<\epsilon_i$ holds if and only if $\mathcal{S}_{\Theta(1)}\neq\emptyset$.
Most odd-distance quantum codes, including the Steane code and odd-distance rotated surface codes, satisfy this condition.
Meanwhile, the $[[5,1,3]]$ perfect code does not satisfy this condition because all error syndromes occur with $p_s=O(1)$ or $p_s=O(\eta)$, and thus $\epsilon_{i,s}=O(\eta^2)$ or $\epsilon_{i,s}=O(\eta)$.
Therefore, we cannot obtain an advantage for such a code.

Another important example that violates the condition in Proposition~\ref{prop_2} arises when decoding bit-flip errors in CSS codes, especially when an arbitrary $X$-type logical observable has odd weight.
In CSS codes, it is common to decode bit-flip errors using only $Z$-type stabilizer measurements.
Under such a decoding strategy, we have $\mathcal{S}_{\Theta(1)}=\emptyset$ whenever an arbitrary $X$-type logical observable has odd weight.
This is because $\mathcal{S}_{\Theta(1)}\neq\emptyset$ would imply the existence of syndromes occurring with $p_s=\Theta(\eta^{(d+1)/2})$ and having conditional logical error rate $\epsilon_{i,s}=\Theta(1)$.
Equivalently, there must exist two distinct weight-$(d+1)/2$ bit-flip errors that induce different logical actions, which in turn would require the existence of an even-distance $X$-type logical operator.
Examples exhibiting this behavior include the $[[7,1,3]]$ Steane code and odd-distance rotated surface codes.
This means that whether the effective logical error rate $\epsilon_i^{\mathrm{cSynd}}$ is strictly smaller than the original logical error rate $\epsilon_i$ depends on the decoding strategy: even though we can obtain an advantage from using the error syndrome when decoding using both $X$-type and $Z$-type stabilizer measurements, the advantage is lost when we restrict to using only $Z$-type stabilizer measurements.

\subsection{Numerical analysis}
\label{sec_classical_4}
In this subsection, we numerically evaluate the improvement achievable by classical syndrome-aware protocols for several representative quantum error-correcting codes.
Specifically, we compute the ratio of the effective logical error rate $\epsilon_i^{\mathrm{cSynd}}$ to the logical error rate $\epsilon_i$ under the maximum-likelihood decoder.
We consider the $[[4,1,2]]$ rotated surface code, the $[[5,1,3]]$ perfect code, the $[[7,1,3]]$ Steane code, the $[[9,1,3]]$ rotated surface code, and the $[[12,2,4]]$ carbon code~\cite{paetznick2024demonstration}.
As a noise model, we assume local depolarizing noise acting independently on each physical qubit with physical error rate $\eta$.
We take the target observable $\bar{P}_i$ to be the logical Pauli-$Z$ operator on the first logical qubit and set $\theta_i=0$.
For CSS codes, we additionally consider a restricted decoding strategy in which bit-flip errors are decoded using only the syndrome outcomes of $Z$-type stabilizer generators.

The results are shown in Fig.~\ref{fig_numerics_ratio_smallcode}.
For all codes, we confirm that the ratio $\epsilon_i^{\mathrm{cSynd}}/\epsilon_i$ satisfies the universal bounds established in Theorem~\ref{thm_1}.
We also observe a qualitative difference between even- and odd-distance codes, consistent with the discussion in Sec.~\ref{sec_classical_3}.
For even-distance codes, the dominant contribution to the (effective) logical error rate arises from ambiguous syndromes $s\in\mathcal{S}_{\Theta(1)}$ with non-vanishing conditional logical error probabilities, whose impact can be reduced by syndrome-aware classical post-processing at the estimation stage.
Accordingly, the ratio approaches a value close to the lower bound in Theorem~\ref{thm_1}.
Moreover, for the $[[4,1,2]]$ rotated surface code, we have $\epsilon_{i,s}=1/2$ for all $s\in\mathcal{S}_{\Theta(1)}$, and thus the lower bound is attained, as discussed in Proposition~\ref{prop_1}.

By contrast, for odd-distance codes, syndromes $s\in\mathcal{S}_{\Theta(\eta)}$ with vanishing conditional logical error probabilities can also contribute at leading order to $\epsilon_i$ and $\epsilon_i^{\mathrm{cSynd}}$, and their impact cannot be reduced by syndrome-aware classical post-processing.
As a result, the ratio $\epsilon_i^{\mathrm{cSynd}}/\epsilon_i$ for odd-distance codes is typically larger than that for even-distance codes.
In particular, for the $[[5,1,3]]$ perfect code, we have $\mathcal{S}_{\Theta(1)}=\emptyset$ because the code is non-degenerate, and correspondingly we observe $\epsilon_i^{\mathrm{cSynd}}/\epsilon_i \to 1$, in agreement with Proposition~\ref{prop_2}.
Moreover, under the restricted decoding strategy that uses only $Z$-type stabilizer generators, the ratios also approach unity for the $[[7,1,3]]$ Steane code and the $[[9,1,3]]$ rotated surface code, as discussed in Sec.~\ref{sec_classical_3}.

\begin{figure*}[t]
    \begin{center}
        \includegraphics[width=0.7\linewidth]{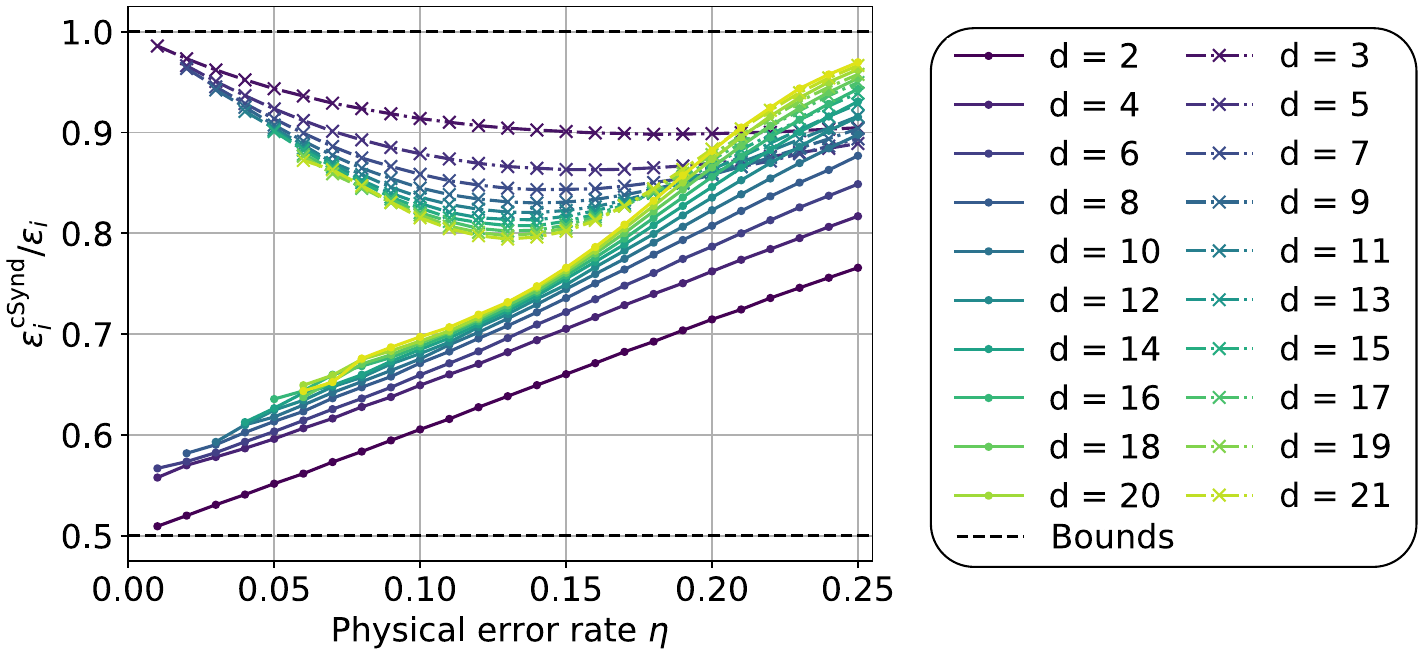}
        \caption{Ratio $\epsilon_i^{\mathrm{cSynd}}/\epsilon_i$ as a function of the physical error rate $\eta$ for rotated surface codes, where $\epsilon_i$ is the logical error rate under the minimum-weight perfect matching decoder and $\epsilon_i^{\mathrm{cSynd}}$ is the effective logical error rate of classical syndrome-aware protocols conditioned on the complementary gap. Solid (dash-dotted) lines correspond to even-distance (odd-distance) codes, and dashed lines indicate the universal bounds $(1-\theta_i^2)/2\leq \epsilon_i^{\mathrm{cSynd}}/\epsilon_i \leq 1$ from Theorem~\ref{thm_1}, evaluated at $\theta_i=0$.}
        \label{fig_numerics_ratio_surfacecode}
    \end{center}
\end{figure*}

So far, we have mainly focused on analyses based on the maximum-likelihood decoder, which is required to ensure the assumption $\epsilon_{i,s}\leq 1/2$ for all $s\in\mathcal{S}$.
However, performing maximum-likelihood decoding---i.e., finding an optimal recovery operation for each syndrome $s\in\mathcal{S}$---is known to be \#P-hard~\cite{iyer2015hardness, fuentes2021degeneracy}.
Moreover, to implement the optimal classical syndrome-aware protocol in Eq.~\eqref{eq_estimator_classical_syndrome_aware}, one needs to learn or calculate the syndrome-conditioned logical error rates $\epsilon_{i,s}$, whose number may grow exponentially with the code size $n$.
For these reasons, the optimal syndrome-agnostic protocols (based on maximum-likelihood decoding) and the optimal classical syndrome-aware protocols (based on learning an exponential number of conditional logical error rates), which form the basis of our theoretical framework, may require excessive resources for actual experiments.

In practice, one instead employs decoders such as minimum-weight perfect matching~\cite{dennis2002topological, higgott2025sparse} or belief-propagation decoders followed by post-processing~\cite{panteleev2021degenerate, roffe2020decoding, muller2025improved, tsubouchi2025degeneracy}.
In such cases, because the decoder may fail to identify the optimal recovery for some syndromes, one can have $\epsilon_{i,s}>1/2$ for certain $s\in\mathcal{S}$.
Furthermore, rather than applying syndrome-dependent post-processing separately for all $s\in\mathcal{S}$, one may coarse-grain the syndrome space by partitioning it into $M$ disjoint subsets as $\mathcal{S}=\bigsqcup_{j=1}^M \mathcal{S}_j$, and perform syndrome-aware estimation based only on which group $\mathcal{S}_j$ contains the observed syndrome $s$.
In this setting, the syndrome-aware protocol is modeled as an estimation procedure on the classical--quantum state
\begin{equation}
    \sum_{j=1}^M p_{\mathcal{S}_j} \ketbra{\mathcal{S}_j} \otimes \overline{\mathcal{N}}_{\mathcal{S}_j}(\bar{\rho}(\vb*{\theta})),
\end{equation}
where $p_{\mathcal{S}_j}=\sum_{s\in\mathcal{S}_j}p_s$ is the probability of observing a syndrome in $\mathcal{S}_j$, $\{\ket{\mathcal{S}_j}\}_{j=1}^M$ is an orthonormal basis of a classical register representing the coarse-grained syndrome outcomes, and $\overline{\mathcal{N}}_{\mathcal{S}_j} = \sum_{s\in\mathcal{S}_j} p_s \overline{\mathcal{N}}_{s}/p_{\mathcal{S}_j}$ is the $\mathcal{S}_j$-conditioned logical noise channel.
For the classical protocol, this corresponds to estimation based on the joint distribution
\begin{equation}
    \label{eq_simplified_distribution}
    \mathrm{Pr}^{\mathrm{Synd}}(x,\mathcal{S}_j|\theta_i) = p_{\mathcal{S}_j} \mathrm{Pr}^{\mathrm{L}}(x|\mathcal{S}_j,\theta_i),
\end{equation}
where
\begin{equation}
    \mathrm{Pr}^{\mathrm{L}}(x|\mathcal{S}_j,\theta_i) = \frac{1+x(1-2\epsilon_{i,\mathcal{S}_j})\theta_i}{2}
\end{equation}
is the outcome distribution conditioned on the group $\mathcal{S}_j$, and
$\epsilon_{i,\mathcal{S}_j} = \sum_{s\in\mathcal{S}_j} p_s\epsilon_{i,s}/p_{\mathcal{S}_j}$ is the corresponding conditional logical error rate.
By choosing the number of groups $M$ to be tractable, one can estimate all $\epsilon_{i,\mathcal{S}_j}$ in advance and perform syndrome-aware estimation with reasonable computational cost.

We now evaluate the effectiveness of such practical syndrome-aware protocols.
We consider the rotated surface code $[[d^2,1,d]]$ for various distances $d$ and decode using minimum-weight perfect matching.
As a noise model, we assume local depolarizing noise acting independently on each physical qubit with physical error rate $\eta$.
We take the target observable $\bar{P}_i$ to be the logical Pauli-$Z$ operator and set $\theta_i=0$.
We decode bit-flip errors using only the syndrome outcomes of $Z$-type stabilizer generators.
We then coarse-grain syndromes according to the complementary gap~\cite{hutter2014efficient, bombin2024fault, gidney2025yoked, smith2024mitigating, meister2024efficient} as $\mathcal{S}=\bigsqcup_{j=1}^M \mathcal{S}_j$, where $\mathcal{S}_j$ denotes the set of syndromes with the same complementary gap, and apply classical post-processing conditioned on this label, which is modeled by Eq.~\eqref{eq_simplified_distribution}.
We use Stim~\cite{gidney2021stim} and PyMatching~\cite{higgott2025sparse} for the simulation.

In Fig.~\ref{fig_numerics_ratio_surfacecode}, we plot the ratio between the logical error rate $\epsilon_i$ under minimum-weight perfect matching and the effective logical error rate $\epsilon_i^{\mathrm{cSynd}}$ defined via Eq.~\eqref{eq_simplified_distribution}.
Because the group-conditioned logical error rates $\epsilon_{i,\mathcal{S}_j}$ can exceed $1/2$, the assumptions underlying Theorem~\ref{thm_1} and Propositions~\ref{prop_1}--\ref{prop_2} do not necessarily hold a priori.
Nevertheless, we observe numerically that the same qualitative conclusions persist in this practical scenario.
In particular, the ratio $\epsilon_i^{\mathrm{cSynd}}/\epsilon_i$ never drops below $1/2$; for even-distance codes, we consistently obtain an improvement $\epsilon_i^{\mathrm{cSynd}}/\epsilon_i<1$ in the low-error regime; and for odd-distance codes, the ratio approaches $1$, indicating no improvement in the low-error regime from using syndrome information at the logical estimation stage.
Therefore, we conclude that our theoretical results are likely to hold even for such practical setups.

\section{Exponential advantages in quantum syndrome-aware protocols}
\label{sec_quantum}
In this section, we turn to \emph{quantum} syndrome-aware estimation protocols, in which syndrome-conditioned logical measurement can be tailored to the observed syndrome (Fig.~\ref{fig_syndrome_aware_protocols}(c)).
That is, after observing the error syndrome $s$, one may choose a syndrome-dependent logical measurement basis before extracting classical data.
This additional adaptivity fundamentally changes the information extractable from logical measurements and allows one to surpass the constant-factor limitations for classical syndrome-aware protocols established in Theorem~\ref{thm_1}.

Our analysis is based on the quantum Fisher information of the classical--quantum state $\sum_s p_s\ketbra{s}\otimes\overline{\mathcal{N}}_s(\bar{\rho}(\vb*{\theta}))$, defined in Eq.~\eqref{eq_classical_quantum_state}.
Analogously to the classical case, we define the \emph{effective logical error rate} $\epsilon_i^{\mathrm{qSynd}}$ of a quantum syndrome-aware protocol as the logical error rate such that the syndrome-agnostic Fisher information $f_{\theta_i}(\epsilon_i^{\mathrm{qSynd}})$ matches the quantum Fisher information of the classical--quantum state.
In Sec.~\ref{sec_quantum_1}, we derive the quantum Fisher information and the corresponding effective logical error rate in our setup.
In Sec.~\ref{sec_quantum_2}, we show that the effective logical error rate can be exponentially smaller than the logical error rate, yielding an exponential separation between classical and quantum syndrome-aware protocols.
In Sec.~\ref{sec_quantum_4}, we numerically investigate the generality of our results.
Finally, in Sec.~\ref{sec_quantum_3}, we discuss an optimal estimation scheme that achieves the exponential advantage.

\subsection{Quantum Fisher information and effective logical error rate}
\label{sec_quantum_1}
As discussed in Sec.~\ref{sec_problem_setup}, a syndrome-aware estimation protocol acts on the classical--quantum state
$\sum_s p_s \ketbra{s} \otimes \overline{\mathcal{N}}_s(\bar{\rho}(\vb*{\theta}))$, defined in Eq.~\eqref{eq_classical_quantum_state}.
In quantum syndrome-aware protocols, unlike the classical case, one is allowed to perform an arbitrary quantum measurement on the classical--quantum state.
Accordingly, a quantum syndrome-aware estimation protocol can be modeled as the task of constructing an unbiased estimator $\theta_i^{\mathrm{est}}$ of $\theta_i$ from $N$ copies of $\sum_s p_s \ketbra{s} \otimes \overline{\mathcal{N}}_s(\bar{\rho}(\vb*{\theta}))$.
Thus, the ultimate precision and sampling overhead of quantum syndrome-aware protocols are characterized by the quantum Fisher information of this classical--quantum state.

Let us define the covariance matrix $C(\vb*{\theta})$ of the ideal logical state $\bar{\rho}(\vb*{\theta})$ by
\begin{equation}
    \label{eq_covariance_matrix}
    C(\vb*{\theta})_{ij}
    = \mathrm{tr}\qty[\bar{\rho}(\vb*{\theta})\frac{1}{2}\{\bar{P}_i,\bar{P}_j\}]
    - \mathrm{tr}\qty[\bar{\rho}(\vb*{\theta})\bar{P}_i]\,
      \mathrm{tr}\qty[\bar{\rho}(\vb*{\theta})\bar{P}_j].
\end{equation}
Using $C(\vb*{\theta})$, the quantum Fisher information matrix associated with the syndrome-conditioned noisy logical state
$\overline{\mathcal{N}}_s(\bar{\rho}(\vb*{\theta}))$ can be written as
$\Lambda_{s} C(\Lambda_s\vb*{\theta})^{-1}\Lambda_{s}$
(see Appendix~\ref{sec_estimation_theory_multiqubit} or Ref.~\cite{watanabe2010optimal} for the derivation).
Here, $\Lambda_s=\mathrm{diag}(1-2\epsilon_{1,s},\ldots,1-2\epsilon_{4^{k}-1,s})$ is the Pauli transfer matrix of $\overline{\mathcal{N}}_s$ defined as a $(4^{k}-1)$-dimensional diagonal matrix describing the action of $\overline{\mathcal{N}}_s$ on the generalized Bloch vector, i.e.,
$\overline{\mathcal{N}}_s(\bar{\rho}(\vb*{\theta}))=\bar{\rho}(\Lambda_s\vb*{\theta})$.
Therefore, from Eq.~\eqref{eq_QFIM_classical_quantum}, the quantum Fisher information matrix for the classical--quantum state
$\sum_s p_s \ketbra{s} \otimes \overline{\mathcal{N}}_s(\bar{\rho}(\vb*{\theta}))$ is
\begin{equation}
    \label{eq_qfim_syndrome}
    J^{\mathrm{Synd}}
    = \sum_{s\in\mathcal{S}} p_s \Lambda_{s} C(\Lambda_s\vb*{\theta})^{-1}\Lambda_{s}.
\end{equation}
In particular, the quantum Fisher information for the parameter $\theta_i$ is
\begin{equation}
    \label{eq_qfim_syndrome_2}
    J_{\theta_i}^{\mathrm{Synd}}
    = \frac{1}{\vb*{e}_i^{\mathrm{T}}\bigl(\sum_{s\in\mathcal{S}} p_s \Lambda_{s} C(\Lambda_s\vb*{\theta})^{-1}\Lambda_{s}\bigr)^{-1}\vb*{e}_i},
\end{equation}
where $\vb*{e}_i\in\mathbb{R}^{4^k-1}$ is a unit vector whose $j$-th element is $1$ if $j=i$ and $0$ otherwise.
From the monotonicity of the quantum Fisher information matrix (equivalently, since the quantum Fisher information is the maximum Fisher information over all measurements, and restricting the measurement basis cannot increase the Fisher information), we obtain
\begin{equation}
    J_{\theta_i}^{\mathrm{L}} = F_{\theta_i}^{\mathrm{L}}
    \leq F_{\theta_i}^{\mathrm{Synd}}
    \leq J_{\theta_i}^{\mathrm{Synd}}.
\end{equation}
This shows that allowing syndrome-conditioned quantum measurement can improve estimation precision (equivalently, reduce the required sampling overhead) compared to syndrome-agnostic and classical syndrome-aware protocols.
In other words, the effective impact of logical noise can be reduced by exploiting syndrome-conditioned quantum measurement.

To quantify this improvement, we define the \emph{effective logical error rate} $\epsilon_i^{\mathrm{qSynd}}$ for quantum syndrome-aware protocols by
\begin{equation}
    \label{eq_quantum_effective_error_rate}
    f_{\theta_i}(\epsilon_i^{\mathrm{qSynd}}) = J_{\theta_i}^{\mathrm{Synd}}.
\end{equation}
Equivalently, $\epsilon_i^{\mathrm{qSynd}}$ is the logical error rate of a syndrome-agnostic protocol that would achieve the same quantum Fisher information as the quantum syndrome-aware protocol.
By comparing $\epsilon_i^{\mathrm{qSynd}}$ with the original logical error rate $\epsilon_i$, we quantify how much the impact of noise can be reduced by allowing syndrome-conditioned quantum measurement.

\subsection{Exponential reduction in the effective logical error rate}
\label{sec_quantum_2}
As defined in the previous subsection, the effective logical error rate $\epsilon_i^{\mathrm{qSynd}}$ quantifies the residual effect of errors inherent to quantum syndrome-aware protocols.
By analyzing how much $\epsilon_i^{\mathrm{qSynd}}$ can be reduced relative to the original logical error rate $\epsilon_i$, we can assess how incorporating syndrome information at the logical measurement stage mitigates the impact of noise.
In the classical case, Theorem~\ref{thm_1} established a universal lower bound showing that the effective logical error rate $\epsilon_i^{\mathrm{cSynd}}$ can improve over $\epsilon_i$ only by a constant factor.
In contrast, once syndrome-conditioned quantum measurement is allowed, such a limitation no longer holds, and $\epsilon_i^{\mathrm{qSynd}}$ can be reduced exponentially with the number of code blocks; this is proven in the low-noise limit for even-distance codes in Theorem~\ref{thm_2}, and is further demonstrated numerically for odd-distance codes at moderate physical error rates in Sec.~\ref{sec_quantum_4}.

For a quantitative analysis, we consider the low-error regime, as in Sec.~\ref{sec_classical_3}.
Consider an $[[n,k,d]]$ stabilizer code  under a local noise model in which single-qubit Pauli errors occur independently on each physical qubit with probability $\eta$.
We assume that the minimum weight of a logical Pauli operator that anti-commutes with the target logical Pauli operator $\bar{P}_i$ is also $d$, so that $\epsilon_i = \Theta(\eta^{\lfloor (d+1)/2\rfloor})$.
In the low-error regime, we can expand the effective logical error rate $\epsilon_i^{\mathrm{qSynd}}$ in terms of syndrome contributions as
\begin{equation}
    \label{eq_qeler_expansion}
    \begin{aligned}
        \epsilon_i^{\mathrm{qSynd}}
        &\sim \frac{f_{\theta_i}(\epsilon_i^{\mathrm{qSynd}})^{-1}-f_{\theta_i}(0)^{-1}}{4} \\
        &=  \frac{\vb*{e}_i^{\mathrm{T}}(\sum_{s} p_s \Lambda_{s} C(\Lambda_s\vb*{\theta})^{-1}\Lambda_{s})^{-1}\vb*{e}_i - \vb*{e}_i^{\mathrm{T}}C(\vb*{\theta})\vb*{e}_i}{4} \\
        &\sim \sum_{s\in\mathcal{S}}p_s \Delta_i(\overline{\mathcal{N}}_s).
    \end{aligned}
\end{equation}
Here, $\sim$ denotes equality up to terms that are $o(\epsilon_i^{\mathrm{qSynd}}) = o(\epsilon_i) = o(\eta^{\lfloor (d+1)/2\rfloor}) $, and
\begin{equation}
\label{eq_delta}
\Delta_i(\overline{\mathcal{N}}_s) = \frac{\vb*{e}_i^{\mathrm{T}}\qty(C(\vb*{\theta}) - C(\vb*{\theta})\Lambda_sC(\Lambda_s\vb*{\theta})^{-1}\Lambda_sC(\vb*{\theta}))\vb*{e}_i}{4}
\end{equation}
represents the contribution from syndrome $s$.
To obtain the first relation in Eq.~\eqref{eq_qeler_expansion}, we use $f_{\theta_i}(\epsilon)^{-1}=(1-2\epsilon)^{-2}-\theta_i^2 = 1-\theta_i^2+4\epsilon+O(\epsilon^2)$.
For the second relation, we use the definition of the effective logical error rate in Eq.~\eqref{eq_quantum_effective_error_rate}, together with $f_{\theta_i}(0)^{-1} = 1-\theta_i^2 = \vb*{e}_i^{\mathrm{T}}C(\vb*{\theta})\vb*{e}_i$.
For the third relation, we use $\sum_s p_s\Lambda_{s} C(\Lambda_s\vb*{\theta})^{-1}\Lambda_{s}-C(\vb*{\theta})^{-1} = O(\epsilon_i)$ and the expansion $(A+B)^{-1} \sim A^{-1} - A^{-1}BA^{-1}$ for square matrices $A$ and $B$ when $A\gg B$.

As in the classical case, the low-error behavior of $\epsilon_i$ and $\epsilon_i^{\mathrm{qSynd}}$ depends on the parity of the distance $d$.
Suppose a syndrome $s$ occurs with probability $p_s=\Theta(\eta^{l})$.
Then, by the definition of code distance, the conditional logical error rate scales as $\epsilon_{i,s}=O(\eta^{\max\{d-2l,0\}})$, and hence $p_s\epsilon_{i,s} = O(\eta^{\max\{d-l,l\}})$.
Note that we also have $p_s\Delta_i(\overline{\mathcal{N}}_s)=O(\eta^{\max\{d-l,l\}})$.
On the other hand, in the low-error regime we have $\epsilon_i,\ \epsilon_i^{\mathrm{qSynd}} = \Theta(\eta^{\lfloor (d+1)/2\rfloor})$.
Therefore, syndromes with $p_s\epsilon_{i,s}=o(\eta^{\lfloor (d+1)/2\rfloor})$ do not contribute to the leading-order term of either $\epsilon_i$ or $\epsilon_i^{\mathrm{qSynd}}$.
Consequently, the dominant syndromes are those satisfying
\begin{equation}
p_s\epsilon_{i,s} = O(\eta^{\max\{d-l,l\}})= \Theta(\eta^{\lfloor (d+1)/2\rfloor}),
\end{equation}
which can be characterized as follows.
\begin{itemize}
\item When $d$ is even: Syndromes $s$ with $l=d/2$, i.e., $p_s = \Theta(\eta^{d/2}) = \Theta(\epsilon_i)$ and $\epsilon_{i,s}=\Theta(1)$ (denoted by $\mathcal{S}_{\Theta(1)}$) constitute the dominant terms.
The leading-order contributions are
\begin{equation}
\label{eq_eff_err_rate_even}
\begin{aligned}
\epsilon_i &\sim \sum_{s\in\mathcal{S}_{\Theta(1)}} p_s\epsilon_{i,s}, \\
\epsilon_i^{\mathrm{qSynd}} &\sim \sum_{s\in\mathcal{S}_{\Theta(1)}}p_s\Delta_i(\overline{\mathcal{N}}_s),
\end{aligned}
\end{equation}
up to terms of order $o(\eta^{\lfloor (d+1)/2\rfloor})$.
\item When $d$ is odd: There are two types of dominant syndromes: (i) $s$ with $l=(d+1)/2$, i.e., $p_s = \Theta(\eta^{(d+1)/2}) = \Theta(\epsilon_i)$ and $\epsilon_{i,s}=\Theta(1)$ (again denoted by $\mathcal{S}_{\Theta(1)}$), and (ii) $s$ with $l=(d-1)/2$, i.e., $p_s = \Theta(\eta^{(d-1)/2}) = \Theta(\epsilon_i/\eta)$ and $\epsilon_{i,s}=\Theta(\eta)$ (denoted by $\mathcal{S}_{\Theta(\eta)}$).
The leading-order contributions are
\begin{equation}
\label{eq_eff_err_rate_odd}
\begin{aligned}
\epsilon_i &\sim \sum_{s\in\mathcal{S}_{\Theta(\eta)}} p_s\epsilon_{i,s} + \sum_{s\in\mathcal{S}_{\Theta(1)}} p_s\epsilon_{i,s}, \\
\epsilon_i^{\mathrm{qSynd}} &\sim \sum_{s\in\mathcal{S}_{\Theta(\eta)}} p_s\epsilon_{i,s} + \sum_{s\in\mathcal{S}_{\Theta(1)}}p_s\Delta_i(\overline{\mathcal{N}}_s).
\end{aligned}
\end{equation}
Here, we use $p_s\Delta_i(\overline{\mathcal{N}}_s) = p_s(\epsilon_{i,s} + O(\epsilon_{i,s}^2)) = p_s\epsilon_{i,s} + o(\eta^{(d+1)/2})$ for $s\in\mathcal{S}_{\Theta(\eta)}$.
\end{itemize}

As in the classical case, we see that error syndromes with non-vanishing conditional logical error rates, $s\in\mathcal{S}_{\Theta(1)}$, are responsible for the improvement in syndrome-aware protocols.
In particular, when the conditional noise channel $\overline{\mathcal{N}}_s$ affects only a small subset of the logical qubits, the contribution $\Delta_i(\overline{\mathcal{N}}_s)$ of the corresponding syndrome $s\in\mathcal{S}_{\Theta(1)}$ to the effective logical error rate $\epsilon_i^{\mathrm{qSynd}}$ becomes exponentially small:
\begin{lem}
    \label{lem_1}
    Let $s$ be an error syndrome of an $[[n,k,d]]$ stabilizer code such that the conditional logical noise channel $\overline{\mathcal{N}}_s$ only affects $k' < k/2$ logical qubits.
    Then, when the ideal logical quantum state $\bar{\rho}(\vb*{\theta})=\ketbra{\psi}$ is drawn from the $k$-qubit Haar measure, the average contribution of the syndrome $s$ to the effective logical error rate satisfies
    \begin{equation}
        \mathbb{E}_{\mathrm{Haar}}[\Delta_i(\overline{\mathcal{N}}_s)] \leq \frac{1}{2^{k-2k'+1}}.
    \end{equation}
\end{lem}
Now consider an $[[n=mn',k=mk',d]]$ stabilizer code, which is composed of $m$ code blocks of a fixed $[[n',k',d]]$ stabilizer code whose distance $d$ is even.
In this case, $s\in\mathcal{S}_{\Theta(1)}$ corresponds to a syndrome caused by a weight-$d/2$ physical Pauli error within a single block.
This means that the conditional logical noise channel $\overline{\mathcal{N}}_s$ affects only a single code block, and the remaining $m-1$ code blocks remain unaffected.
From Lemma~\ref{lem_1}, the contribution $\Delta_i(\overline{\mathcal{N}}_s)$ of the corresponding syndrome $s$ decays exponentially.
Therefore, the ratio between the average effective logical error rate $\epsilon^{\mathrm{qSynd}}_i$ and the original logical error rate $\epsilon_i$ decays exponentially with the number of code blocks $m$:
\begin{thm}
    \label{thm_2}
    Consider an $[[n=mn',k=mk',d]]$ stabilizer code, which is composed of $m$ code blocks of a fixed $[[n',k',d]]$ stabilizer code whose distance $d$ is even.
    Then, when the ideal logical quantum state $\bar{\rho}(\vb*{\theta})=\ketbra{\psi}$ is drawn from the $k$-qubit Haar measure, the ratio between the average effective logical error rate $\epsilon^{\mathrm{qSynd}}_i$ and the original logical error rate $\epsilon_i$ decays exponentially with the number of code blocks $m$ as
    \begin{equation}
        \lim_{\eta\to0}\frac{\mathbb{E}_{\mathrm{Haar}}[\epsilon^{\mathrm{qSynd}}_i]}{\epsilon_i} = O\qty(\frac{1}{2^{k'm}})
    \end{equation}
    as a function of $m$.
\end{thm}

We leave the proofs of Lemma~\ref{lem_1} and Theorem~\ref{thm_2} to Appendices~\ref{sec_proof_quantum_3} and~\ref{sec_proof_quantum_4}.
Theorem~\ref{thm_2} shows that, by performing syndrome-dependent logical measurements, the impact of errors on the estimator can be made exponentially small.
The mechanism behind this can be understood from an information-scrambling viewpoint.
Ambiguous syndromes $s\in\mathcal{S}_{\Theta(1)}$, which constitute the dominant contributions to the logical error rates, identify a single code block on which a weight-$d/2$ error occurs, while the remaining $m-1$ code blocks remain intact.
In a highly entangled logical state, the information associated with the damaged block is strongly correlated with the unaffected code blocks.
An optimal quantum measurement constructed adaptively from the syndrome information exploits this correlation and reduces the uncertainty in estimating $\theta_i$ using information available from the unaffected code blocks.
For Haar-random logical states, this scrambling is nearly maximal on average, so the residual irrecoverable part is only a small fraction of the information that remains isolated on the damaged block.
A Page-type estimate then suppresses this leftover contribution by $O(2^{-k'm})$, yielding Theorem~\ref{thm_2}.
We emphasize that Haar randomness is used here as a convenient way to guarantee such strong entanglement on average.
The same mechanism can also reduce the effective logical error rate for non-Haar states, provided that the states are sufficiently entangled and the information associated with the damaged block is strongly correlated with the unaffected code blocks. 
In Sec.~\ref{sec_quantum_4}, we demonstrate that the exponential suppression is indeed observed for random stabilizer states, which form a state 3-design rather than a Haar ensemble~\cite{kueng2015qubit}.

We remark that Theorem~\ref{thm_2} is an asymptotic statement in the low-error limit $\eta\to0$ for even-distance codes.
In this limit, the leading contribution to the logical error rate comes from ambiguous syndromes $s\in\mathcal{S}_{\Theta(1)}$, whose conditional logical error rates remain non-vanishing.
When these syndrome-conditioned logical noise channels act only on a fixed number of logical qubits, Lemma~\ref{lem_1} suppresses their contribution exponentially with the number of code blocks.
Meanwhile, at finite physical error rate $\eta$, higher-order syndromes that are negligible in the strict low-error expansion can also contribute to the effective logical error rate.
Nevertheless, as we will see numerically in Sec.~\ref{sec_quantum_4}, even-distance codes still exhibit an exponential reduction over a broad finite-error regime, until high-order or multi-block error events become dominant.

In contrast to even-distance codes, the low-error-limit argument leading to Theorem~\ref{thm_2} does not directly apply to odd-distance codes.
For odd-distance codes, the dominant contributions to $\epsilon_i^{\mathrm{qSynd}}$ come not only from ambiguous syndromes with non-vanishing conditional logical error rates, $s\in\mathcal{S}_{\Theta(1)}$, but also from syndromes with vanishing conditional logical error rates, $s\in\mathcal{S}_{\Theta(\eta)}$.
While both types of syndromes can correspond to syndrome-conditioned logical noise localized on a fixed number of code blocks, the strength of the residual logical noise is different for $s\in\mathcal{S}_{\Theta(1)}$ and $s\in\mathcal{S}_{\Theta(\eta)}$.
For $s\in\mathcal{S}_{\Theta(1)}$, the conditional logical error rate remains $O(1)$, and the locality mechanism in Lemma~\ref{lem_1} can reduce its contribution to an exponentially small scale in $m$.
For $s\in\mathcal{S}_{\Theta(\eta)}$, on the other hand, the conditional logical error rate satisfies $\epsilon_{i,s}=O(\eta)$, and the syndrome-conditioned logical channel is already close to the identity.
In this case, expanding around the identity channel gives $\Delta_i(\overline{\mathcal{N}}_s) = \epsilon_{i,s} + O(\epsilon_{i,s}^2)$.
Thus, when the low-error limit $\eta\to0$ is taken at fixed $m$, the leading contribution from $\mathcal{S}_{\Theta(\eta)}$ remains at the level of $\sum_{s\in\mathcal{S}_{\Theta(\eta)}}p_s\epsilon_{i,s}$.
In this sense, odd-distance codes are not covered by the same low-error-limit argument as the even-distance block-code families in Theorem~\ref{thm_2}.

This conclusion, however, depends on the order of limits $\eta\to0$ and $m\to \infty$.
At fixed finite physical error rate $\eta$, the contribution of a localized syndrome can be understood as being controlled by two competing scales.
One is the low-error scale $\epsilon_{i,s}$, which appears from the expansion around the identity channel.
The other is the locality-induced scale set by Lemma~\ref{lem_1}, which is exponentially small in the number of code blocks, scaling as $O(2^{-k'm})$ up to constants depending only on the fixed block size.
Roughly speaking, the behavior is therefore governed by
\begin{equation}
    \label{eq_odd_scaling}
    \mathbb{E}_{\mathrm{Haar}}[\Delta_i(\overline{\mathcal{N}}_s)]
    \lesssim
    \min\{\epsilon_{i,s}, C2^{-k'm}\},
\end{equation}
where $C$ is a constant independent of $m$.
If $\eta\to0$ is taken first at fixed $m$, then $\epsilon_{i,s}=O(\eta)$ becomes smaller than the locality-induced scale $O(2^{-k'm})$, so the contribution $\Delta_i(\overline{\mathcal{N}}_s)\sim\epsilon_{i,s}$ for $s\in\mathcal{S}_{\Theta(\eta)}$ remains at leading order.
If instead $\eta$ is fixed and $m$ is increased, the exponential scale $O(2^{-k'm})$ can eventually become smaller than $\epsilon_{i,s}$, and the locality mechanism can dominate, resulting in the exponential suppression of the effective logical error rate.
Consequently, although odd-distance codes are not covered by Theorem~\ref{thm_2} in the same low-error-limit sense as even-distance codes, they may still exhibit exponential suppression at finite $\eta$ as the number of code blocks increases.
We numerically investigate this finite-error behavior in Sec.~\ref{sec_quantum_4}.

\subsection{Numerical analysis}
\label{sec_quantum_4}
\begin{figure*}[t]
    \begin{center}
        \includegraphics[width=0.99\linewidth]{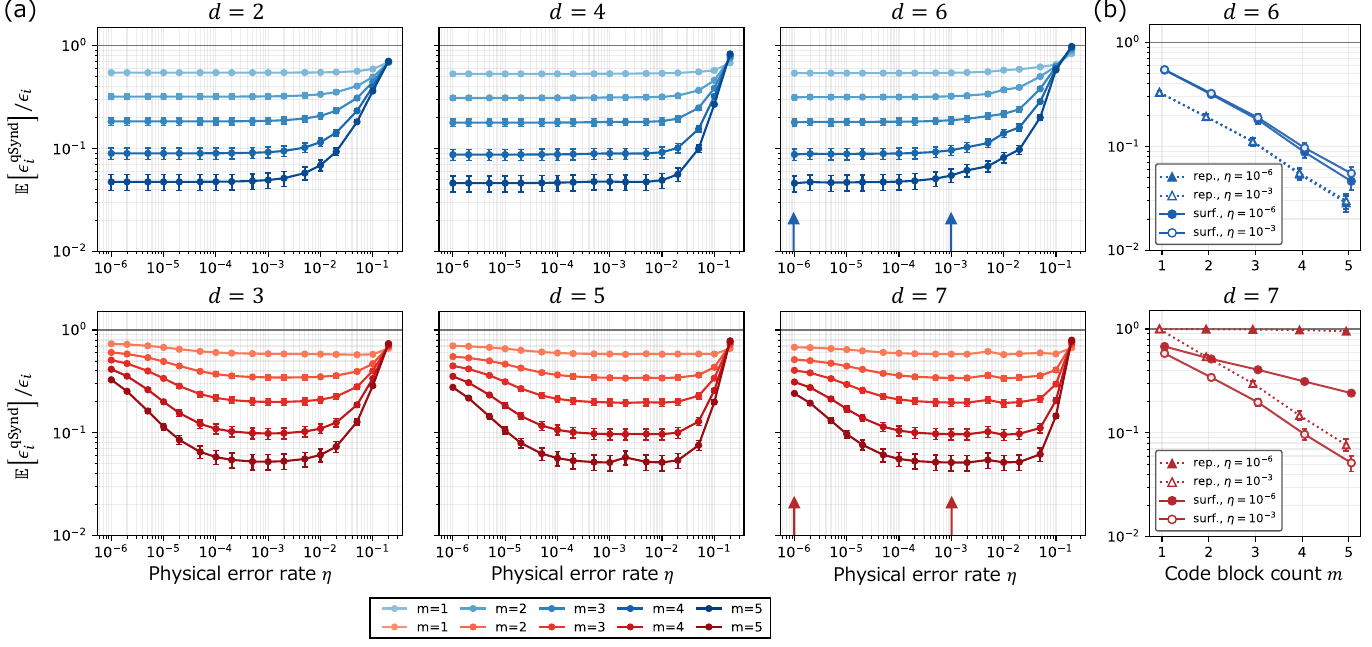}
        \caption{Finite-error-rate ratio $\mathbb{E}[\epsilon^{\mathrm{qSynd}}_i]/\epsilon_i$ as a function of the physical error rate $\eta$ and the number of code blocks $m$. Here, the average is taken over $2000$ random stabilizer states, and the error bars represent 95\% confidence intervals for the estimated ensemble average. Details of the numerical procedure are provided in Appendix~\ref{sec_numerical_details}. (a) Ratio as a function of $\eta$ for $m$ code blocks of fixed rotated surface codes with distances $d=2,\ldots,7$. Blue and red curves correspond to even-distance and odd-distance codes, respectively, and different shades indicate different numbers of code blocks $m=1,\ldots,5$. The arrows indicate the representative physical error rates $\eta=10^{-6}$ and $\eta=10^{-3}$ used in panel (b). (b) Ratio as a function of $m$ for representative distances $d=6$ and $d=7$, comparing repetition code and rotated surface code at $\eta=10^{-6}$ and $\eta=10^{-3}$.}
        \label{fig_numerics_finiteeta}
    \end{center}
\end{figure*}

In this subsection, we numerically investigate the improvement achievable by quantum syndrome-aware protocols in a finite-error-rate regime.
The results are shown in Fig.~\ref{fig_numerics_finiteeta}, and the details of the numerical procedure are provided in Appendix~\ref{sec_numerical_details}.
We consider $m$ code blocks of fixed repetition-code and rotated surface-code families, and take the target logical Pauli operator to be $\bar{P}_i=\bar{Z}\otimes\bar{I}^{\otimes(m-1)}$.
Target logical quantum states are drawn uniformly from random stabilizer states.

Figure~\ref{fig_numerics_finiteeta}(a) shows the ratio $\mathbb{E}[\epsilon_i^{\mathrm{qSynd}}]/\epsilon_i$ as a function of the physical error rate $\eta$ for rotated surface codes with distances $d=2,\ldots,7$.
For even-distance codes, the ratio approaches an $\eta$-independent plateau in the low-error regime, and this plateau decreases exponentially with the number of code blocks $m$ as $O(2^{-m})$.
Meanwhile, at larger physical error rates, the ratio increases and the exponential trend weakens, which is expected because higher-order errors and multi-block error patterns become relevant beyond the strict low-error expansion.
Indeed, beyond the leading ambiguous-syndrome contribution, higher-order contributions of the form $p_s\Delta_i(\overline{\mathcal{N}}_s)=O(\eta^{\lfloor (d+1)/2\rfloor+1})$ can appear.
After normalization by the leading logical error rate $\epsilon_i=\Theta(\eta^{\lfloor (d+1)/2\rfloor})$, these terms give an $O(\eta)$ contribution to the ratio.
Since such higher-order syndrome contributions may involve multi-block error patterns and are not guaranteed to exhibit the same $O(2^{-m})$ suppression as the leading localized contribution, the ratio in this regime can be roughly understood as being controlled by the larger of the two scales:
\begin{equation}
    \frac{\mathbb{E}[\epsilon_i^{\mathrm{qSynd}}]}{\epsilon_i}
    \sim
    \max\{O(2^{-m}),O(\eta)\}.
\end{equation}
This explains why the exponential decay weakens at larger physical error rates, while at experimentally relevant physical error rates around $\eta=10^{-3}$, the $O(2^{-m})$ trend remains clearly visible over the range of code block numbers considered here.

The finite-error results also reveal that odd-distance codes exhibit qualitatively different behavior depending on the physical error rate $\eta$.
At large physical error rates, higher-order and multi-block error events become important.
In this regime, the syndrome-conditioned logical noise is no longer well approximated by noise localized on a fixed number of logical qubits, and the exponential suppression with the number of code blocks becomes weak.
At intermediate values of $\eta$, the situation is more favorable.
Higher-order and multi-block error events are sufficiently suppressed, while the low-error scale $\epsilon_{i,s}$ associated with syndromes in $\mathcal{S}_{\Theta(\eta)}$ remains larger than the locality-induced scale $C2^{-k'm}$ over the range of $m$ considered (see Eq.~\eqref{eq_odd_scaling}).
When the syndrome-conditioned logical noise remains sufficiently localized in this regime, the same locality mechanism underlying Lemma~\ref{lem_1} can dominate as $m$ increases.
Consequently, our numerical results show that odd-distance codes can also exhibit an exponential suppression of $\mathbb{E}[\epsilon_i^{\mathrm{qSynd}}]/\epsilon_i$ with the number of code blocks.
In particular, at $\eta=10^{-3}$, the odd-distance data still show a strong suppression with $m$, suggesting that syndrome information can be useful even for odd-distance codes in experimentally relevant finite-error regimes.

When $\eta$ is made even smaller while $m$ is fixed, the low-error scale $\epsilon_{i,s}$ for the dominant syndromes in $\mathcal{S}_{\Theta(\eta)}$ becomes smaller than the locality-induced scale $C2^{-k'm}$ over the range of $m$ considered in these numerics (see Eq.~\eqref{eq_odd_scaling}).
In this regime, the leading-order relation $\Delta_i(\overline{\mathcal{N}}_s) \sim \epsilon_{i,s}$ remains valid, and the decay with $m$ is weakened compared with the $O(2^{-m})$ behavior observed in the intermediate finite-error regime.
The value of $\eta$ at which this first-order low-error behavior becomes visible depends strongly on the code family.
For example, in the repetition-code data in Fig.~\ref{fig_numerics_finiteeta}(b), the physical error rate $\eta=10^{-6}$ already shows this first-order low-error behavior: the ratio remains nearly independent of $m$.
By contrast, in the rotated surface-code data, the suppression with $m$ remains visible even at $\eta=10^{-6}$, indicating that the system has not yet fully entered the first-order low-error regime.
This suggests that the usefulness of syndrome information at finite physical error rates is not determined solely by the parity of the code distance, but also by the detailed structure of the syndrome-conditioned logical noise.
A detailed theoretical characterization of this code-dependent transition between the intermediate finite-error regime and the first-order low-error regime is left for future work.

\subsection{Optimal estimation protocol}
\label{sec_quantum_3}
In the previous subsections, we saw that the effective logical error rate $\epsilon_i^{\mathrm{qSynd}}$ can be exponentially smaller than the logical error rate $\epsilon_i$.
Here, we investigate the optimal estimation protocol---namely, the optimal syndrome-dependent measurement basis---that achieves this exponential reduction, and discuss the underlying mechanism behind the exponential advantage.
The main idea is to use the correlation between the damaged code block and the unaffected code blocks when an ambiguous syndrome $s\in\mathcal{S}_{\Theta(1)}$ is observed.
We first explain that the correlation between code blocks can be used to improve the estimation precision of the target observable.
Then, we describe the optimal estimation protocol: we first identify the correlated observable in the other blocks using a small number of samples, and then slightly shift the measurement basis toward it so that the estimation precision can be improved by exploiting the correlation.

For simplicity, we consider the toy model:
\begin{equation}
    \label{eq_cq_state_simplified}
    (1-p) \ketbra{0} \otimes \bar{\rho}(\vb*{\theta}) + p \ketbra{1} \otimes \qty( \frac{\bar{I}_A}{2^{k'}}\otimes \mathrm{tr}_A[\bar{\rho}(\vb*{\theta})])
\end{equation}
with
\begin{equation}
    \frac{\bar{I}_A}{2^{k'}}\otimes\mathrm{tr}_A[\bar{\rho}(\vb*{\theta})] = \frac{1}{2^{k}}\qty(\bar{I} + \sum_{j\in\mathcal{J}_B}\theta_j\bar{P}_j) = \bar{\rho}(\Lambda_A\vb*{\theta}).
\end{equation}
Here, $A$ represents the label for the $k' (<k/2)$-qubit subsystem to which the erasure error is applied, and we label the subsystem for the remaining $k-k'$ qubits as $B$.
The index set $\mathcal{J}=\{1,\ldots,4^k-1\}$ for nontrivial $k$-qubit Pauli operators $\bar{P}_j$ is split as $\mathcal{J}=\mathcal{J}_A\sqcup\mathcal{J}_B$, where
\begin{equation}
    \begin{aligned}
        \mathcal{J}_A &= \{j\in\mathcal{J} ~|~ \bar{P}_j = \bar{P}_A \otimes \bar{P}_B, ~ \bar{P}_A \neq \bar{I}_A\}, \\
        \mathcal{J}_B &= \{j\in\mathcal{J} ~|~ \bar{P}_j = \bar{I}_A \otimes \bar{P}_B\}.
    \end{aligned}
\end{equation}
The diagonal matrix $\Lambda_A$ represents the Pauli transfer matrix of the erasure channel satisfying $(\Lambda_A)_{jj} = 0$ for $j\in\mathcal{J}_A$ and $(\Lambda_A)_{jj} = 1$ for $j\in\mathcal{J}_B$.
We assume that the target logical Pauli operator $\bar{P}_i$ satisfies $i\in\mathcal{J}_A$.

The state in Eq.~\eqref{eq_cq_state_simplified} can be regarded as a coarse-grained version of the classical--quantum state arising in the setting of Theorem~\ref{thm_2}.
We coarse-grain all syndromes $s\notin\mathcal{S}_{\Theta(1)}$ into a single label $s=0$, with the corresponding syndrome-conditioned state approximated by $\bar{\rho}(\vb*{\theta})$, by neglecting the vanishing conditional logical error $\epsilon_{i,s}\to0$ as $\eta \to 0$.
We then coarse-grain all syndromes $s\in\mathcal{S}_{\Theta(1)}$ into a single label $s=1$.
Since the conditional noise channel for $s\in\mathcal{S}_{\Theta(1)}$ only acts on $k'$ logical qubits, we consider the ultimate case where the conditional noise channel is an erasure channel on a fixed set of $k'$ logical qubits.
Therefore, the state in Eq.~\eqref{eq_cq_state_simplified} provides a useful idealized model for understanding the situation considered in Theorem~\ref{thm_2}.

When the observed syndrome is $s=0$, we obtain the noiseless logical state $\bar{\rho}(\vb*{\theta})$, and hence can extract its full information, including the target expectation value $\theta_i=\mathrm{tr}[\bar{\rho}(\vb*{\theta})\bar{P}_i]$ of the target logical observable $\bar{P}_i$.
By contrast, when the observed syndrome is $s=1$, we obtain the reduced logical state $\bar{\rho}(\Lambda_A\vb*{\theta})$, from which we cannot extract information about $\theta_i$ with $i\in\mathcal{J}_A$ because information about the subsystem $A$ is completely lost, and the reduced state does not depend on $\theta_i$.
Nevertheless, we can still extract information about parameters $\theta_j$ with $j\in\mathcal{J}_B$ associated with the unaffected subsystem $B$.
Such information about subsystem $B$ can then help reduce the uncertainty in estimating $\theta_i$ from the noiseless branch $\bar{\rho}(\vb*{\theta})$.
In other words, when the error syndrome $s=1$ is observed, we can use $\bar{\rho}(\Lambda_A\vb*{\theta})$ to estimate parameters $\theta_j$ with $j\in\mathcal{J}_B$ associated with the unaffected subsystem $B$, and then treat them as effectively known when analyzing the $s=0$ branch.
This reduces the nuisance-parameter uncertainty and thereby improves the estimation of $\theta_i$ from the noiseless branch $\bar{\rho}(\vb*{\theta})$.

To illustrate this mechanism, we consider the following auxiliary problem.
Suppose we are given $N$ copies of the noiseless logical state $\bar{\rho}(\vb*{\theta})$ (corresponding to $s=0$) and ask how knowledge of parameters $\theta_j$ with $j\in\mathcal{J}_B$ associated with the unaffected subsystem $B$ affects the estimation of $\theta_i$.
When we do not know all parameters $\vb*{\theta}\in\mathbb{R}^{4^k-1}$---that is, when parameters $\theta_j$ with $j\neq i$ are treated as nuisance parameters~\cite{suzuki2020quantum}---the corresponding quantum Fisher information matrix is given by the inverse covariance matrix $C(\vb*{\theta})^{-1}\in\mathbb{R}^{(4^{k}-1)\times(4^{k}-1)}$ defined in Eq.~\eqref{eq_covariance_matrix} (see Appendix~\ref{sec_estimation_theory_multiqubit} or Ref.~\cite{watanabe2010optimal} for the derivation).
In this case, the optimal measurement basis is given by the target Pauli operator $\bar{P}_i$ itself, and the optimal unbiased estimator $\theta_i^{\mathrm{est}}$ is obtained by averaging the measurement outcomes.
Indeed, $\mathbb{E}[\theta_i^{\mathrm{est}}] = \mathrm{tr}[\bar{\rho}(\vb*{\theta})\bar{P}_i] = \theta_i$, and its variance is
$\mathrm{Var}[\theta_i^{\mathrm{est}}] = N^{-1}(\mathrm{tr}[\bar{\rho}(\vb*{\theta})\bar{P}_i^2]-\mathrm{tr}[\bar{\rho}(\vb*{\theta})\bar{P}_i]^2)
= N^{-1}(1-\theta_i^2)=N^{-1}C(\vb*{\theta})_{ii}$,
which attains the quantum Cram\'er--Rao bound in Eq.~\eqref{eq_QCRB}.

Now suppose that we have complete information about all parameters $\theta_j$ with $j\in\mathcal{J}_B$ associated with the unaffected subsystem $B$ (which can, in principle, be estimated from the $s=1$ branch $\bar{\rho}(\Lambda_A\vb*{\theta})$).
In this case, parameters $\theta_j$ satisfying $j\in\mathcal{J}_A$ constitute the target and the nuisance parameters, so the corresponding quantum Fisher information matrix becomes $(C(\vb*{\theta})^{-1})_{AA}\in\mathbb{R}^{\abs{\mathcal{J}_A}\times\abs{\mathcal{J}_A}}$.
Here, $A$ denotes the block corresponding to indices in $\mathcal{J}_{A}$.
Therefore, the variance of the optimal estimator achieving the quantum Cram\'er--Rao bound in Eq.~\eqref{eq_QCRB} is represented as
\begin{equation}
    \label{eq_var_1}
    \mathrm{Var}[\theta_i^{\mathrm{est}}] = \frac{1}{N} \vb*{e}_i^{\mathrm{T}}((C(\vb*{\theta})^{-1})_{AA})^{-1}\vb*{e}_i.
\end{equation}

As we will see in Lemma~\ref{lem_3}, we have
\begin{equation}
    \label{eq_var_2}
    \vb*{e}_i^{\mathrm{T}}((C(\vb*{\theta})^{-1})_{AA})^{-1}\vb*{e}_i 
    = \min_{\vb*{u}_B\in\mathbb{R}^{\abs{\mathcal{J}_B}}}\mathrm{Var}_{\bar{\rho}(\vb*{\theta})}(\bar{P}_i + \vb*{u}_B^{\mathrm{T}}\bar{\vb*{P}}_{B}),
\end{equation}
where the minimum is achieved by choosing
\begin{equation}
    \label{eq_nB_optimal}
    \vb*{u}_B^{\mathrm{T}}=-\vb*{e}_i^{\mathrm{T}}C(\vb*{\theta})_{AB}(C(\vb*{\theta})_{BB})^{-1}.
\end{equation}
Here, $\mathrm{Var}_{\bar{\rho}(\vb*{\theta})}(\cdot)$ represents the variance of the operator $\cdot$ for the state $\bar{\rho}(\vb*{\theta})$, and $\bar{\vb*{P}}_{B}$ denotes the vector of logical Pauli operators with indices in $\mathcal{J}_B$.
When we measure $\bar{\rho}(\vb*{\theta})$ with $\bar{P}_i + \vb*{u}_B^{\mathrm{T}}\bar{\vb*{P}}_{B}$, the expectation value is given as
\begin{equation}
    \theta_i + \vb*{u}_B^{\mathrm{T}}\vb*{\theta}_B,
\end{equation}
where $\vb*{\theta}_B$ denotes the vector of parameters $\theta_j$ with indices in $\mathcal{J}_B$.
Since the values of $\vb*{\theta}_B$ are assumed to be known, we can construct an unbiased estimator $\theta_i^{\mathrm{est}}$ by measuring $N$ copies of $\bar{\rho}(\vb*{\theta})$ with $\bar{P}_i + \vb*{u}_B^{\mathrm{T}}\bar{\vb*{P}}_{B}$ and subtracting $\vb*{u}_B^{\mathrm{T}}\vb*{\theta}_B$ from the empirical average.
Eqs.~\eqref{eq_var_1} and~\eqref{eq_var_2} mean that such an estimator achieves the quantum Cram\'er--Rao bound when $\vb*{u}_B$ is chosen as in Eq.~\eqref{eq_nB_optimal}.
In other words, when the parameters associated with subsystem $B$ are known, we can use this information to reduce the uncertainty associated with subsystem $A$, especially when the two subsystems are strongly correlated.
For later convenience, we denote the corresponding optimal measurement operator as
\begin{equation}
    \bar{P}_i - \bar{O}_i(\vb*{\theta})
\end{equation}
with
\begin{equation}
    \bar{O}_i(\vb*{\theta}) = \vb*{e}_i^{\mathrm{T}}C(\vb*{\theta})_{AB}(C(\vb*{\theta})_{BB})^{-1}\bar{\vb*{P}}_{B}.
\end{equation}

The key point is that the optimal measurement basis changes from $\bar{P}_i$ to $\bar{P}_i-\bar{O}_i(\vb*{\theta})$ once the parameters $\theta_j$ for $j\in\mathcal{J}_B$ are known.
This modification reduces the average uncertainty of the remaining unknown parameters indexed by $\mathcal{J}_A$ by utilizing the correlation between subsystems $A$ and $B$.
In particular, the Haar average of the inverse of the relevant quantum Fisher information matrix decays exponentially, as stated in the following proposition, whose proof is provided in Appendix~\ref{sec_proof_quantum_2}.
\begin{prop}
    \label{prop_3}
    Let the ideal quantum state $\bar{\rho}(\vb*{\theta})=\ketbra{\psi}$ be drawn from the $k$-qubit Haar measure.
    Then, assuming $k'<k/2$, the Haar average of the inverse of the quantum Fisher information matrix $(C(\vb*{\theta})^{-1})_{AA}$ in the setting where the parameters $\theta_j$ for $j\in\mathcal{J}_{B}$ are known and the parameters indexed by $\mathcal{J}_{A}$ are unknown satisfies
    \begin{equation}
        \mathbb{E}_{\mathrm{Haar}}[((C(\vb*{\theta})^{-1})_{AA})^{-1}] \leq \frac{2}{2^{k-2k'}}I.
    \end{equation}
\end{prop}
Therefore, the estimator variance is exponentially reduced from $N^{-1}(1-\theta_i^2)\sim N^{-1}$ to
\begin{equation}
    \label{eq_var_exponential}
    \mathbb{E}_{\mathrm{Haar}}[\mathrm{Var}[\theta_i^{\mathrm{est}}]] \leq \frac{1}{N}\frac{2}{2^{k-2k'}}.
\end{equation}

Let us explain the intuition behind the exponential decay.
For simplicity, we assume that $k'=1$ and the target Pauli operator is represented in the form $\bar{P}_i=\bar{P}_A\otimes\bar{I}_B$.
In this case, the Schmidt decomposition of the ideal quantum state can be written as
\begin{equation}
    \ket{\psi} = \sqrt{\frac{1+r}{2}}\ket{\psi_0}_A\ket{\phi_0}_B + \sqrt{\frac{1-r}{2}}\ket{\psi_1}_A\ket{\phi_1}_B,
\end{equation}
where $r$ quantifies the imbalance of the Schmidt coefficients.
A smaller value of $r$ corresponds to stronger entanglement between subsystems $A$ and $B$.
For a Haar-random state, $r^2\sim O(2^{-k})$, meaning that the two subsystems are nearly maximally entangled on average.
To see how this entanglement reduces the estimator variance, let us define Pauli-like operators in the Schmidt bases $\{\ket{\psi_0}_A,\ket{\psi_1}_A\}$ and $\{\ket{\phi_0}_B,\ket{\phi_1}_B\}$ as
\begin{equation}
    \begin{aligned}
        \bar{X}_A' &= \ketbra{\psi_0}{\psi_1}_A + \ketbra{\psi_1}{\psi_0}_A, \\
        \bar{Y}_A' &= -i\ketbra{\psi_0}{\psi_1}_A + i\ketbra{\psi_1}{\psi_0}_A, \\
        \bar{Z}_A' &= \ketbra{\psi_0}{\psi_0}_A - \ketbra{\psi_1}{\psi_1}_A, \\
        \bar{X}_B' &= \ketbra{\phi_0}{\phi_1}_B + \ketbra{\phi_1}{\phi_0}_B, \\
        \bar{Y}_B' &= i\ketbra{\phi_0}{\phi_1}_B - i\ketbra{\phi_1}{\phi_0}_B, \\
        \bar{Z}_B' &= \ketbra{\phi_0}{\phi_0}_B - \ketbra{\phi_1}{\phi_1}_B.
    \end{aligned}
\end{equation}
Here, the operators on subsystem $B$ are understood to act on the Schmidt support and to be extended by zero on its orthogonal complement.
The variances of these operators satisfy
\begin{equation}
    \begin{aligned}
        \mathrm{Var}_{\bar{\rho}(\vb*{\theta})}\qty(\bar{X}_A'\otimes\bar I_B) = \mathrm{Var}_{\bar{\rho}(\vb*{\theta})}\qty(\bar I_A\otimes\bar{X}_B') &= 1, \\
        \mathrm{Var}_{\bar{\rho}(\vb*{\theta})}\qty(\bar{Y}_A'\otimes\bar I_B) = \mathrm{Var}_{\bar{\rho}(\vb*{\theta})}\qty(\bar I_A\otimes\bar{Y}_B') &= 1, \\
        \mathrm{Var}_{\bar{\rho}(\vb*{\theta})}\qty(\bar{Z}_A'\otimes\bar I_B) = \mathrm{Var}_{\bar{\rho}(\vb*{\theta})}\qty(\bar I_A\otimes\bar{Z}_B') &= 1-r^2.
    \end{aligned}
\end{equation}
Moreover, these operators are strongly correlated:
\begin{equation}
    \begin{aligned}
        \mathrm{Cov}_{\bar{\rho}(\vb*{\theta})}\qty(\bar{X}'_A \otimes \bar{I}_B, \bar{I}_A \otimes \bar{X}_B') &= \sqrt{1-r^2}, \\
        \mathrm{Cov}_{\bar{\rho}(\vb*{\theta})}\qty(\bar{Y}'_A \otimes \bar{I}_B, \bar{I}_A \otimes \bar{Y}_B') &= \sqrt{1-r^2}, \\
        \mathrm{Cov}_{\bar{\rho}(\vb*{\theta})}\qty(\bar{Z}'_A \otimes \bar{I}_B, \bar{I}_A \otimes \bar{Z}_B') &= 1-r^2.
    \end{aligned}
\end{equation}
Because of these correlations, adding suitable operators on subsystem $B$ can greatly reduce the variance of an operator on subsystem $A$:
\begin{equation}
    \begin{aligned}
        \min_x\mathrm{Var}_{\bar{\rho}(\vb*{\theta})}\qty(\bar{X}'_A \otimes \bar{I}_B +x \bar{I}_A \otimes \bar{X}_B') &= r^2 = O(2^{-k}), \\
        \min_x\mathrm{Var}_{\bar{\rho}(\vb*{\theta})}\qty(\bar{Y}'_A \otimes \bar{I}_B +x \bar{I}_A \otimes \bar{Y}_B') &= r^2 = O(2^{-k}), \\
        \min_x\mathrm{Var}_{\bar{\rho}(\vb*{\theta})}\qty(\bar{Z}'_A \otimes \bar{I}_B +x \bar{I}_A \otimes \bar{Z}_B') &= 0.
    \end{aligned}
\end{equation}
Since any single-qubit Pauli operator on subsystem $A$ can be written as $\bar{P}_A=n_x\bar{X}'_A+n_y\bar{Y}'_A+n_z\bar{Z}'_A$ with $n_x^2+n_y^2+n_z^2=1$, and since the corresponding operators on subsystem $B$ can be expanded in the Pauli basis $\bar{\vb*{P}}_B$, we obtain
\begin{equation}
    \min_{\vb*{u}_B\in\mathbb{R}^{\abs{\mathcal{J}_B}}}\mathrm{Var}_{\bar{\rho}(\vb*{\theta})}(\bar{P}_i + \vb*{u}_B^{\mathrm{T}}\bar{\vb*{P}}_{B}) \leq r^2 = O(2^{-k}).
\end{equation}
In summary, the strong correlation between subsystems $A$ and $B$ allows one to reduce the variance of the estimator by adding an appropriate operator on the unaffected subsystem $B$.
For Haar-random states, the two subsystems are nearly maximally entangled on average, so $r^2\sim O(2^{-k})$, which explains the exponential decay.

We now return to the classical--quantum state in Eq.~\eqref{eq_cq_state_simplified}.
As we will discuss in Appendix~\ref{sec_proof_quantum_5}, the inverse of the quantum Fisher information matrix of the classical--quantum state is
\begin{equation}
    (J^{\mathrm{Synd}})^{-1} 
    = C(\vb*{\theta}) + \frac{p}{1-p}
    \begin{pmatrix}
        ((C(\vb*{\theta})^{-1})_{AA})^{-1} & 0 \\
        0 & 0 \\
    \end{pmatrix}.
\end{equation}
As discussed above, $(C(\vb*{\theta})^{-1})_{AA}$ is the quantum Fisher information matrix when the parameters indexed by $\mathcal{J}_B$ (associated with the unaffected subsystem $B$) are known while the parameters indexed by $\mathcal{J}_A$ (associated with the noisy subsystem $A$) remain unknown.
Thus, the increase in the variance induced by the noisy branch is governed by the residual uncertainty in the parameters associated with the noisy subsystem $A$, given the information available from the unaffected subsystem $B$.
Since this uncertainty decays exponentially on average as in Proposition~\ref{prop_3}, we obtain an exponential advantage in quantum syndrome-aware protocols.

Finally, we describe an optimal quantum syndrome-aware estimation protocol for the simplified state in Eq.~\eqref{eq_cq_state_simplified}.
As we will discuss in Appendix~\ref{sec_proof_quantum_5}, the optimal measurement basis defined through Eq.~\eqref{eq_optimal_measurement_basis} for the target parameter $\theta_i$ is 
\begin{equation}
    \ketbra{0}\otimes\frac{1}{1-p}(\bar{P}_i-p\bar{O}_i(\vb*{\theta})) + \ketbra{1}\otimes \bar{O}_i(\vb*{\theta}).
\end{equation}
Thus, an asymptotically optimal quantum syndrome-aware estimation protocol can be constructed as follows, which is based on an adaptive two-step strategy~\cite{yang2019attaining, zhou2020saturating}.
\begin{enumerate}
    \item Using the initial $N'=\lfloor\sqrt{N}\rfloor$ copies, construct a rough estimator $\vb*{\theta}^{\mathrm{init}}$ for all parameters.
    \item For $t=N'+1,\ldots,N$, construct $\Delta_t$ from the remaining $N-N'$ copies as follows:
    \begin{itemize}
        \item If the error syndrome $s=0$ is obtained for the $t$-th copy, measure $\bar{\rho}(\vb*{\theta})$ with the operator $\bar{P}_i-p\bar{O}_i(\vb*{\theta}^{\mathrm{init}})$ and obtain outcome $o_t$. Then set
        $\Delta_t=(o_t-\mathrm{tr}[\bar{\rho}(\vb*{\theta}^{\mathrm{init}})(\bar{P}_i-p\bar{O}_i(\vb*{\theta}^{\mathrm{init}}))])/(1-p)$.
        \item If the error syndrome $s=1$ is obtained for the $t$-th copy, measure $\bar{\rho}(\Lambda_A\vb*{\theta})$ with the operator $\bar{O}_i(\vb*{\theta}^{\mathrm{init}})$ and obtain outcome $o_t$. Then set
        $\Delta_t=o_t-\mathrm{tr}[\bar{\rho}(\vb*{\theta}^{\mathrm{init}})\bar{O}_i(\vb*{\theta}^{\mathrm{init}})]$.
    \end{itemize}
    \item Construct the estimator as $\theta_i^{\mathrm{est}} = \theta_i^{\mathrm{init}} + \frac{1}{N-N'}\sum_{t=N'+1}^N\Delta_t$.
\end{enumerate}

Conditioned on the first-stage data (i.e., fixing $\vb*{\theta}^{\mathrm{init}}$), it is straightforward to verify that the expectation of $\theta_i^{\mathrm{est}}$ over the remaining $N-N'$ copies equals $\theta_i$.
Hence, $\theta_i^{\mathrm{est}}$ constructed above is an unbiased estimator of $\theta_i$.
Moreover, the variance of $\theta_i^{\mathrm{est}}$ is obtained by averaging (over the first-stage data) the conditional variance given fixed $\vb*{\theta}^{\mathrm{init}}$.
When $\vb*{\theta}^{\mathrm{init}}=\vb*{\theta}$, the variance of $\Delta_t$ is $\vb*{e}_i^{\mathrm{T}}(J^{\mathrm{Synd}})^{-1}\vb*{e}_i$.
Therefore, in the asymptotic limit $N\to\infty$ we typically have $\vb*{\theta}^{\mathrm{init}}\to\vb*{\theta}$, resulting in
$\mathrm{Var}[\theta_i^{\mathrm{est}}] \sim N^{-1} \vb*{e}_i^{\mathrm{T}}(J^{\mathrm{Synd}})^{-1}\vb*{e}_i$.
Thus, the above estimator is asymptotically optimal and attains the quantum Cram\'er--Rao bound.

\section{Discussion}
\label{sec_discussion}
In this work, we developed an information-theoretic framework to quantify the utility of error-syndrome information for noisy logical observable estimation. 
We introduced syndrome-agnostic, classical syndrome-aware, and quantum syndrome-aware estimation protocols and compared their achievable precision through Fisher information, summarized via an effective logical error rate. 
For classical syndrome-aware protocols---where the logical measurement basis is fixed and syndrome information is used only in classical post-processing---we proved a universal limitation showing that, on average, the effective logical error rate can improve by at most a factor of two.
In contrast, when syndrome-conditioned logical measurement is permitted, we have demonstrated that the effective logical error rate can become exponentially smaller than the original logical error rate as the number of code blocks increases.
We further clarified the mechanism behind this exponential advantage by analyzing the structure of the optimal syndrome-dependent measurement, and supported the generality of the phenomenon via numerical studies beyond the idealized assumptions used in our analytical constructions.

Our results establish a sharp separation between what can and cannot be gained from error-syndrome information in logical-layer estimation tasks.
Our no-go theorem for classical syndrome-aware protocols shows that syndrome-dependent classical post-processing alone cannot remove the exponential sampling overhead inherent to logical-layer error mitigation, placing fundamental constraints on post-selection and related syndrome-conditioned classical strategies. 
At the same time, our exponential separations identify syndrome-conditioned quantum control---in particular, adapting the logical measurement basis to the observed syndrome record---as the key ingredient that unlocks genuine exponential improvements. 
Conceptually, this highlights error syndromes as a form of side information whose value depends qualitatively on the allowed interface between decoding and logical control, and provides a principled criterion for when syndrome records should be actively exploited rather than discarded. 
Practically, our framework offers guidance for designing early fault-tolerant architectures and experimental protocols: if one seeks genuine gains in logical observable estimation, it is not sufficient to merely condition classical post-processing on syndrome data; one must instead enable syndrome-dependent quantum operations at the logical layer.

Future research could explore several promising directions.
A particularly interesting direction is to develop more practical quantum syndrome-aware estimation protocols.
Although we proposed an asymptotically optimal protocol in Sec.~\ref{sec_quantum_3}, it relies on characterizing the logical state in advance.
Moreover, the required syndrome-dependent measurements are generally non-stabilizer measurements and thus may not be implementable fault-tolerantly without additional resources.
It is therefore important, from a practical standpoint, to construct fault-tolerant quantum syndrome-aware protocols that use only stabilizer operations or other native fault-tolerant primitives and avoid extensive pre-characterization of the logical state.

Another interesting direction is to investigate the utility of syndrome information in quantum metrology.
In error-corrected quantum metrology~\cite{zhou2018achieving, zhou2021asymptotic}, probe systems are encoded into quantum error-correcting codes so that the effect of noise can be suppressed while preserving sensitivity to the signal parameter.
Although the objective in metrology is to estimate a physical signal parameter rather than a logical observable, the performance is again characterized by quantum Fisher information, suggesting a possible connection to our framework.
However, this extension is not a straightforward application of our present state-estimation framework, because metrology is fundamentally a channel-estimation problem: the parameter is imprinted through a signal dynamics rather than being encoded in a fixed quantum state.
Extending our framework from state estimation to channel estimation could provide a way to address this problem and clarify when syndrome-aware strategies improve the precision of error-corrected metrology beyond standard protocols.

We can further explore the utility of quantum syndrome-aware protocols for other tasks by considering alternative information measures.
Here, we studied observable estimation through the quantum Fisher information of a classical--quantum state, where the classical register encodes the observed error syndrome and the quantum register is the corresponding syndrome-conditioned noisy logical state.
Since syndrome-aware protocols can be modeled as operations on such classical--quantum states, the utility of syndrome information for other tasks may be analyzed by evaluating appropriate information measures on the same state.
For example, advantages in quantum-state discrimination or hypothesis testing can be quantified via trace distance or relative entropy of the classical--quantum state.
Such extensions would further broaden the applicability of our framework.

Finally, it will be crucial to develop early fault-tolerant algorithms that actively exploit syndrome information to further suppress the impact of errors.
Since quantum syndrome-aware protocols can yield exponential improvements for logical observable estimation, it is natural to expect that adaptive modifications of logical operations conditioned on observed syndromes could also enhance algorithmic performance.
Developing syndrome-adaptive variants of key primitives such as Hamiltonian simulation and phase estimation may therefore accelerate the practical deployment of early fault-tolerant devices.

\section*{Acknowledgements}
The authors wish to thank Seok-Hyung Lee for helpful discussions.
K.T. is supported by the Program for Leading Graduate Schools (MERIT-WINGS) and JST BOOST Grant No. JPMJBS2418.
This work was supported in part with funding from Google.org.
H.K. is supported by the IITP (RS-2025-25464252) and the NRF (RS-2025-25464492, RS-2024-00442710) funded by the Ministry of Science and ICT (MSIT), Korea.
L.J. is supported by the ARO(W911NF-23-1-0077), ARO MURI (W911NF-21-1-0325), AFOSR MURI (FA9550-21-1-0209, FA9550-23-1-0338), DARPA (HR0011-24-9-0359, HR0011-24-9-0361), NSF (ERC-1941583, OMA-2137642, OSI-2326767, CCF-2312755, OSI-2426975), and the Packard Foundation (2020-71479).
N.Y. is supported by JST Grant Number JPMJPF2221, JST CREST Grant Number JPMJCR23I4, IBM Quantum, Google Quantum AI, JST ASPIRE Grant Number JPMJAP2316, JST ERATO Grant Number JPMJER2302, JST [Moonshot R\&D] [Grant Number JPMJMS256J], and Institute of AI and Beyond of the University of Tokyo.

\appendix
\section{Estimation protocol applied at the physical layer}
\label{sec_physical_layer}
In this section, we consider an estimation protocol performed directly on a noisy \emph{physical} quantum state and show that it is equivalent (up to a unitary transformation) to the syndrome-aware estimation protocol, as illustrated in Fig.~\ref{fig_decoding_estimation}.
Throughout this section, operators with $\bar{\cdot}$ act on the $k$-qubit logical Hilbert space, while operators without $\bar{\cdot}$ act on the $n$-qubit physical Hilbert space.

As in the main text, consider the task of preparing a $k$-qubit logical state
\begin{equation}
    \bar{\rho}(\vb*{\theta}) = \frac{1}{2^k}\qty(\bar{I} + \sum_{i=1}^{4^k-1}\theta_i\bar{P}_i),
\end{equation}
parameterized by an unnormalized generalized Bloch vector $\vb*{\theta} = (\theta_1,\ldots, \theta_{4^k-1})$~\cite{kimura2003bloch}.
Our goal is to estimate the expectation value of a logical Pauli operator $\bar{P}_i$, i.e., the parameter $\theta_i=\mathrm{tr}[\bar{\rho}(\vb*{\theta})\bar{P}_i]$.
Here, $\{\bar{P}_i\}_{i=0}^{4^k-1}$ denotes the set of $k$-qubit Pauli operators, with $\bar{P}_0=\bar{I}$.

When the device is subject to noise, we encode $\bar{\rho}(\vb*{\theta})$ into an $n$-qubit physical state using an $[[n,k,d]]$ stabilizer code.
Let $\{g_{a}\}_{a=1}^{r}$ be a set of stabilizer generators with $r=n-k$, and let $\{P_{i}\}_{i=1}^{4^{k}-1}$ be a choice of physical representatives of the nontrivial logical Pauli operators.
We describe encoding by an isometry channel $\mathcal{U}_{\mathrm{E}}(\cdot)=U_{\mathrm{E}}\cdot U_{\mathrm{E}}^\dagger$, which can be implemented as
\begin{equation}
    \mathcal{U}_{\mathrm{E}}(\cdot)
    = U_{\mathrm{E}}'(\ketbra{0^{r}}\otimes \cdot)\,U_{\mathrm{E}}'^\dagger,
\end{equation}
where the unitary $U_{\mathrm{E}}'$ satisfies
\begin{equation}
    \begin{aligned}
        U_{\mathrm{E}}' Z_a U_{\mathrm{E}}'^\dagger &= g_a \qquad (1\leq a\leq r), \\
        U_{\mathrm{E}}' (I_{r}\otimes\bar{P}_i) U_{\mathrm{E}}'^\dagger &= P_i \qquad (1\leq i\leq 4^{k}-1).
    \end{aligned}
\end{equation}
Here, $I_{r}$ is the identity operator on the $r$-qubit ancilla system and $Z_a$ is the Pauli-$Z$ operator acting on the $a$-th ancilla qubit.

The encoded physical state $\rho(\vb*{\theta})=\mathcal{U}_{\mathrm{E}}(\bar{\rho}(\vb*{\theta}))$ can be written as
\begin{equation}
    \begin{aligned}
        \rho(\vb*{\theta})
        &= U_{\mathrm{E}}' \qty(\prod_{a=1}^{r}\frac{1}{2}(I + Z_a))\frac{1}{2^k}\qty(I + \sum_{i=1}^{4^k-1}\theta_i\,I_{r}\otimes\bar{P}_i) U_{\mathrm{E}}'^\dagger \\
        &= \Pi_{0}\frac{1}{2^k}\qty(I + \sum_{i=1}^{4^k-1}\theta_i P_{i}),
    \end{aligned}
\end{equation}
where, for each syndrome $s\in\{0,1\}^r$, we define the projector
\begin{equation}
    \Pi_{s} = \prod_{a=1}^{r}\frac{1}{2}\bigl(I + (-1)^{s_a}g_{a}\bigr).
\end{equation}

Next, consider a physical Pauli channel $\mathcal{E}_Q(\cdot)=Q\cdot Q^\dagger$.
Then,
\begin{equation}
    \mathcal{E}_Q(\rho(\vb*{\theta}))
    = \Pi_{s}\frac{1}{2^k}\qty(I + \sum_{i=1}^{4^k-1}(-1)^{\langle P_i, Q \rangle}\theta_iP_{i}),
\end{equation}
where we define $\langle P,Q\rangle=0$ if $P$ and $Q$ commute and $\langle P,Q\rangle=1$ if they anticommute.
The syndrome $s\in\{0,1\}^r$ associated with $Q$ is determined by $s_a=\langle g_a,Q\rangle$.
Therefore, after physical Pauli noise channel $\mathcal{N}$ acts on $\rho(\vb*{\theta})$, the state can be decomposed into syndrome sectors as
\begin{equation}
    \mathcal{N}(\rho(\vb*{\theta}))
    = \sum_{s\in\{0,1\}^r} p_{s}\,\Pi_{s}\frac{1}{2^k}\qty(I + \sum_{i=1}^{4^k-1}\lambda_{i,s}\theta_iP_{i}),
\end{equation}
where $p_s\ge 0$ with $\sum_s p_s=1$ is the probability of obtaining syndrome $s$, and $\lambda_{i,s}\in[-1,1]$ describes the Pauli eigenvalue associated with $P_i$.

We now perform syndrome measurement followed by a Pauli recovery.
Upon observing syndrome $s$, we apply a Pauli operator $R_{s}$ such that $\langle g_a,R_s\rangle=s_a$ for all $a$.
The resulting corrected state through the recovery channel $\mathcal{R}$ is
\begin{equation}
    \begin{aligned}
        &~~~~\mathcal{R}\circ\mathcal{N}(\rho(\vb*{\theta}))\\
        &= \Pi_{0}\frac{1}{2^k}\qty(I + \sum_{i=1}^{4^k-1}\qty(\sum_{s\in\{0,1\}^r} p_s(-1)^{\langle P_i, R_s \rangle} \lambda_{i,s})\theta_iP_{i}).
    \end{aligned}
\end{equation}
Defining the induced logical noise channel
\begin{equation}
    \overline{\mathcal{N}} = \mathcal{U}_{\mathrm{E}}^\dagger \circ \mathcal{R}\circ \mathcal{N}\circ\mathcal{U}_{\mathrm{E}},
\end{equation}
we obtain
\begin{equation}
    \overline{\mathcal{N}}(\bar{\rho}(\vb*{\theta})) = \frac{1}{2^k}\qty(\bar{I} + \sum_{i=1}^{4^k-1} (1-2\epsilon_i)\theta_i\bar{P}_{i}),
\end{equation}
where
\begin{equation}
    \epsilon_i = \frac{1}{2}\qty(1-\sum_{s\in\{0,1\}^r} p_s\,(-1)^{\langle P_i, R_{s} \rangle}\lambda_{i,s}).
\end{equation}

On the other hand, applying $(U_{\mathrm{E}}')^\dagger$ to the noisy physical state $\mathcal{N}(\rho(\vb*{\theta}))$ yields a unitarily equivalent state of the form
\begin{equation}
    \mathcal{N}(\rho(\vb*{\theta}))
    \cong
    \sum_{s\in\{0,1\}^r} p_{s}\ketbra{s}\otimes\frac{1}{2^k}\qty(\bar{I} + \sum_{i=1}^{4^k-1}\lambda_{i,s}\theta_i\bar{P}_{i}),
\end{equation}
where $\cong$ means unitary equivalence.
Moreover, by applying an $s$-conditioned logical Pauli correction $\bar{R}_s$ satisfying
$\langle P_i, R_s\rangle=\langle \bar{P}_i,\bar{R}_s\rangle$ for all $i$,
we obtain
\begin{equation}
    \mathcal{N}(\rho(\vb*{\theta}))
    \cong
    \sum_{s\in\{0,1\}^r} p_s \ketbra{s} \otimes \overline{\mathcal{N}}_s(\bar{\rho}(\vb*{\theta})),
\end{equation}
where the $s$-conditioned noise channel is given as
\begin{equation}
    \overline{\mathcal{N}}_s(\bar{\rho}(\vb*{\theta})) = \frac{1}{2^k}\qty(\bar{I} + \sum_{i=1}^{4^k-1}(1-2\epsilon_{i,s})\theta_i\bar{P}_i)
\end{equation}
with
\begin{equation}
    \epsilon_{i,s}=\frac{1}{2}\qty(1-(-1)^{\langle P_i, R_s\rangle}\lambda_{i,s}).
\end{equation}
Therefore, the noisy physical state $\mathcal{N}(\rho(\vb*{\theta}))$ is unitarily equivalent to the classical--quantum state
$\sum_s p_s\ketbra{s}\otimes \overline{\mathcal{N}}_s(\bar{\rho}(\vb*{\theta}))$ used to model syndrome-aware estimation in the main text.

Since the quantum Fisher information is invariant under unitary transformations, any estimation protocol performed directly on the noisy physical state $\mathcal{N}(\rho(\vb*{\theta}))$ can be mapped to a quantum syndrome-aware estimation protocol acting on
$\sum_s p_s\ketbra{s}\otimes \overline{\mathcal{N}}_s(\bar{\rho}(\vb*{\theta}))$
with the same achievable precision.
In particular, protocols that perform estimation directly at the physical layer---including physical-level error-mitigation schemes such as those studied in Ref.~\cite{jeon2026quantum}---can be interpreted within our framework as quantum syndrome-aware estimation protocols that retain and exploit the syndrome record.
Therefore, the information difference between the physical-state setting and the usual syndrome-agnostic logical-state setting is solely characterized by whether the syndrome information is kept or discarded.

\section{Quantum Fisher information of multi-qubit quantum state}
\label{sec_estimation_theory_multiqubit}
In this section, based on Ref.~\cite{watanabe2010optimal}, we review the quantum Fisher information of noiseless and noisy multi-qubit quantum states.

\subsection{Analysis of noiseless multi-qubit quantum state}
First, we review the quantum Fisher information of a $k$-qubit quantum state parameterized as
\begin{equation}
    \bar{\rho}(\vb*{\theta}) = \frac{1}{2^k}\qty(\bar{I} + \sum_{i=1}^{4^k-1}\theta_i\bar{P}_i).
\end{equation}
Here, the parameters $\vb*{\theta} = (\theta_1,\ldots, \theta_{4^k-1})$ form an unnormalized Bloch vector~\cite{kimura2003bloch}, which represents the expectation value of the Pauli operator $\bar{P}_i$ as $\theta_i = \mathrm{tr}[\bar{\rho}(\vb*{\theta})\bar{P}_i]$.
By defining the covariance matrix $C(\vb*{\theta})$ of this state as
\begin{equation}
    \begin{aligned}
    C(\vb*{\theta})_{ij} 
    &= \mathrm{tr}\qty[\bar{\rho}(\vb*{\theta})\frac{1}{2}\{\bar{P}_i-\theta_i\bar{I},\bar{P}_j-\theta_j\bar{I}\}]\\
    &= \mathrm{tr}\qty[\bar{\rho}(\vb*{\theta})\frac{1}{2}\{\bar{P}_i,\bar{P}_j\}]
    - \mathrm{tr}\qty[\bar{\rho}(\vb*{\theta})\bar{P}_i]\,
      \mathrm{tr}\qty[\bar{\rho}(\vb*{\theta})\bar{P}_j],
    \end{aligned}
\end{equation}
the SLD operator $\bar{L}_i$ of this state, defined as in Eq.~\eqref{eq_SLD}, can be represented as
\begin{equation}
    \label{eq_SLD_multiqubit}
    \bar{L}_i = \sum_{j}(C(\vb*{\theta})^{-1})_{ij}(\bar{P}_j-\theta_{j}\bar{I}).
\end{equation}
This is because the SLD operator $\bar{L}_i$ satisfies
\begin{equation}
    \begin{aligned}
        \mathrm{tr}[\bar{\rho}(\vb*{\theta})\frac{1}{2}\{\bar{L}_i,\bar{P}_j-\theta_j\bar{I}\}]
        &= \mathrm{tr}\qty[\frac{1}{2}\{\bar{\rho}(\vb*{\theta}),\bar{L}_i\}(\bar{P}_j-\theta_j\bar{I})] \\
        &= \mathrm{tr}[\partial_{\theta_i}\bar{\rho}(\vb*{\theta})(\bar{P}_j-\theta_j\bar{I})] \\
        &= \mathrm{tr}[2^{-k}\bar{P}_i(\bar{P}_j-\theta_j\bar{I})] \\
        &= \delta_{ij}.
    \end{aligned}
\end{equation}
Therefore, the quantum Fisher information matrix $J$, defined in Eq.~\eqref{eq_QFIM}, is given as
\begin{equation}
    J = C(\vb*{\theta})^{-1}C(\vb*{\theta})C(\vb*{\theta})^{-1} = C(\vb*{\theta})^{-1},
\end{equation}
and the quantum Fisher information for the parameter $\mathrm{tr}[\bar{\rho}(\vb*{\theta})\bar{P}_i] = \theta_i$, defined in Eq.~\eqref{eq_QFI}, is given as
\begin{equation}
    J_{\theta_i}=\frac{1}{(J^{-1})_{ii}} = \frac{1}{C(\vb*{\theta})_{ii}} = \frac{1}{1-\theta_i^2}.
\end{equation}

Assume that our goal is to estimate the parameter $\theta_i = \mathrm{tr}[\bar{\rho}(\vb*{\theta})\bar{P}_i]$ given $N$ copies of $\bar{\rho}(\vb*{\theta})$, in the situation where all the parameters in $\vb*{\theta}$ are unknown.
In this case, the optimal measurement operator, given by Eq.~\eqref{eq_optimal_measurement_basis}, is
\begin{equation}
    \sum_{j}(J^{-1})_{ij}\bar{L}_j = \bar{P}_i-\theta_{i}\bar{I}.
\end{equation}
Therefore, the optimal measurement basis is given by the target Pauli operator $\bar{P}_i$ itself, and the optimal unbiased estimator $\theta_i^{\mathrm{est}}$ is obtained by averaging the measurement outcomes.
Indeed, $\mathbb{E}[\theta_i^{\mathrm{est}}] = \mathrm{tr}[\bar{\rho}(\vb*{\theta})\bar{P}_i] = \theta_i$, and its variance is
$\mathrm{Var}[\theta_i^{\mathrm{est}}] = N^{-1}(\mathrm{tr}[\bar{\rho}(\vb*{\theta})\bar{P}_i^2]-\mathrm{tr}[\bar{\rho}(\vb*{\theta})\bar{P}_i]^2)
= N^{-1}(1-\theta_i^2)$,
which attains the quantum Cram\'er--Rao bound in Eq.~\eqref{eq_QCRB}.

\subsection{Analysis of noisy multi-qubit quantum state}
Next, we consider the case where the quantum state $\bar{\rho}(\vb*{\theta})$ is affected by a noise channel $\overline{\mathcal{N}}$, resulting in a noisy quantum state
\begin{equation}
    \overline{\mathcal{N}}(\bar{\rho}(\vb*{\theta})) = \frac{1}{2^k}\qty(\bar{I} + \sum_{i=1}^{4^k-1}((A\vb*{\theta})_i+c_i)\bar{P}_i).
\end{equation}
Here, $A\in\mathbb{R}^{(4^{k}-1)\times (4^{k}-1)}$ is a Pauli transfer matrix defined as $A_{ij} = 2^{-k}\mathrm{tr}[\bar{P}_i\overline{\mathcal{N}}(\bar{P}_j)]$, and $c_i = 2^{-k}\mathrm{tr}[\bar{P}_i\overline{\mathcal{N}}(\bar{I})]$.
By defining the modified parameter $\vb*{\theta}' = A\vb*{\theta}+\vb*{c}$, the noisy state can be represented as $\overline{\mathcal{N}}(\bar{\rho}(\vb*{\theta})) = \bar{\rho}(\vb*{\theta}')$.
For $\bar{\rho}(\vb*{\theta}')$, the SLD operators and quantum Fisher information matrix in terms of the parameters $\vb*{\theta}'$ can be represented as
\begin{align}
    \bar{L}_i' &= \sum_{j}(C(\vb*{\theta}')^{-1})_{ij}(\bar{P}_j-\theta_{j}'\bar{I}), \\
    J' &= C(\vb*{\theta}')^{-1}.
\end{align}
Meanwhile, the SLD operators and quantum Fisher information matrices under the two parameterizations $\vb*{\theta}$ and $\vb*{\theta}'$ are related as
\begin{align}
    \bar{L}_i &= \sum_j \pdv{\theta'_j}{\theta_i} \bar{L}'_j = \sum_j (A^{\mathrm{T}})_{ij} \bar{L}'_j, \\
    J_{ij} &= \sum_{kl} \pdv{\theta'_k}{\theta_i} J_{kl}' \pdv{\theta'_l}{\theta_j} = \sum_{kl} (A^{\mathrm{T}})_{ik} J_{kl}' A_{lj}.
\end{align}
Therefore, the SLD operators and the quantum Fisher information matrix of the noisy state $\overline{\mathcal{N}}(\bar{\rho}(\vb*{\theta}))$ in terms of the parameters $\vb*{\theta}$ are
\begin{align}
    \bar{L}_i &= \sum_{jk} (A^{\mathrm{T}})_{ij} (C(A\vb*{\theta}+\vb*{c})^{-1})_{jk}(\bar{P}_k-\theta_{k}'\bar{I}), \\
    J &= A^{\mathrm{T}} C(A\vb*{\theta}+\vb*{c})^{-1} A,
\end{align}
and the quantum Fisher information for the parameter $\theta_i = \mathrm{tr}[\bar{\rho}(\vb*{\theta})\bar{P}_i]$ is given as
\begin{equation}
    J_{\theta_i} = \frac{1}{(A^{-1} C(A\vb*{\theta}+\vb*{c})(A^{\mathrm{T}})^{-1})_{ii}}.
\end{equation}

Consider the task of estimating the parameter $\theta_i = \mathrm{tr}[\bar{\rho}(\vb*{\theta})\bar{P}_i]$ given $N$ copies of the noisy state $\overline{\mathcal{N}}(\bar{\rho}(\vb*{\theta}))$.
In this case, we have
\begin{equation}
    \sum_{j}(J^{-1})_{ij}\bar{L}_j = \sum_{j}(A^{-1})_{ij}(\bar{P}_j-\theta_{j}'\bar{I}) = \qty(\overline{\mathcal{N}}^{-1})^{\dagger}(\bar{P}_i) + a \bar{I}
\end{equation}
for some $a\in\mathbb{R}$.
Therefore, the optimal measurement basis is given by $\big(\overline{\mathcal{N}}^{-1}\big)^{\dagger}(\bar{P}_i)$, and the optimal unbiased estimator $\theta_i^{\mathrm{est}}$ is obtained by averaging the measurement outcomes.
Indeed, $\mathbb{E}[\theta_i^{\mathrm{est}}] = \mathrm{tr}[\overline{\mathcal{N}}(\bar{\rho}(\vb*{\theta}))\big(\overline{\mathcal{N}}^{-1}\big)^{\dagger}(\bar{P}_i)] = \theta_i$, and its variance is
$\mathrm{Var}[\theta_i^{\mathrm{est}}] =N^{-1} (A^{-1} C(A\vb*{\theta}+\vb*{c})(A^{\mathrm{T}})^{-1})_{ii}$, which attains the quantum Cram\'er--Rao bound in Eq.~\eqref{eq_QCRB}.

In the main text, we especially consider the case where the noise channel $\overline{\mathcal{N}}$ is a Pauli channel, so $A=\Lambda$ becomes a diagonal matrix and $\vb*{c}=0$.
By representing $\Lambda = \mathrm{diag}(1-2\epsilon_1,\ldots,1-2\epsilon_{4^{k}-1})$, the optimal measurement basis is represented as $\big(\overline{\mathcal{N}}^{-1}\big)^{\dagger}(\bar{P}_i) = (1-2\epsilon_i)^{-1}\bar{P}_i$.
Therefore, the optimal estimation protocol boils down to measuring the noisy state $\overline{\mathcal{N}}(\bar{\rho}(\vb*{\theta}))$ with the original Pauli operator $\bar{P}_i$, and then rescaling the measurement outcome by a factor of $(1-2\epsilon_i)^{-1}$.

\section{Proof of the results for classical syndrome-aware protocols}
\label{sec_proof_classical}
In this section, we provide detailed proofs of the results for classical syndrome-aware protocols discussed in Sec.~\ref{sec_classical}.

\subsection{Proof of Theorem~\ref{thm_1}}
\label{sec_proof_classical_1}
\begin{proof}
    From the definition, the effective logical error rate satisfies
    \begin{equation}
        \label{eq_proof_c_1_1}
        \begin{aligned}
            \epsilon_i^{\mathrm{cSynd}} 
            &= f_{\theta_i}^{-1}\qty(\sum_{s\in\mathcal{S}} p_{s}  f_{\theta_i}\qty(\epsilon_{i,s}))\\
            &\leq \sum_{s\in\mathcal{S}} p_{s}f_{\theta_i}^{-1}\qty( f_{\theta_i}\qty(\epsilon_{i,s}))\\
            &= \sum_{s\in\mathcal{S}} p_{s}\epsilon_{i,s} \\
            &= \epsilon_{i}.
        \end{aligned}
    \end{equation}
    Here, $f_{\theta_i}^{-1}$ is the inverse function of 
    \begin{equation}
        f_{\theta_i}(\epsilon)
        = \frac{(1-2\epsilon)^2}{1-(1-2\epsilon)^2\theta_i^2}
    \end{equation}
    defined in Eq.~\eqref{eq_def_f} on the domain $0\leq\epsilon\leq1/2$.
    The first equality follows from the definition of the effective logical error rate in Eq.~\eqref{eq_effective_logical_error_rate_classical}.
    The next inequality follows from Jensen's inequality (since $f_{\theta_i}^{-1}$ is convex).
    The last equality follows from Eq.~\eqref{eq_conditional_logical_error}.
    Moreover, from the convexity of $f^{-1}_{\theta_i}$, we have
    \begin{equation}
        \label{eq_proof_c_1_2}
        \begin{aligned}
            \epsilon_{i,s}
            &= f^{-1}_{\theta_i}(f_{\theta_i}(\epsilon_{i,s})) \\
            &\leq \frac{f^{-1}_{\theta_i}(f_{\theta_i}(0)) - f^{-1}_{\theta_i}(f_{\theta_i}(1/2))}{f_{\theta_i}(0) - f_{\theta_i}(1/2)}\bigl(f_{\theta_i}(\epsilon_{i,s})-f_{\theta_i}(1/2)\bigr) \\
            &~~~~+  f^{-1}_{\theta_i}(f_{\theta_i}(1/2))\\
            &= -\frac{1-\theta_i^2}{2}f_{\theta_i}(\epsilon_{i,s}) + \frac{1}{2}.
        \end{aligned}
    \end{equation}
    Multiplying by $p_s$ and summing over $s\in\mathcal{S}$, we obtain
    \begin{equation}
        \label{eq_proof_c_1_3}
        \begin{aligned}
            \epsilon_{i}
            &\leq -\frac{1-\theta_i^2}{2}f_{\theta_i}(\epsilon_i^{\mathrm{cSynd}}) + \frac{1}{2} \\
            &= \epsilon_i^{\mathrm{cSynd}}\frac{2-2\epsilon_i^{\mathrm{cSynd}}}{1-\theta_i^2(1-2\epsilon_i^{\mathrm{cSynd}})^2}.
        \end{aligned}
    \end{equation}
    Meanwhile, we have
    \begin{equation}
        \label{eq_proof_c_1_4}
        \frac{2-2\epsilon_i^{\mathrm{cSynd}}}{1-\theta_i^2(1-2\epsilon_i^{\mathrm{cSynd}})^2} \leq \frac{2}{1-\theta_i^2},
    \end{equation}
    since the left-hand side is maximized when $\epsilon_i^{\mathrm{cSynd}}=0$.
    By combining Eqs.~\eqref{eq_proof_c_1_1}, \eqref{eq_proof_c_1_3}, \eqref{eq_proof_c_1_4}, we obtain
    \begin{equation}
        \frac{1-\theta_i^2}{2}\epsilon_i \leq \epsilon_i^{\mathrm{cSynd}} \leq \epsilon_i.
    \end{equation}
    
    The Haar-average result follows from the fact that
    \begin{equation}
        \label{eq_proof_c_1_5}
        \mathbb{E}_{\mathrm{Haar}}[\theta^2_i]
        = \mathbb{E}_{\mathrm{Haar}}[\mathrm{tr}[\ketbra{\psi} \bar{P}_i]^2]
        = \frac{1}{2^k+1}.
    \end{equation}
    See Ref.~\cite{mele2024introduction} for details.
\end{proof}

\subsection{Proof of Corollary~\ref{cor_1}}
\label{sec_proof_classical_2}
\begin{proof}
    From Theorem~\ref{thm_1}, we have
    \begin{equation}
        \label{eq_proof_c_2_1}
        \frac{1-\theta_i^2}{2}\epsilon_i \leq \epsilon_i^{\mathrm{cSynd}} \leq \epsilon_i.
    \end{equation}
    This implies that if $\epsilon_i=\Theta(\eta^l)$, then $\epsilon_i^{\mathrm{cSynd}}=\Theta(\eta^l)$, assuming $\theta_i\neq\pm1$.

    Moreover, the code threshold for the logical error rate $\epsilon_i$ (respectively, the effective logical error rate $\epsilon_i^{\mathrm{cSynd}}$) is defined as the maximal physical error rate such that, if $\eta$ is below this value, $\epsilon_i\to0$ (respectively, $\epsilon_i^{\mathrm{cSynd}}\to0$) as the code distance $d$ increases.
    From Eq.~\eqref{eq_proof_c_2_1}, $\epsilon_i\to0$ implies $\epsilon_i^{\mathrm{cSynd}}\to0$, and $\epsilon_i^{\mathrm{cSynd}}\to0$ implies $\epsilon_i\to0$.
    Therefore, the threshold defined in terms of the effective logical error rate $\epsilon_i^{\mathrm{cSynd}}$ is identical to that defined in terms of the logical error rate $\epsilon_i$ under the maximum-likelihood decoder.
\end{proof}

\subsection{Proof of Corollary~\ref{cor_2}}
\label{sec_proof_classical_3}
\begin{proof}
    Since $N^{\mathrm{cSynd}} \leq N^{\mathrm{L}}$ is obvious from Eq.~\eqref{eq_cfi_inequality}, we prove the other inequality.
    From Eq.~\eqref{eq_proof_c_1_2}, we have
    \begin{equation}
        f_{\theta_i}(\epsilon_{i,s}) \leq \frac{1-2\epsilon_{i,s}}{1-\theta_i^2}.
    \end{equation}
    Multiplying by $p_s$ and summing over $s\in\mathcal{S}$, we obtain
    \begin{equation}
        \label{eq_proof_c_3_2}
        F_{\theta_i}^{\mathrm{Synd}} \leq \frac{1-2\epsilon_i}{1-\theta_i^2}
        = \frac{1}{1-\theta_i^2}\sqrt{\frac{f_{\theta_i}(\epsilon_{i})}{1+\theta_i^2f_{\theta_i}(\epsilon_i)}}.
    \end{equation}
    Taking the inverse of the above inequality, we obtain
    \begin{equation}
        (1-\theta_i^2)\sqrt{(F_{\theta_i}^{\mathrm{L}})^{-1}+\theta_i^2}\leq (F_{\theta_i}^{\mathrm{Synd}})^{-1}.
    \end{equation}
    Dividing both sides by $\sigma_{\mathrm{target}}^2$ yields the claimed inequality.

    The Haar-average results follow from the fact that
    \begin{equation}
        \begin{aligned}
            \mathbb{E}_{\mathrm{Haar}}[N^{\mathrm{cSynd}}]
            &=\frac{\mathbb{E}_{\mathrm{Haar}}[(F_{\theta_i}^{\mathrm{Synd}})^{-1}]}{\sigma_{\mathrm{target}}^2}\\
            &\geq \frac{1-\mathbb{E}_{\mathrm{Haar}}[\theta_i^2]}{\sigma_{\mathrm{target}}^2(1-2\epsilon_i)}\\
            &= \frac{1}{\sigma_{\mathrm{target}}^2(1-2\epsilon_i)}\frac{2^k}{2^k+1}\\
            &\sim \frac{1}{\sigma_{\mathrm{target}}^2(1-2\epsilon_i)},
        \end{aligned}
    \end{equation}
    and
    \begin{equation}
        \begin{aligned}
            \mathbb{E}_{\mathrm{Haar}}[N^{\mathrm{L}}]
            &=\frac{\mathbb{E}_{\mathrm{Haar}}[(F_{\theta_i}^{\mathrm{L}})^{-1}]}{\sigma_{\mathrm{target}}^2}\\
            &= \frac{1}{\sigma_{\mathrm{target}}^2}\qty(\frac{1}{(1-2\epsilon_i)^2} - \mathbb{E}_{\mathrm{Haar}}[\theta_i^2])\\
            &= \frac{1}{\sigma_{\mathrm{target}}^2}\qty(\frac{1}{(1-2\epsilon_i)^2} - \frac{1}{2^k+1})\\
            &\sim \frac{1}{\sigma_{\mathrm{target}}^2}\frac{1}{(1-2\epsilon_i)^2}.
        \end{aligned}
    \end{equation}
    Here, we use Eqs.~\eqref{eq_QFI_agnostic}, \eqref{eq_proof_c_1_5}, and \eqref{eq_proof_c_3_2}.
\end{proof}

\subsection{Proof of Proposition~\ref{prop_1}}
\label{sec_proof_classical_4}
\begin{proof}
    From Eq.~\eqref{eq_eff_err_rate_classical_even}, we have
    \begin{equation}
        \label{eq_c_4_1}
        \lim_{\eta\to0} \frac{\epsilon_i^{\mathrm{cSynd}}}{\epsilon_i}
        = \lim_{\eta\to0} \frac{\sum_{s\in\mathcal{S}_{\Theta(1)}}p_s\bigl(f_{\theta_i}(\epsilon_{i,s})-f_{\theta_i}(0)\bigr)}{\sum_{s\in\mathcal{S}_{\Theta(1)}} p_s\,f_{\theta_i}^{(1)}\,\epsilon_{i,s}}.
    \end{equation}
    Here, we expand the function $f_{\theta_i}(\epsilon)$ around $\epsilon=0$ as
    $f_{\theta_i}(\epsilon) = f_{\theta_i}(0) + f_{\theta_i}^{(1)}\epsilon + O(\epsilon^2)$,
    where $f_{\theta_i}^{(1)} = -4(1-\theta_i^2)^{-2}$.
    From the convexity of the function $f_{\theta_i}(\epsilon)$, we have
    \begin{equation}
        \label{eq_c_4_2}
        \frac{1-\theta_i^2}{2} \leq \frac{f_{\theta_i}(\epsilon_{i,s}) - f_{\theta_i}(0)}{f_{\theta_i}^{(1)}\epsilon_{i,s}} < 1
    \end{equation}
    for $s\in\mathcal{S}_{\Theta(1)}$, where $0<\epsilon_{i,s}\leq 1/2$.
    The equality holds if and only if $\epsilon_{i,s} = 1/2$.
    Combining Eqs.~\eqref{eq_c_4_1} and~\eqref{eq_c_4_2}, we obtain the claimed inequality in Proposition~\ref{prop_1}.
\end{proof}

\subsection{Proof of Proposition~\ref{prop_2}}
\label{sec_proof_classical_5}
\begin{proof}
    From Eq.~\eqref{eq_eff_err_rate_classical_odd}, we have
    \begin{equation}
        \begin{aligned}
            &~~~~\lim_{\eta\to0} \frac{\epsilon_i^{\mathrm{cSynd}}}{\epsilon_i}\\
            &= \lim_{\eta\to0}
            \frac{\sum_{s\in\mathcal{S}_{\Theta(\eta)}} p_s f_{\theta_i}^{(1)}\epsilon_{i,s}
            + \sum_{s\in\mathcal{S}_{\Theta(1)}} p_s\bigl(f_{\theta_i}(\epsilon_{i,s})-f_{\theta_i}(0)\bigr)}
            {\sum_{s\in\mathcal{S}_{\Theta(\eta)}} p_s f_{\theta_i}^{(1)}\epsilon_{i,s}
            + \sum_{s\in\mathcal{S}_{\Theta(1)}} p_s f_{\theta_i}^{(1)}\epsilon_{i,s}}.
        \end{aligned}
    \end{equation}
    For $s\in\mathcal{S}_{\Theta(1)}$ (where $0<\epsilon_{i,s}\leq 1/2$), we have
    \begin{equation}
        \frac{1-\theta_i^2}{2} \leq \frac{f_{\theta_i}(\epsilon_{i,s}) - f_{\theta_i}(0)}{f_{\theta_i}^{(1)}\epsilon_{i,s}} < 1.
    \end{equation}
    Therefore, the limit ratio satisfies $\lim_{\eta\to0}\epsilon_i^{\mathrm{cSynd}}/\epsilon_i<1$ if and only if $\mathcal{S}_{\Theta(1)}\neq\emptyset$.
\end{proof}

\section{Proof of the results for quantum syndrome-aware protocols}
\label{sec_proof_quantum}
In this section, we provide detailed proofs of the results for quantum syndrome-aware protocols discussed in Sec.~\ref{sec_quantum}.
Before presenting the proofs, we first introduce useful lemmas about the covariance matrix $C(\vb*{\theta})$.
We then prove Proposition~\ref{prop_3}, Lemma~\ref{lem_1}, and Theorem~\ref{thm_2} in this order.
Finally, we perform additional analysis on the quantum Fisher information of the classical--quantum state in Eq.~(\ref{eq_cq_state_simplified}).

Before going into details, we note that the ideal logical state $\bar{\rho}(\vb*{\theta})$ should be of full rank for the covariance matrix $C(\vb*{\theta})$ to have an inverse.
Therefore, strictly speaking, whenever an inverse matrix appears, we take the ideal state to be a Haar-random state affected by a small depolarizing noise,
$\bar{\rho}((1-\delta)\vb*{\theta}) = (1-\delta)\ketbra{\psi} + \delta \frac{\bar{I}}{2^k}$,
and then take $\delta\to0$ after the calculation.
Accordingly, when writing $\mathbb{E}_{\mathrm{Haar}}[f(\vb*{\theta})]$ for some function $f$ in which matrix inversion appears, we mean that we first take the Haar average of $f((1-\delta)\vb*{\theta})$ over Haar-random states $\ket{\psi}$, and then take the limit $\delta\to0$.
Nevertheless, in most cases where $f(\vb*{\theta})$ does not involve matrix inversion, we can simply set $\delta = 0$ and take the average.
Therefore, we regard the ideal logical state as $\bar{\rho}((1-\delta)\vb*{\theta})$ only when an explicit calculation involving matrix inversion is required.

\subsection{Properties of the covariance matrix}
\label{sec_proof_quantum_1}
For a $k$-qubit quantum state
\begin{equation}
    \bar{\rho}(\vb*{\theta}) = \frac{1}{2^k}\qty(\bar{I} + \sum_{i=1}^{4^k-1}\theta_i\bar{P}_i),
\end{equation}
we have defined its covariance matrix as
\begin{equation}
    C(\vb*{\theta})_{ij}
    = \mathrm{tr}\qty[\bar{\rho}(\vb*{\theta})\frac{1}{2}\{\bar{P}_i,\bar{P}_j\}]
    - \mathrm{tr}\qty[\bar{\rho}(\vb*{\theta})\bar{P}_i]\,
      \mathrm{tr}\qty[\bar{\rho}(\vb*{\theta})\bar{P}_j].
\end{equation}
In particular, when $\bar{\rho}(\vb*{\theta})=\ketbra{\psi}$ is a pure state, we obtain the following lemma.

\begin{lem}
    \label{lem_2}
    Let $\bar{\rho}(\vb*{\theta}) = \ketbra{\psi}$ be a pure state.
    Then, the corresponding $(4^k-1)\times(4^k-1)$ covariance matrix $C(\vb*{\theta})$ has the following properties:
    \begin{itemize}
        \item $C(\vb*{\theta})\vb*{\theta} = \vb*{0}$.
        \item $C(\vb*{\theta})$ has eigenvalues $0$ and $2^{k-1}$ with multiplicities $(2^k-1)^2$ and $2(2^k-1)$, respectively.
    \end{itemize}
\end{lem}

\begin{proof}
    For $\vb*{o}\in \mathbb{R}^{4^k-1}$, define an observable $\bar{O} = \sum_i o_i\bar{P}_i$.
    Let $\{\ket{\psi_i}\}_{i=1}^{2^k}$ denote an orthonormal basis of the $k$-qubit Hilbert space with $\ket{\psi_1} = \ket{\psi}$, and write the matrix representation of $\bar{O}$ in this basis as
    \begin{equation}
        \begin{pmatrix}
            a & \vb*{h}^\dagger \\
            \vb*{h} & B \\
        \end{pmatrix}
    \end{equation}
    with $a\in\mathbb{C}$, $\vb*{h}\in\mathbb{C}^{2^k-1}$, and $B\in\mathbb{C}^{(2^k-1)\times(2^k-1)}$ satisfying $a + \mathrm{tr}[B] = 0$.
    Note that the first row and column correspond to $\ket{\psi_1} = \ket{\psi}$.
    Then,
    \begin{equation}
        \vb*{o}^{\mathrm{T}}C(\vb*{\theta})\vb*{o} = \ev*{\bar{O}^2}{\psi} - \ev*{\bar{O}}{\psi}^2 = \norm{\vb*{h}}^2.
    \end{equation}
    This implies that if the operator $\bar{O} = \sum_i o_i\bar{P}_i$ is of the form
    \begin{equation}
        \begin{pmatrix}
            a & \vb*{0}^\dagger \\
            \vb*{0} & B \\
        \end{pmatrix},
    \end{equation}
    then the corresponding vector $\vb*{o}$ belongs to the kernel of $C(\vb*{\theta})$.
    Since $\sum_i \theta_i \bar{P}_i = 2^k\ketbra{\psi}-\bar{I}$ is of this form, we have $C(\vb*{\theta})\vb*{\theta} = \vb*{0}$.
    Moreover, since a Hermitian matrix $B\in\mathbb{C}^{(2^k-1)\times(2^k-1)}$ has $(2^k-1)^2$ real degrees of freedom, the dimension of the kernel of $C(\vb*{\theta})$ is $(2^k-1)^2$.

    Meanwhile, the image of $C(\vb*{\theta})$ is spanned by vectors $\vb*{o}$ whose corresponding operators $\bar{O} = \sum_i o_i\bar{P}_i$ are of the form
    \begin{equation}
        \begin{pmatrix}
            0 & \vb*{h}^\dagger \\
            \vb*{h} & 0 \\
        \end{pmatrix}.
    \end{equation}
    Let us choose an orthonormal basis $\{\vb*{o}_i\}_{i=1}^{2(2^k-1)}$ of the support such that $\bar{O}_i = \sum_j (\vb*{o}_i)_j\bar{P}_j$ is given by
    \begin{equation}
        \begin{aligned}
            &2^{(k-1)/2}(\ketbra{\psi}{\psi_{(i+3)/2}} + \ketbra{\psi_{(i+3)/2}}{\psi})~~~~(\text{$i$: odd}),\\
            &2^{(k-1)/2}i(\ketbra{\psi}{\psi_{i/2+1}} - \ketbra{\psi_{i/2+1}}{\psi})~~~~(\text{$i$: even}).
        \end{aligned}
    \end{equation}
    We can check the orthonormality of this basis as $\vb*{o}_i^{\mathrm{T}}\vb*{o}_j = 2^{-k}\mathrm{tr}[\bar{O}_i\bar{O}_j] = \delta_{ij}$.
    Moreover,
    \begin{equation}
        \vb*{o}_i^{\mathrm{T}}C(\vb*{\theta})\vb*{o}_j = \frac{1}{2}\ev*{\{\bar{O}_i,\bar{O}_j\}}{\psi} = 2^{k-1} \delta_{ij}.
    \end{equation}
    This means that all the nonzero eigenvalues  of $C(\vb*{\theta})$ are $2^{k-1}$.
    Therefore, $C(\vb*{\theta})$ has eigenvalues $0$ and $2^{k-1}$ with multiplicities $(2^k-1)^2$ and $2(2^k-1)$, respectively.
\end{proof}

Next, as in the main text, let us decompose the $k$-qubit Hilbert space as $\mathcal{H}=\mathcal{H}_A\otimes\mathcal{H}_B$, where $\mathcal{H}_A$ consists of $k'(<k/2)$ qubits and $\mathcal{H}_B$ consists of the remaining $k-k'$ qubits.
We also decompose the index set $\mathcal{J}=\{1,\ldots,4^k-1\}$ of nontrivial $k$-qubit logical Pauli operators $\{\bar P_j\}_{j\in\mathcal J}$ as $\mathcal{J}=\mathcal{J}_A\sqcup\mathcal{J}_B$, where
\begin{equation}
    \begin{aligned}
        \mathcal{J}_A &= \{j\in\mathcal{J}~|~\bar{P}_j=\bar P_A\otimes \bar P_B,\ \bar P_A\neq \bar I_A\},\\
        \mathcal{J}_B &= \{j\in\mathcal{J}~|~\bar{P}_j=\bar I_A\otimes \bar P_B,\ \bar P_B\neq \bar I_B\}.
    \end{aligned}
\end{equation}
Here, $\bar P_A$ and $\bar P_B$ denote Pauli operators on subsystems $A$ and $B$, respectively.
In the following, the subscripts $A$ and $B$ on vectors and matrices refer to the blocks associated with the decomposition $\mathcal{J}=\mathcal{J}_A\sqcup\mathcal{J}_B$.

The quantum Fisher information matrix for all parameters is given by $C(\vb*{\theta})^{-1}\in\mathbb{R}^{|\mathcal{J}|\times|\mathcal{J}|}$.
Therefore, when the parameters $\theta_j$ with $j\in\mathcal{J}_B$ are known, while the parameters $\theta_j$ with $j\in\mathcal{J}_A$ are treated as the target and nuisance parameters, the relevant quantum Fisher information matrix is the principal submatrix $(C(\vb*{\theta})^{-1})_{AA}\in\mathbb{R}^{|\mathcal{J}_A|\times|\mathcal{J}_A|}$.
The inverse of this matrix admits the following variational characterization.

\begin{lem}
    \label{lem_3}
    For any vector $\vb*{u}_A\in\mathbb{R}^{|\mathcal{J}_A|}$, we have
    \begin{equation}
        \begin{aligned}
            &~~~~\vb*{u}_A^{\mathrm{T}} \bigl((C(\vb*{\theta})^{-1})_{AA}\bigr)^{-1} \vb*{u}_A \\
            &= \min_{\vb*{u}_B\in\mathbb{R}^{|\mathcal{J}_B|}} \mathrm{Var}_{\bar{\rho}(\vb*{\theta})} \left( \vb*{u}_A^{\mathrm{T}}\bar{\vb*{P}}_{A} + \vb*{u}_B^{\mathrm{T}}\bar{\vb*{P}}_{B} \right),
        \end{aligned}
    \end{equation}
    where the minimum is achieved by choosing
    \begin{equation}
        \vb*{u}_B = - (C(\vb*{\theta})_{BB})^{-1} C(\vb*{\theta})_{BA} \vb*{u}_A.
    \end{equation}
    Here, $\mathrm{Var}_{\bar{\rho}(\vb*{\theta})}(\cdot)$ denotes the variance of the operator $\cdot$ with respect to $\bar{\rho}(\vb*{\theta})$, and $\bar{\vb*{P}}_{A}$ and $\bar{\vb*{P}}_{B}$ denote the vectors of logical Pauli operators whose indices belong to $\mathcal{J}_A$ and $\mathcal{J}_B$, respectively.
\end{lem}

Note that for the covariance matrix $C(\vb*{\theta})$ and a vector
$\vb*{u}\in\mathbb{R}^{|\mathcal{J}|}$, we have
\begin{equation}
    \vb*{u}^{\mathrm{T}}C(\vb*{\theta})\vb*{u}
    =
    \mathrm{Var}_{\bar{\rho}(\vb*{\theta})}
    \left(\vb*{u}^{\mathrm{T}}\bar{\vb*{P}}\right).
\end{equation}
Therefore, Lemma~\ref{lem_3} shows that, once the parameters $\theta_j$ with $j\in\mathcal{J}_B$ are known, one can reduce the estimator variance by adding observables supported on the $B$ sector whose expectation values are known.
\begin{proof}
    From the definition of the covariance matrix, we have
    \begin{equation}
        \begin{aligned}
            & \mathrm{Var}_{\bar{\rho}(\vb*{\theta})} \left( \vb*{u}_A^{\mathrm{T}}\bar{\vb*{P}}_{A} + \vb*{u}_B^{\mathrm{T}}\bar{\vb*{P}}_{B} \right)\\
            &=
            \begin{pmatrix}
                \vb*{u}_A^{\mathrm{T}} & \vb*{u}_B^{\mathrm{T}}
            \end{pmatrix}
            C(\vb*{\theta})
            \begin{pmatrix}
                \vb*{u}_A\\
                \vb*{u}_B
            \end{pmatrix}\\
            &= \vb*{u}_A^{\mathrm{T}} C(\vb*{\theta})_{AA}\vb*{u}_A + 2\vb*{u}_A^{\mathrm{T}} C(\vb*{\theta})_{AB}\vb*{u}_B + \vb*{u}_B^{\mathrm{T}} C(\vb*{\theta})_{BB}\vb*{u}_B.
        \end{aligned}
    \end{equation}
    This quadratic form in $\vb*{u}_B$ is minimized at
    \begin{equation}
        \vb*{u}_B = - (C(\vb*{\theta})_{BB})^{-1} C(\vb*{\theta})_{BA}\vb*{u}_A,
    \end{equation}
    and the minimum value is
    \begin{equation}
        \vb*{u}_A^{\mathrm{T}} \left[ C(\vb*{\theta})_{AA} - C(\vb*{\theta})_{AB} (C(\vb*{\theta})_{BB})^{-1} C(\vb*{\theta})_{BA} \right] \vb*{u}_A.
    \end{equation}
    By the Schur complement identity,
    \begin{equation}
        \left((C(\vb*{\theta})^{-1})_{AA}\right)^{-1}
        = C(\vb*{\theta})_{AA} - C(\vb*{\theta})_{AB} (C(\vb*{\theta})_{BB})^{-1} C(\vb*{\theta})_{BA}.
    \end{equation}
    Therefore,
    \begin{equation}
        \begin{aligned}
            &~~~~ \min_{\vb*{u}_B\in\mathbb{R}^{|\mathcal{J}_B|}} \mathrm{Var}_{\bar{\rho}(\vb*{\theta})} \left( \vb*{u}_A^{\mathrm{T}}\bar{\vb*{P}}_{A} + \vb*{u}_B^{\mathrm{T}}\bar{\vb*{P}}_{B} \right)\\
            &= \vb*{u}_A^{\mathrm{T}} \bigl((C(\vb*{\theta})^{-1})_{AA}\bigr)^{-1} \vb*{u}_A.
        \end{aligned}
    \end{equation}
\end{proof}

\subsection{Proof of Proposition~\ref{prop_3}}
\label{sec_proof_quantum_2}
\begin{proof}
We divide the proof into three steps.
First, using the invariance of the Haar measure under local Clifford transformations, we reduce the problem to evaluating two constants $\alpha_0$ and $\alpha_1$.
These constants are the two possible diagonal entries of
$\mathbb{E}_{\mathrm{Haar}}[((C(\vb*{\theta})^{-1})_{AA})^{-1}]$:
$\alpha_0$ corresponds to the diagonal entries indexed by Pauli operators of the form $\bar{P}_A\otimes\bar{I}_B$, while $\alpha_1$ corresponds to those indexed by Pauli operators of the form $\bar{P}_A\otimes\bar{P}_B$ with $\bar{P}_B\neq\bar{I}_B$.
The desired matrix inequality then follows by showing that both $\alpha_0$ and $\alpha_1$ are upper bounded by $2/2^{k-2k'}$.

\paragraph*{Step 1: reduction to two constants.}
Let an orthonormal matrix $V$ denote the Pauli transfer matrix of a unitary operator $\bar{U} = \bar{U}_A\otimes\bar{U}_B$, i.e., $V_{ij} = 2^{-k}\mathrm{tr}[\bar{P}_i\bar{U}\bar{P}_j\bar{U}^\dagger]$.
Since the off-diagonal blocks of $V$ are zero ($V_{AB} = V_{BA} = 0$), we have
\begin{equation}
    \label{eq_proof_q_2_1}
    \begin{aligned}
        &~~~~\mathbb{E}_{\mathrm{Haar}}[((C(\vb*{\theta})^{-1})_{AA})^{-1}]\\
        &= \mathbb{E}_{\mathrm{Haar}}[((C(V\vb*{\theta})^{-1})_{AA})^{-1}] \\
        &= V_{AA}\mathbb{E}_{\mathrm{Haar}}[((C(\vb*{\theta})^{-1})_{AA})^{-1}](V_{AA})^{\mathrm{T}}.
    \end{aligned}
\end{equation}
Here, we use the fact that $\mathbb{E}_{\mathrm{Haar}}[f(\vb*{\theta})] = \mathbb{E}_{\mathrm{Haar}}[f(V\vb*{\theta})]$ and $C(V\vb*{\theta}) = VC(\vb*{\theta})V^{\mathrm{T}}$.
Let $\bar{U}$ be a Pauli operator that anti-commutes with $\bar{P}_j$ and commutes with $\bar{P}_{j'}$ (note that any Pauli operator can be represented as $\bar{U} = \bar{U}_A\otimes\bar{U}_B$).
Then, $V$ is a diagonal matrix with $V_{jj} = -1$ and $V_{j'j'} = 1$, so we have
\begin{equation}
    \begin{aligned}
        &~~~~(\mathbb{E}_{\mathrm{Haar}}[((C(\vb*{\theta})^{-1})_{AA})^{-1}])_{jj'} \\
        &= -(\mathbb{E}_{\mathrm{Haar}}[((C(\vb*{\theta})^{-1})_{AA})^{-1}])_{jj'}.
    \end{aligned}
\end{equation}
This implies that the off-diagonal element ($j\neq j'$) is $0$, so $\mathbb{E}_{\mathrm{Haar}}[((C(\vb*{\theta})^{-1})_{AA})^{-1}]$ is diagonal.

Let us further decompose the set $\mathcal{J}_A$ as $\mathcal{J}_A = \mathcal{J}_{A_0}\sqcup\mathcal{J}_{A_1}$, where
\begin{equation}
    \begin{aligned}
        \mathcal{J}_{A_0} &= \{j\in\mathcal{J}_A ~|~ \bar{P}_j = \bar{P}_A \otimes \bar{I}_B, ~ \bar{P}_A \neq \bar{I}_A\}, \\
        \mathcal{J}_{A_1} &= \{j\in\mathcal{J}_A ~|~ \bar{P}_j = \bar{P}_A \otimes \bar{P}_B, ~ \bar{P}_A \neq \bar{I}_A, ~ \bar{P}_B \neq \bar{I}_B\}.
    \end{aligned}
\end{equation}
For arbitrary $j,j'\in\mathcal{J}_{A_0}$, there exists a Clifford unitary $\bar{U} = \bar{U}_A\otimes\bar{U}_B$ that satisfies $\bar{U}\bar{P}_j\bar{U}^\dagger = \bar{P}_{j'}$ and $\bar{U}\bar{P}_{j'}\bar{U}^\dagger = \bar{P}_{j}$, meaning that the corresponding Pauli transfer matrix satisfies $(V\vb*{\theta})_j = \theta_{j'}$ and $(V\vb*{\theta})_{j'} = \theta_{j}$.
Therefore, from Eq.~\eqref{eq_proof_q_2_1}, $(j,j)$-th and $(j',j')$-th element in $\mathbb{E}_{\mathrm{Haar}}[((C(\vb*{\theta})^{-1})_{AA})^{-1}]$ must be identical.
The same statement holds for $\mathcal{J}_{A_1}$.
Thus, we obtain
\begin{equation}
    \label{eq_proof_q_2_2}
    \mathbb{E}_{\mathrm{Haar}}[((C(\vb*{\theta})^{-1})_{AA})^{-1}] = \alpha_0 \Pi_{A_0} + \alpha_1\Pi_{A_1},
\end{equation}
where $\Pi_{A_0} = \sum_{j\in\mathcal{J}_{A_0}}\vb*{e}_j \vb*{e}_j^{\mathrm{T}}$ ($\Pi_{A_1} = \sum_{j\in\mathcal{J}_{A_1}}\vb*{e}_j \vb*{e}_j^{\mathrm{T}}$) represents the projector to the subsystems spanned by basis indexed by $\mathcal{J}_{A_0}$ ($\mathcal{J}_{A_1}$), and constants $\alpha_0$ and $\alpha_1$ are defined as
\begin{equation}
    \begin{aligned}
        \alpha_0 &= \frac{1}{\abs{\mathcal{J}_{A_0}}} \mathbb{E}_{\mathrm{Haar}}[\mathrm{tr}[((C(\vb*{\theta})^{-1})_{AA})^{-1}\Pi_{A_0}]], \\
        \alpha_1 &=  \frac{1}{\abs{\mathcal{J}_{A_1}}}  \mathbb{E}_{\mathrm{Haar}}[\mathrm{tr}[((C(\vb*{\theta})^{-1})_{AA})^{-1}\Pi_{A_1}]] .
    \end{aligned}
\end{equation}

\paragraph*{Step 2: evaluation of $\alpha_0$.}
Next, let us evaluate $\alpha_0$.
Let the Schmidt decomposition of the state $\ket{\psi}$ be
\begin{equation}
    \ket{\psi} = \sum_{a=1}^{d_A}\sqrt{\lambda_a}\ket{a}_A\ket{a}_B,
\end{equation}
where $d_A = 2^{k'}$ ($d_B=2^{k-k'}$) is the dimension of the subsystem $A$ ($B$) and $\lambda_a > 0$ are Schmidt coefficients.
Note that for Haar-random state $\ket{\psi}$, its Schmidt rank is almost surely $d_A$, since we assume $k' < k/2$ and thus $d_A < d_B$.
Using the Schmidt basis, we define the orthogonal basis of the traceless Hermitian matrix of subsystem $A$ as
\begin{equation}
    \begin{aligned}
        \bar{X}^{ab}_A &= \sqrt{\frac{d_A}{2}}(\ketbra{a}{b}_A + \ketbra{b}{a}_A),\\
        \bar{Y}^{ab}_A &= -i\sqrt{\frac{d_A}{2}}(\ketbra{a}{b}_A - \ketbra{b}{a}_A), \\
        \bar{Z}^{l}_A &= \sum_{a=1}^{d_A}h_{a,l}\ketbra{a}{a}_A. \\
    \end{aligned}
\end{equation}
Here, $1\leq a < b \leq d_A$, $l = 1,\ldots,d_A-1$, and coefficients $h_{a,l}$ are taken to satisfy
\begin{equation}
    \mathrm{tr}[\bar{Z}^{l}_A\bar{Z}^{l'}_A] = d_A\delta_{ll'}, ~~~~ \sum_a h_{a,l} = 0.
\end{equation}
Using the orthogonal basis, we define orthonormal basis for the subsystem projected by $\Pi_{A_0}$ as
\begin{equation}
    \begin{aligned}
        (\vb*{x}_A^{ab})^{\mathrm{T}}\bar{\vb*{P}}_{A} &= \bar{X}^{ab}_A \otimes \bar{I}_B, \\
        (\vb*{y}_A^{ab})^{\mathrm{T}}\bar{\vb*{P}}_{A} &= \bar{Y}^{ab}_A \otimes \bar{I}_B,\\
        (\vb*{z}_A^{l})^{\mathrm{T}}\bar{\vb*{P}}_{A} &= \bar{Z}^{l}_A  \otimes \bar{I}_B.
    \end{aligned}
\end{equation}
Then, the constant $\alpha_0$ is represented as
\begin{equation}
    \begin{aligned}
        \alpha_0 = \frac{1}{\abs{\mathcal{J}_{A_0}}}\mathbb{E}_{\mathrm{Haar}}\Bigg[&\sum_{a<b} (\vb*{x}_A^{ab})^{\mathrm{T}}((C(\vb*{\theta})^{-1})_{AA})^{-1}\vb*{x}_A^{ab} \\
        +&\sum_{a<b} (\vb*{y}_A^{ab})^{\mathrm{T}}((C(\vb*{\theta})^{-1})_{AA})^{-1}\vb*{y}_A^{ab} \\
        +&\sum_{l} (\vb*{z}_A^{l})^{\mathrm{T}}((C(\vb*{\theta})^{-1})_{AA})^{-1}\vb*{z}_A^{l} \Bigg].
    \end{aligned}
\end{equation}

Now, we calculate the terms in the above equation.
From Lemma~\ref{lem_3}, we have
\begin{equation}
    \begin{aligned}
        &~~~~(\vb*{x}_A^{ab})^{\mathrm{T}}((C(\vb*{\theta})^{-1})_{AA})^{-1}\vb*{x}_A^{ab} \\
        &\leq \min_{x\in\mathbb{R}} \mathrm{Var}_{\bar{\rho}(\vb*{\theta})} \qty(\bar{X}^{ab}_A \otimes \bar{I}_B + x \bar{I}_A \otimes \bar{X}^{ab}_B),
    \end{aligned}
\end{equation}
where we have defined
\begin{equation}
    \bar{X}^{ab}_B = \sqrt{\frac{d_B}{2}}(\ketbra{a}{b}_B + \ketbra{b}{a}_B).
\end{equation}
The variance and covariance for the operators $\bar{X}_{ab}^A \otimes \bar{I}_B$ and $\bar{I}_A \otimes \bar{X}_{ab}^B$ are represented as
\begin{equation}
    \begin{aligned}
        \mathrm{Var}_{\bar{\rho}(\vb*{\theta})} \qty(\bar{X}^{ab}_A \otimes \bar{I}_B) &= \frac{d_A}{2}(\lambda_a + \lambda_b), \\
        \mathrm{Var}_{\bar{\rho}(\vb*{\theta})} \qty(\bar{I}_A \otimes \bar{X}^{ab}_B) &= \frac{d_B}{2}(\lambda_a + \lambda_b), \\
        \mathrm{Cov}_{\bar{\rho}(\vb*{\theta})} \qty(\bar{X}^{ab}_A \otimes \bar{I}_B, \bar{I}_A \otimes \bar{X}^{ab}_B) &= \sqrt{d_Ad_B}\sqrt{\lambda_a\lambda_b}.\\
    \end{aligned}
\end{equation}
Therefore, 
\begin{equation}
    \begin{aligned}
        &~~~~(\vb*{x}_A^{ab})^{\mathrm{T}}((C(\vb*{\theta})^{-1})_{AA})^{-1}\vb*{x}_A^{ab} \\
        &\leq \min_{x\in\mathbb{R}} \mathrm{Var}_{\bar{\rho}(\vb*{\theta})} \qty(\bar{X}^{ab}_A \otimes \bar{I}_B + x \bar{I}_A \otimes \bar{X}^{ab}_B) \\
        &= \mathrm{Var}_{\bar{\rho}(\vb*{\theta})} \qty(\bar{X}^{ab}_A \otimes \bar{I}_B) - \frac{\mathrm{Cov}_{\bar{\rho}(\vb*{\theta})} \qty(\bar{X}^{ab}_A \otimes \bar{I}_B, \bar{I}_A \otimes \bar{X}^{ab}_B)^2}{\mathrm{Var}_{\bar{\rho}(\vb*{\theta})} \qty(\bar{I}_A \otimes \bar{X}^{ab}_B)} \\
        &= \frac{d_A}{2}\frac{(\lambda_a - \lambda_b)^2}{\lambda_a+\lambda_b}.
    \end{aligned}
\end{equation}
In the same manner, we obtain
\begin{equation}
    (\vb*{y}_A^{ab})^{\mathrm{T}}((C(\vb*{\theta})^{-1})_{AA})^{-1}\vb*{y}_A^{ab}
    \leq \frac{d_A}{2}\frac{(\lambda_a - \lambda_b)^2}{\lambda_a+\lambda_b}.
\end{equation}
Meanwhile, for the vector $\vb*{z}_A^{l}$, we have
\begin{equation}
    \begin{aligned}
        &~~~~(\vb*{z}_A^{l})^{\mathrm{T}}((C(\vb*{\theta})^{-1})_{AA})^{-1}\vb*{z}_A^{l} \\
        &\leq  \mathrm{Var}_{\bar{\rho}(\vb*{\theta})} \qty(\bar{Z}^{l}_A \otimes \bar{I}_B - \bar{I}_A \otimes \bar{Z}^{l}_B) \\
        &= 0,
    \end{aligned}
\end{equation}
where $\bar{Z}^{l}_B$ is defined as
\begin{equation}
    \bar{Z}^{l}_B = \sum_{a=1}^{d_A}h_{a,l}\ketbra{a}{a}_B.
\end{equation}
Thus, we obtain
\begin{equation}
    \label{eq_proof_q_2_3}
    \alpha_0 \leq \frac{d_A}{d_A^2-1}\mathbb{E}_{\mathrm{Haar}}\qty[\sum_{1\leq a < b \leq d_A}\frac{(\lambda_a - \lambda_b)^2}{\lambda_a+\lambda_b}].
\end{equation}

Let us define $u=\lambda_a+\lambda_b$ and $v=\lambda_a-\lambda_b$.
Then
\begin{equation}
    \left(\lambda_a-d_A^{-1}\right)^2+\left(\lambda_b-d_A^{-1}\right)^2
    = \frac{1}{2}\left(v^2+(u-2d_A^{-1})^2\right).
\end{equation}
Hence
\begin{equation}
    \begin{aligned}
        &~~~~ 2d_A\left[\left(\lambda_a-d_A^{-1}\right)^2+\left(\lambda_b-d_A^{-1}\right)^2\right] - \frac{(\lambda_a-\lambda_b)^2}{\lambda_a+\lambda_b}\\
        &=\frac{u(u-2d_A^{-1})^2+(u-d_A^{-1})v^2}{d_A^{-1}u}.
    \end{aligned}
\end{equation}
Since the numerator is a linear function of $v^2$ with $0\leq v^2\leq u^2$, its minimum is attained at $v^2=0$ or $v^2=u^2$.
At $v^2=0$, it is equal to $u(u-2d_A^{-1})^2\geq 0$.
At $v^2=u^2$, it is equal to
\begin{equation}
    u\qty(2(u-5d_A^{-1}/4)^2+7d_A^{-2}/8) \geq 0.
\end{equation}
Therefore,
\begin{equation}
    \frac{(\lambda_a-\lambda_b)^2}{\lambda_a+\lambda_b}
    \leq 2d_A\left[\left(\lambda_a-\frac{1}{d_A}\right)^2+\left(\lambda_b-\frac{1}{d_A}\right)^2\right].
\end{equation}
Combining this with Eq.~\eqref{eq_proof_q_2_3}, we obtain
\begin{equation}
    \begin{aligned}
        \alpha_0 
        &\leq \frac{2d_A^2(d_A-1)}{d_A^2-1}\mathbb{E}_{\mathrm{Haar}}\qty[\sum_a\left(\lambda_a-\frac{1}{d_A}\right)^2] \\
        &= \frac{2d_A^2(d_A-1)}{d_A^2-1}\mathbb{E}_{\mathrm{Haar}}\qty[\mathrm{tr}[\bar{\rho}_A^2] - \frac{1}{d_A}], \\
    \end{aligned}
\end{equation}
where $\bar{\rho}_A$ is the reduced density matrix of the subsystem $A$.
Since
\begin{equation}
    \mathbb{E}_{\mathrm{Haar}}\qty[\mathrm{tr}[\bar{\rho}_A^2]] = \frac{d_A+d_B}{d_Ad_B+1},
\end{equation}
(see Example~50 of Ref.~\cite{mele2024introduction},) we obtain 
\begin{equation}
    \label{eq_proof_q_2_6}
    \alpha_0 \leq \frac{2d_A(d_A-1)}{d_Ad_B+1} \leq \frac{2}{2^{k-2k'}}.
\end{equation}

\paragraph*{Step 3: Evaluation of $\alpha_1$.}
Finally, we evaluate $\alpha_1$.
For this purpose, let us evaluate the trace of $((C(\vb*{\theta})^{-1})_{AA})^{-1}$.
By considering the Schur complement, we obtain
\begin{equation}
    \begin{aligned}
        &~~~~\mathrm{tr}[((C(\vb*{\theta})^{-1})_{AA})^{-1}] \\
        &= \mathrm{tr}[C(\vb*{\theta})_{AA}] - \mathrm{tr}[C(\vb*{\theta})_{AB} (C(\vb*{\theta})_{BB})^{-1} C(\vb*{\theta})_{BA}].
    \end{aligned}
\end{equation}
Meanwhile, by considering $((C(\vb*{\theta}))^2)_{BB}$, we have
\begin{equation}
    (C(\vb*{\theta})_{BB})^2 + C(\vb*{\theta})_{BA}C(\vb*{\theta})_{AB} = ((C(\vb*{\theta}))^2)_{BB}.
\end{equation}
By multiplying $(C(\vb*{\theta})_{BB})^{-1}$ and taking the trace, we obtain
\begin{equation}
    \begin{aligned}
        &~~~~\mathrm{tr}[C(\vb*{\theta})_{AB} (C(\vb*{\theta})_{BB})^{-1} C(\vb*{\theta})_{BA}] \\
        &= -\mathrm{tr}[C(\vb*{\theta})_{BB}] + \mathrm{tr}[(C(\vb*{\theta})_{BB})^{-1}((C(\vb*{\theta}))^2)_{BB}],
    \end{aligned}
\end{equation}
which leads to
\begin{equation}
    \label{eq_proof_q_2_8}
    \begin{aligned}
        &~~~~\mathrm{tr}[((C(\vb*{\theta})^{-1})_{AA})^{-1}]\\
        &=\mathrm{tr}[C(\vb*{\theta})] - \mathrm{tr}[(C(\vb*{\theta})_{BB})^{-1}((C(\vb*{\theta}))^2)_{BB}].
    \end{aligned}
\end{equation}
For an accurate calculation of this value, we need to replace $\vb*{\theta}$ with $(1-\delta)\vb*{\theta}$ and take $\delta\to0$.
Under this treatment, we can regard the inverse matrix as a pseudo-inverse matrix by considering an eigenvalue decomposition of the matrices $C((1-\delta)\vb*{\theta})_{BB}$ and $((C((1-\delta)\vb*{\theta}))^2)_{BB}$ and taking $\delta\to0$.
From Lemma~\ref{lem_2}, we recall that the covariance matrix $C(\vb*{\theta})$ can be represented as
\begin{align}
    \label{eq_proof_q_2_9}
    C(\vb*{\theta}) &= 2^{k-1} \Pi_{\mathrm{Im}},
\end{align}
where $\Pi_{\mathrm{Im}}$ represents the projector onto the image of $C(\vb*{\theta})$, whose dimension is $2(2^k-1)$.
Therefore, we have
\begin{equation}
    \label{eq_proof_q_2_10}
    \begin{aligned}
        &~~~~\lim_{\delta\to0}\mathrm{tr}[(C((1-\delta)\vb*{\theta})_{BB})^{-1}((C((1-\delta)\vb*{\theta}))^2)_{BB}]\\
        &= 2^{k-1}\mathrm{tr}[((\Pi_{\mathrm{Im}})_{BB})^{\mathrm{+}}(\Pi_{\mathrm{Im}})_{BB}]\\
        &= 2^{k-1}\bigl(\mathrm{dim}(\mathrm{Im}(C(\vb*{\theta}))) - \mathrm{dim}(\mathrm{Im}(C(\vb*{\theta})) \cap \mathrm{Im}(\Pi_A))\bigr).
    \end{aligned}
\end{equation}
Here, $((\Pi_{\mathrm{Im}})_{BB})^{\mathrm{+}}$ is the pseudo-inverse matrix of $(\Pi_{\mathrm{Im}})_{BB}$, $\Pi_A=\sum_{j\in\mathcal{J}_{A}}\vb*{e}_j\vb*{e}_j^{\mathrm{T}}$ is a projector to the subsystem spanned by basis indexed by $\mathcal{J}_A$, $\mathrm{Im}(\cdot)$ represents the image of $\cdot$, and $\mathrm{dim}(\cdot)$ represents the dimension of a vector space $\cdot$.
The second equality follows from the fact that $\mathrm{tr}[((\Pi_{\mathrm{Im}})_{BB})^{\mathrm{+}}(\Pi_{\mathrm{Im}})_{BB}]$ equals the dimension of the image of the linear map $\vb*{o}\in\mathrm{Im}(C(\vb*{\theta})) \mapsto (I-\Pi_A)\vb*{o}$, together with the rank--nullity theorem.
From Eqs.~\eqref{eq_proof_q_2_8}, \eqref{eq_proof_q_2_9}, \eqref{eq_proof_q_2_10}, we obtain
\begin{equation}
    \label{eq_proof_q_2_4}
    \begin{aligned}
        &~~~~\mathrm{tr}[((C(\vb*{\theta})^{-1})_{AA})^{-1}] \\
        &=2^{k-1}\mathrm{dim}(\mathrm{Im}(C(\vb*{\theta})) \cap \mathrm{Im}(\Pi_A)).
    \end{aligned}
\end{equation}

Now, we analyze the dimension of $\mathrm{Im}(C(\vb*{\theta})) \cap \mathrm{Im}(\Pi_A)$.
From the proof of Lemma~\ref{lem_2}, we note that when $\vb*{o}\in\mathrm{Im}(C(\vb*{\theta}))$, the corresponding operator $\sum_i o_i\bar{P}_i$ is of the form
\begin{equation}
    \sum_i o_i\bar{P}_i = \ketbra{\psi}{w} + \ketbra{w}{\psi},
\end{equation}
where $\ket{w}\in\mathbb{C}^{2^k}$ with $\ket{w}\perp\ket{\psi}$.
Moreover, when $\vb*{o}\in\mathrm{Im}(\Pi_A)$, the operator $\sum_i o_i\bar{P}_i$ should satisfy
\begin{equation}
    \label{eq_proof_q_2_5}
    \mathrm{tr}_A\qty[\sum_i o_i\bar{P}_i] = 0.
\end{equation}
Denote 
\begin{equation}
    \ket{w} = \sum_{a=1}^{d_A}\ket{a}_A\ket{u_a}_B.
\end{equation}
Then Eq.~\eqref{eq_proof_q_2_5} becomes
\begin{equation}
    \sum_{a=1}^{d_A} \sqrt{\lambda_a}(\ketbra{a}{u_a}_B + \ketbra{u_a}{a}_B) = 0.
\end{equation}
This means that each $\ket{u_a}_B$ can be represented as a linear combination of the Schmidt basis $\{\ket{a}_B\}$.
Hence we may write
\begin{equation}
    \ket{u_a}_B = \sum_{b=1}^{d_A} \sqrt{\lambda_b} K_{ba} \ket{b}_B.
\end{equation}
Under this notation, Eq.~\eqref{eq_proof_q_2_5} is equivalent to
\begin{equation}
    K = -K^\dagger.
\end{equation}
Moreover, since $\ket{w}\perp\ket{\psi}$, we have
\begin{equation}
    \sum_{a=1}^{d_A} \lambda_a K_{aa} = 0.
\end{equation}
The space of anti-Hermitian $d_A\times d_A$ matrices has real dimension $d_A^2$, and the above equation imposes one independent real linear constraint.
Therefore,
\begin{equation}
    \mathrm{dim}(\mathrm{Im}(C(\vb*{\theta})) \cap \mathrm{Im}(\Pi_A))=d_A^2-1.
\end{equation}
Combining this with Eq.~\eqref{eq_proof_q_2_4} results in
\begin{equation}
    \mathrm{tr}[((C(\vb*{\theta})^{-1})_{AA})^{-1}] = 2^{k-1}(2^{2k'}-1).
\end{equation}

Now we go back to the analysis of $\alpha_1$.
From Eq.~\eqref{eq_proof_q_2_2}, we have
\begin{equation}
    \label{eq_proof_q_2_7}
    \alpha_1 \leq \frac{\mathbb{E}_{\mathrm{Haar}}\qty[\mathrm{tr}[((C(\vb*{\theta})^{-1})_{AA})^{-1}]]}{(2^{2k'}-1)(2^{2(k-k')}-1)} \leq \frac{1}{2^{k-2k'}}.
\end{equation}
By combining Eq.~\eqref{eq_proof_q_2_2}, Eq.~\eqref{eq_proof_q_2_6}, and Eq.~\eqref{eq_proof_q_2_7}, we obtain
\begin{equation}
    \mathbb{E}_{\mathrm{Haar}}[((C(\vb*{\theta})^{-1})_{AA})^{-1}] \leq \frac{2}{2^{k-2k'}} I.
\end{equation}
\end{proof}

\subsection{Proof of Lemma~\ref{lem_1}}
\label{sec_proof_quantum_3}
\begin{proof}
Let $\overline{\mathcal{N}}_s$ be a Pauli noise channel acting only on $k' < k/2$ logical qubits.
We label the subsystem of the noisy $k'$ qubits as $A$ and the remaining $k-k'$ qubits as $B$.
Let 
\begin{equation}
    \overline{\mathcal{N}}_{A}(\cdot) = \frac{\bar{I}_A}{2^{k'}}\otimes \mathrm{tr}_A[\cdot]
\end{equation}
be an erasure channel on the subsystem $A$.
Since $\overline{\mathcal{N}}_s$ acts only on subsystem $A$, we have
\begin{equation}
    \overline{\mathcal{N}}_{A}\circ\overline{\mathcal{N}}_s = \overline{\mathcal{N}}_{A}.
\end{equation}
Denote the Pauli transfer matrices of $\overline{\mathcal{N}}_s$ and $\overline{\mathcal{N}}_A$ as $\Lambda_s$ and $\Lambda_A$.
Since $\Lambda_{s}C(\Lambda_{s}\vb*{\theta})^{-1}\Lambda_{s}$ and $\Lambda_{A}C(\Lambda_{A}\vb*{\theta})^{-1}\Lambda_{A}$ are the quantum Fisher information matrices of $\overline{\mathcal{N}}_{s}(\bar{\rho}(\vb*{\theta}))$ and $\overline{\mathcal{N}}_{A}(\bar{\rho}(\vb*{\theta})) = \overline{\mathcal{N}}_{A}\circ\overline{\mathcal{N}}_{s}(\bar{\rho}(\vb*{\theta}))$ (see Appendix~\ref{sec_estimation_theory_multiqubit}), we obtain
\begin{equation}
    \Lambda_{s}C(\Lambda_{s}\vb*{\theta})^{-1}\Lambda_{s} \geq \Lambda_{A}C(\Lambda_{A}\vb*{\theta})^{-1}\Lambda_{A}
\end{equation}
from the monotonicity of the quantum Fisher information matrix.
Therefore,
\begin{equation}
    \Delta_i(\overline{\mathcal{N}}_s)
    \leq \frac{\vb*{e}_i^{\mathrm{T}}\qty(C(\vb*{\theta}) - C(\vb*{\theta})\Lambda_{A}C(\Lambda_{A}\vb*{\theta})^{-1}\Lambda_{A}C(\vb*{\theta}))\vb*{e}_i}{4}.
\end{equation}

As in the main text, let us separate the set of indices $\mathcal{J}=\{1,\ldots,4^k-1\}$ for nontrivial $k$-qubit Pauli operators $\bar{P}_j$ as $\mathcal{J}=\mathcal{J}_A\sqcup\mathcal{J}_B$, where
\begin{equation}
    \begin{aligned}
        \mathcal{J}_A &= \{j\in\mathcal{J} ~|~ \bar{P}_j = \bar{P}_A \otimes \bar{P}_B, ~ \bar{P}_A \neq \bar{I}_A\}, \\
        \mathcal{J}_B &= \{j\in\mathcal{J} ~|~ \bar{P}_j = \bar{I}_A \otimes \bar{P}_B\}.
    \end{aligned}
\end{equation}
In the following, the subscripts $A$ and $B$ on vectors and matrices refer to the blocks associated with the decomposition $\mathcal{J}=\mathcal{J}_A\sqcup\mathcal{J}_B$.
Let us represent the covariance matrix as
\begin{equation}
    C(\vb*{\theta})
    =
    \begin{pmatrix}
        C(\vb*{\theta})_{AA} & C(\vb*{\theta})_{AB} \\
        C(\vb*{\theta})_{BA} & C(\vb*{\theta})_{BB}
    \end{pmatrix}.
\end{equation}
The Pauli transfer matrix $\Lambda_{A}$ can be represented as
\begin{equation}
    \Lambda_{A}
    =
    \begin{pmatrix}
        0 & 0 \\
        0 & I
    \end{pmatrix}.
\end{equation}
Therefore, the covariance matrix $C(\Lambda_{A}\vb*{\theta})$ satisfies
\begin{equation}
    C(\Lambda_{A}\vb*{\theta})
    =
    \begin{pmatrix}
        C(\Lambda_{A}\vb*{\theta})_{AA} & 0 \\
        0 & C(\vb*{\theta})_{BB}
    \end{pmatrix}.
\end{equation}
Thus,
\begin{equation}
    \begin{aligned}
        &~~~~C(\vb*{\theta}) - C(\vb*{\theta})\Lambda_{A}C(\Lambda_{A}\vb*{\theta})^{-1}\Lambda_{A}C(\vb*{\theta}) \\
        &=
        \begin{pmatrix}
            C(\vb*{\theta})_{AA} - C(\vb*{\theta})_{AB}\bigl(C(\vb*{\theta})_{BB}\bigr)^{-1}C(\vb*{\theta})_{BA} & 0 \\
            0 & 0
        \end{pmatrix} \\
        &=
        \begin{pmatrix}
            \bigl((C(\vb*{\theta})^{-1})_{AA}\bigr)^{-1} & 0 \\
            0 & 0
        \end{pmatrix}.
    \end{aligned}
\end{equation}
Here, we have used the Schur complement in the last equality.
This yields
\begin{equation}
    \Delta_i(\overline{\mathcal{N}}_s)
    \leq
    \frac{1}{4}
    \vb*{e}_i^{\mathrm{T}}
    \begin{pmatrix}
        \bigl((C(\vb*{\theta})^{-1})_{AA}\bigr)^{-1} & 0 \\
        0 & 0
    \end{pmatrix}
    \vb*{e}_i.
\end{equation}

If $i\in\mathcal{J}_B$, the right-hand side is $0$, so the claim is trivial.
If $i\in\mathcal{J}_A$, then taking the Haar average and combining with Proposition~\ref{prop_3}, we obtain
\begin{equation}
    \mathbb{E}_{\mathrm{Haar}}[\Delta_i(\overline{\mathcal{N}}_s)] \leq \frac{1}{2^{k-2k'+1}}.
\end{equation}
\end{proof}

\subsection{Proof of Theorem~\ref{thm_2}}
\label{sec_proof_quantum_4}
\begin{proof}
    Since the code distance $d$ is assumed to be even, we have
    \begin{equation}
        \lim_{\eta\to0}\frac{\mathbb{E}_{\mathrm{Haar}}[\epsilon^{\mathrm{qSynd}}_i]}{\epsilon_i}
        = \lim_{\eta\to0} \frac{\sum_{s\in\mathcal{S}_{\Theta(1)}}p_s\mathbb{E}_{\mathrm{Haar}}[\Delta_i(\overline{\mathcal{N}}_s)]}{\sum_{s\in\mathcal{S}_{\Theta(1)}} p_s\epsilon_{i,s}}
    \end{equation}
    from Eq.~\eqref{eq_eff_err_rate_even}.
    For an $[[n=mn',k=mk',d]]$ stabilizer code composed of $m$ code blocks of a fixed $[[n',k',d]]$ stabilizer code, the ambiguous error syndrome $s\in\mathcal{S}_{\Theta(1)}$ with non-vanishing conditional logical error rate corresponds to a syndrome caused by a weight-$d/2$ physical Pauli error within a single block.
    This means that the conditional logical noise channel $\overline{\mathcal{N}}_s$ affects only a single code block, and the remaining $m-1$ code blocks remain unaffected.
    Therefore, from Lemma~\ref{lem_1}, we have
    \begin{equation}
        \mathbb{E}_{\mathrm{Haar}}[\Delta_i(\overline{\mathcal{N}}_s)]\leq\frac{1}{2^{k-2k'+1}} = \frac{1}{2^{(m-2)k'+1}}.
    \end{equation}
    Meanwhile, since the conditional noise channel $\overline{\mathcal{N}}_s$ affects only a single code block and the underlying block code is fixed, the corresponding conditional logical error rate $\epsilon_{i,s}$ is independent of $m$.
    In other words,
    \begin{equation}
        \epsilon_{i,s} = \Theta(1)
    \end{equation}
    as a function of $m$.
    Therefore, we obtain
    \begin{equation}
        \lim_{\eta\to0}\frac{\mathbb{E}_{\mathrm{Haar}}[\epsilon^{\mathrm{qSynd}}_i]}{\epsilon_i}
        = O\qty(\frac{1}{2^{(m-2)k'}})
        = O\qty(\frac{1}{2^{k'm}})
    \end{equation}
    as a function of $m$.
\end{proof}

\subsection{Quantum Fisher information analysis of the classical--quantum state in Eq.~(\ref{eq_cq_state_simplified})}
\label{sec_proof_quantum_5}
In this subsection, we analyze the quantum Fisher information of the classical--quantum state in Eq.~\eqref{eq_cq_state_simplified}:
\begin{equation}
    \label{eq_cq_state_simplified_2}
    (1-p) \ketbra{0} \otimes \bar{\rho}(\vb*{\theta}) + p \ketbra{1} \otimes \bar{\rho}(\Lambda_{A}\vb*{\theta}).
\end{equation}
Here, $\Lambda_{A}$ is a $(4^k-1)$-dimensional diagonal matrix whose $(j,j)$-th element is $0$ if $j\in\mathcal{J}_A$ and $1$ otherwise.
From Eq.~\eqref{eq_qfim_syndrome}, the quantum Fisher information matrix of this state is represented as
\begin{equation}
    J^{\mathrm{Synd}} = (1-p)C(\vb*{\theta})^{-1} + p\Lambda_{A}C(\Lambda_{A}\vb*{\theta})^{-1}\Lambda_{A}.
\end{equation}
Following the notation of Sec.~\ref{sec_quantum_3}, the inverse of the quantum Fisher information matrix can be written as
\begin{equation}
    (J^{\mathrm{Synd}})^{-1} 
    = C(\vb*{\theta}) + \frac{p}{1-p}
    \begin{pmatrix}
        \bigl((C(\vb*{\theta})^{-1})_{AA}\bigr)^{-1} & 0 \\
        0 & 0
    \end{pmatrix}.
\end{equation}
Here, we use the Woodbury matrix identity $(A+UCV)^{-1} = A^{-1}-A^{-1}U(C^{-1}+VA^{-1}U)^{-1}VA^{-1}$ together with the Schur complement identity
$\bigl((C(\vb*{\theta})^{-1})_{AA}\bigr)^{-1} = C(\vb*{\theta})_{AA} - C(\vb*{\theta})_{AB}(C(\vb*{\theta})_{BB})^{-1}C(\vb*{\theta})_{BA}$.
Meanwhile, from Eq.~\eqref{eq_QFIM_classical_quantum} and Eq.~\eqref{eq_SLD_multiqubit}, the SLD operator $\bar{L}_j$ of this state is given by
\begin{equation}
    \begin{aligned}
        &\ketbra{0}\otimes\vb*{e}_j^{\mathrm{T}}C(\vb*{\theta})^{-1}(\bar{\vb*{P}}-\vb*{\theta}\bar{I}) \\
        +\,&\ketbra{1}\otimes\vb*{e}_j^{\mathrm{T}}(C(\vb*{\theta})_{BB})^{-1}(\bar{\vb*{P}}_{B}-\vb*{\theta}_{B}\bar{I})
    \end{aligned}
\end{equation}
for $j\in\mathcal{J}_{B}$, and by
\begin{equation}
    \ketbra{0}\otimes\vb*{e}_j^{\mathrm{T}}C(\vb*{\theta})^{-1}(\bar{\vb*{P}}-\vb*{\theta}\bar{I})
\end{equation}
for $j\in\mathcal{J}_{A}$.
Here, $\bar{\vb*{P}}$ denotes the vector of all nontrivial Pauli operators $\bar{P}_j$, and $\bar{\vb*{P}}_{B}$ and $\vb*{\theta}_{B}$ denote the subvectors restricted to indices $j\in\mathcal{J}_{B}$.
Therefore, the optimal measurement basis defined through Eq.~\eqref{eq_optimal_measurement_basis} for estimating the target parameter $\theta_i$ with $i\in\mathcal{J}_{A}$ is
\begin{equation}
    \ketbra{0}\otimes\frac{1}{1-p}\bigl(\bar{P}_i-p\bar{O}_i(\vb*{\theta})\bigr) + \ketbra{1}\otimes \bar{O}_i(\vb*{\theta}),
\end{equation}
where
\begin{equation}
    \bar{O}_i(\vb*{\theta}) = \vb*{e}_i^{\mathrm{T}}C(\vb*{\theta})_{AB}(C(\vb*{\theta})_{BB})^{-1}\bar{\vb*{P}}_{B}.
\end{equation}
Note that we use $((C(\vb*{\theta})^{-1})_{AA})^{-1}(C(\vb*{\theta})^{-1})_{AB} = -C(\vb*{\theta})_{AB}(C(\vb*{\theta})_{BB})^{-1}$ from the Schur complement.

\section{Details on numerical demonstration} \label{app:numerics-details}
\label{sec_numerical_details}
In this section, we describe the details of the numerical simulations presented in Fig.~\ref{fig_numerics_finiteeta} in the main text.
The noise model is taken as the code capacity noise model of
independent single-qubit depolarizing noise with error rate \(\eta\). For simplicity, we focus on the case where each logical patch contains a single logical qubit since calculations are done for the bit-flip repetition code and the rotated surface code, while the extension to patches encoding multiple logical qubits is straightforward. 

For a given measured syndrome \(s\), we assume maximum-likelihood decoding (MLD) is performed.
Let
\(w_{s,g}\) denote the joint probability that the measured syndrome is $s$ and that, after the MLD recovery,  the residual logical
Pauli error is \(g\in\{I,X,Y,Z\}\). Thus 
\[
  p_s=\sum_g w_{s,g}
\]
is the probability of obtaining the syndrome $s$. 
For the $i$-th logical Pauli observable
\(P_i\), define the conditional attenuation factor as
\begin{align}
  \lambda_{s,i}
  &=
  \frac{1}{p_s}
  \sum_g \chi_i(g) w_{s,g}, \\
  \chi_i(g)&=
  \begin{cases}
    +1, & [P_i,g]=0,\\
    -1, & \{P_i,g\}=0.
  \end{cases}
  \label{eq:syndrome-attenuation}
\end{align}
The MLD frame is chosen so that the relevant attenuations are non-negative; 
thus \(\lambda_{s,i}\in[0,1]\) for the channels used below.  
In the case of syndrome-agnostic estimation, as explained in the main text, the
 logical error rate $\epsilon_i$ is related with attenuation factors as
\begin{align}
  \epsilon_i
  =
  \frac{1-\bar\lambda_i}{2},
  \qquad
  \bar\lambda_i
  =
  \sum_s p_s\lambda_{s,i}.
  \label{eq:mld-logical}
\end{align}

For syndrome-aware protocols, we consider grouping syndromes into a small number of classes \(c\) within the code block. Concretely, the grouping is done as follows:
\begin{itemize}
    \item {\it even-distance codes:} Three classes ($\mathcal{S}_{\Theta(1)}$, $\mathcal{S}_{\Theta(\eta^d)}$, and others).
    \item {\it odd-distance codes:} $d-1$ classes (according to error scaling $k$ where $\epsilon_{i, s} = \Theta(\eta^k)$).
\end{itemize}
For each class, we store
its probability \(p_c\) and its class-conditioned attenuation
as
\begin{align}
  p_c=\sum_{s\in c}p_s,
  \qquad
  \lambda_{c,i}
  =
  \frac{1}{p_c}\sum_{s\in c}p_s\lambda_{s,i}.
  \label{eq:class-attenuation}
\end{align}
Furthermore, for $m$ logical patches, we consider a joint class label ${\bf c}=(c_1,\ldots,c_m)$ which naturally extends the probability as $p_{\bf c} = \prod_{\ell=1}^m p_{c_l}$ and $  \lambda_{{\bf c},i}
  =
  \prod_{\ell:i_\ell\neq I}
  \lambda_{c_\ell,i_\ell}.$

This defines  a syndrome-conditioned logical Pauli
channel
\begin{align}
  \bar{\mathcal N}_{\bf c}\!\left(\bar\rho(\theta)\right)
  =
  \bar\rho(\Lambda_{\bf c}\theta)
  =
  \frac{1}{2^m}
  \left(
    \bar I+\sum_i\lambda_{{\bf c},i}\theta_i\bar P_i
  \right),
  \label{eq:class-logical-channel}
\end{align}
where
$
  \Lambda_{\bf c}
  =
  \mathrm{diag}\!\left(\lambda_{{\bf c},i}\right)_i
 $
is the diagonal Pauli transfer matrix of
\(\bar{\mathcal N}_{\bf c}\).  Let \(C(\Lambda_{\bf c}\theta)\) denote the
real Pauli covariance matrix of \(\bar\rho(\Lambda_{\bf c}\theta)\):
\begin{align}
  C(\Lambda_{\bf c}\theta)_{ij}
  =
  \mathrm{Tr}\!\left[
    \bar\rho(\Lambda_{\bf c}\theta)
    \frac{\{\bar P_i,\bar P_j\}}{2}
  \right]
  -
  (\lambda_{{\bf c},i}\theta_i)
  (\lambda_{{\bf c},j}\theta_j),
  \label{eq:attenuated-covariance}
\end{align}
which yields the syndrome-aware quantum Fisher information
matrix as
\begin{align}
  J^{\rm qSynd}(\bar\rho)
  =
  \sum_{\bf c}
  p_{\bf c}\,
  \Lambda_{\bf c}\,
  C(\Lambda_{\bf c}\theta)^{-1}
  \Lambda_{\bf c}.
  \label{eq:syndrome-qfi}
\end{align}

From the definition of the effective logical error rate as defined in the main text, we obtain
\begin{align}
  \epsilon_{i}^{\rm qSynd}(\bar\rho)
  &=
  \frac{1-\lambda_{i}^{\rm qSynd}(\bar\rho)}{2} \nonumber\\
  &=
  \frac{1}{2}
  \left(
    1-\frac{1}{\sqrt{\bigl[J^{\rm qSynd}(\bar\rho)^{-1}\bigr]_{ii}+\theta_{i}^2}}
  \right).
  \label{eq:effective-error}
\end{align}

Since the number of possible syndromes grows exponentially with the
number of stabilizer generators, exact evaluation of the class
probabilities \(p_{\bf c}\) and attenuations \(\lambda_{{\bf c},i}\) is practical only
for small codes.  For the rotated surface codes considered in the main
text, we therefore use one of the following three strategies, depending
on \(d\) and \(\eta\):
 \begin{enumerate}
  \item \emph{Exact enumeration.}
  The syndrome-conditioned logical channel is obtained by directly
  enumerating the relevant syndrome contributions.  

  \item \emph{Low-weight enumeration.}
  Syndromes corresponding to low-weight errors are enumerated. This is a good approximation at small $\eta$ when the behavior is dominated solely by the leading order.

  \item \emph{Sampled enumeration.}
  Syndromes are enumerated by performing Monte Carlo sampling on the physical errors. This is valid
  in the higher-\(\eta\) regime, where the syndrome classes are
  no longer dominated by low-weight events. 
\end{enumerate}

Table~\ref{tab:surface-source-selection} records where each method is used in
Fig.~\ref{fig_numerics_finiteeta} in the main text.  The important distinction is that low
\(\eta\) points are treated by exact or low-weight calculations, while
sampling is used only where rare low-weight classes no longer control the
dominant behavior. The sampling is done up to $10^6$ non-zero syndromes.

\begin{table}[ht]
  \centering
  \small
  \caption{Enumeration method used in Fig.~\ref{fig_numerics_finiteeta} in the main text.}
  \label{tab:surface-source-selection}
  \begin{tabular}{@{}c |p{0.34\linewidth}| p{0.48\linewidth}@{}}
    \toprule
    \(d\) & \(\eta\) region & method used \\
    \midrule
    2 & All & Exact \\
    3 & All & Exact \\
    4 & All & Exact \\
    5 & $ \leq 10^{-3}$ &
    Low-weight (\(\leq 4\)) \\
    5 & $> 10^{-3}$ & Sampled \\
    6 & $\leq 5\times 10^{-3}$ & Low-weight \((\leq 4)\) \\
    6 & $> 5\times 10^{-3}$ & Sampled \\
    7 & $\leq 5\times10^{-3}$ &
    Low-weight (\(\leq 6\)) \\
    7 & $>5\times 10^{-3}$ & Sampled \\
    \bottomrule
  \end{tabular}
\end{table}

\bibliography{bib.bib}

\begin{thebibliography}{64}%
\makeatletter
\providecommand \@ifxundefined [1]{%
 \@ifx{#1\undefined}
}%
\providecommand \@ifnum [1]{%
 \ifnum #1\expandafter \@firstoftwo
 \else \expandafter \@secondoftwo
 \fi
}%
\providecommand \@ifx [1]{%
 \ifx #1\expandafter \@firstoftwo
 \else \expandafter \@secondoftwo
 \fi
}%
\providecommand \natexlab [1]{#1}%
\providecommand \enquote  [1]{``#1''}%
\providecommand \bibnamefont  [1]{#1}%
\providecommand \bibfnamefont [1]{#1}%
\providecommand \citenamefont [1]{#1}%
\providecommand \href@noop [0]{\@secondoftwo}%
\providecommand \href [0]{\begingroup \@sanitize@url \@href}%
\providecommand \@href[1]{\@@startlink{#1}\@@href}%
\providecommand \@@href[1]{\endgroup#1\@@endlink}%
\providecommand \@sanitize@url [0]{\catcode `\\12\catcode `\$12\catcode `\&12\catcode `\#12\catcode `\^12\catcode `\_12\catcode `\%12\relax}%
\providecommand \@@startlink[1]{}%
\providecommand \@@endlink[0]{}%
\providecommand \url  [0]{\begingroup\@sanitize@url \@url }%
\providecommand \@url [1]{\endgroup\@href {#1}{\urlprefix }}%
\providecommand \urlprefix  [0]{URL }%
\providecommand \Eprint [0]{\href }%
\providecommand \doibase [0]{https://doi.org/}%
\providecommand \selectlanguage [0]{\@gobble}%
\providecommand \bibinfo  [0]{\@secondoftwo}%
\providecommand \bibfield  [0]{\@secondoftwo}%
\providecommand \translation [1]{[#1]}%
\providecommand \BibitemOpen [0]{}%
\providecommand \bibitemStop [0]{}%
\providecommand \bibitemNoStop [0]{.\EOS\space}%
\providecommand \EOS [0]{\spacefactor3000\relax}%
\providecommand \BibitemShut  [1]{\csname bibitem#1\endcsname}%
\let\auto@bib@innerbib\@empty
\bibitem [{\citenamefont {Shor}(1995)}]{shor1995scheme}%
  \BibitemOpen
  \bibfield  {author} {\bibinfo {author} {\bibfnamefont {P.~W.}\ \bibnamefont {Shor}},\ }\bibfield  {title} {\bibinfo {title} {Scheme for reducing decoherence in quantum computer memory},\ }\href {https://doi.org/10.1103/PhysRevA.52.R2493} {\bibfield  {journal} {\bibinfo  {journal} {Phys. Rev. A}\ }\textbf {\bibinfo {volume} {52}},\ \bibinfo {pages} {R2493} (\bibinfo {year} {1995})}\BibitemShut {NoStop}%
\bibitem [{\citenamefont {Knill}\ \emph {et~al.}(1996)\citenamefont {Knill}, \citenamefont {Laflamme},\ and\ \citenamefont {Zurek}}]{knill1996threshold}%
  \BibitemOpen
  \bibfield  {author} {\bibinfo {author} {\bibfnamefont {E.}~\bibnamefont {Knill}}, \bibinfo {author} {\bibfnamefont {R.}~\bibnamefont {Laflamme}},\ and\ \bibinfo {author} {\bibfnamefont {W.}~\bibnamefont {Zurek}},\ }\bibfield  {title} {\bibinfo {title} {Threshold accuracy for quantum computation},\ }\href {https://doi.org/10.48550/arXiv.quant-ph/9610011} {\bibfield  {journal} {\bibinfo  {journal} {arXiv:quant-ph/9610011}\ } (\bibinfo {year} {1996})}\BibitemShut {NoStop}%
\bibitem [{\citenamefont {Aharonov}\ and\ \citenamefont {Ben-Or}(1997)}]{aharonov1997fault}%
  \BibitemOpen
  \bibfield  {author} {\bibinfo {author} {\bibfnamefont {D.}~\bibnamefont {Aharonov}}\ and\ \bibinfo {author} {\bibfnamefont {M.}~\bibnamefont {Ben-Or}},\ }\bibfield  {title} {\bibinfo {title} {Fault-tolerant quantum computation with constant error},\ }in\ \href {https://doi.org/10.1145/258533.258579} {\emph {\bibinfo {booktitle} {Proc. 29th Annu. ACM Symp. Theory Comput.}}}\ (\bibinfo  {publisher} {ACM},\ \bibinfo {address} {El Paso, Texas, USA},\ \bibinfo {year} {1997})\ pp.\ \bibinfo {pages} {176--188}\BibitemShut {NoStop}%
\bibitem [{\citenamefont {Lidar}\ and\ \citenamefont {Brun}(2013)}]{lidar2013quantum}%
  \BibitemOpen
  \bibfield  {author} {\bibinfo {author} {\bibfnamefont {D.~A.}\ \bibnamefont {Lidar}}\ and\ \bibinfo {author} {\bibfnamefont {T.~A.}\ \bibnamefont {Brun}},\ }\href {https://doi.org/10.1017/CBO9781139034807} {\emph {\bibinfo {title} {Quantum error correction}}}\ (\bibinfo  {publisher} {Cambridge university press},\ \bibinfo {address} {Cambridge, UK},\ \bibinfo {year} {2013})\BibitemShut {NoStop}%
\bibitem [{\citenamefont {Ofek}\ \emph {et~al.}(2016)\citenamefont {Ofek}, \citenamefont {Petrenko}, \citenamefont {Heeres}, \citenamefont {Reinhold}, \citenamefont {Leghtas}, \citenamefont {Vlastakis}, \citenamefont {Liu}, \citenamefont {Frunzio}, \citenamefont {Girvin}, \citenamefont {Jiang} \emph {et~al.}}]{ofek2016extending}%
  \BibitemOpen
  \bibfield  {author} {\bibinfo {author} {\bibfnamefont {N.}~\bibnamefont {Ofek}}, \bibinfo {author} {\bibfnamefont {A.}~\bibnamefont {Petrenko}}, \bibinfo {author} {\bibfnamefont {R.}~\bibnamefont {Heeres}}, \bibinfo {author} {\bibfnamefont {P.}~\bibnamefont {Reinhold}}, \bibinfo {author} {\bibfnamefont {Z.}~\bibnamefont {Leghtas}}, \bibinfo {author} {\bibfnamefont {B.}~\bibnamefont {Vlastakis}}, \bibinfo {author} {\bibfnamefont {Y.}~\bibnamefont {Liu}}, \bibinfo {author} {\bibfnamefont {L.}~\bibnamefont {Frunzio}}, \bibinfo {author} {\bibfnamefont {S.}~\bibnamefont {Girvin}}, \bibinfo {author} {\bibfnamefont {L.}~\bibnamefont {Jiang}}, \emph {et~al.},\ }\bibfield  {title} {\bibinfo {title} {Extending the lifetime of a quantum bit with error correction in superconducting circuits},\ }\href {https://doi.org/10.1038/nature18949} {\bibfield  {journal} {\bibinfo  {journal} {Nature}\ }\textbf {\bibinfo {volume} {536}},\ \bibinfo {pages} {441} (\bibinfo {year} {2016})}\BibitemShut {NoStop}%
\bibitem [{\citenamefont {Krinner}\ \emph {et~al.}(2022)\citenamefont {Krinner}, \citenamefont {Lacroix}, \citenamefont {Remm}, \citenamefont {Di~Paolo}, \citenamefont {Genois}, \citenamefont {Leroux}, \citenamefont {Hellings}, \citenamefont {Lazar}, \citenamefont {Swiadek}, \citenamefont {Herrmann} \emph {et~al.}}]{krinner2022realizing}%
  \BibitemOpen
  \bibfield  {author} {\bibinfo {author} {\bibfnamefont {S.}~\bibnamefont {Krinner}}, \bibinfo {author} {\bibfnamefont {N.}~\bibnamefont {Lacroix}}, \bibinfo {author} {\bibfnamefont {A.}~\bibnamefont {Remm}}, \bibinfo {author} {\bibfnamefont {A.}~\bibnamefont {Di~Paolo}}, \bibinfo {author} {\bibfnamefont {E.}~\bibnamefont {Genois}}, \bibinfo {author} {\bibfnamefont {C.}~\bibnamefont {Leroux}}, \bibinfo {author} {\bibfnamefont {C.}~\bibnamefont {Hellings}}, \bibinfo {author} {\bibfnamefont {S.}~\bibnamefont {Lazar}}, \bibinfo {author} {\bibfnamefont {F.}~\bibnamefont {Swiadek}}, \bibinfo {author} {\bibfnamefont {J.}~\bibnamefont {Herrmann}}, \emph {et~al.},\ }\bibfield  {title} {\bibinfo {title} {Realizing repeated quantum error correction in a distance-three surface code},\ }\href {https://doi.org/10.1038/s41586-022-04566-8} {\bibfield  {journal} {\bibinfo  {journal} {Nature}\ }\textbf {\bibinfo {volume} {605}},\ \bibinfo {pages} {669} (\bibinfo {year} {2022})}\BibitemShut {NoStop}%
\bibitem [{\citenamefont {{Google Quantum AI}}(2023)}]{google2023suppressing}%
  \BibitemOpen
  \bibfield  {author} {\bibinfo {author} {\bibnamefont {{Google Quantum AI}}},\ }\bibfield  {title} {\bibinfo {title} {Suppressing quantum errors by scaling a surface code logical qubit},\ }\href {https://doi.org/10.1038/s41586-022-05434-1} {\bibfield  {journal} {\bibinfo  {journal} {Nature}\ }\textbf {\bibinfo {volume} {614}},\ \bibinfo {pages} {676} (\bibinfo {year} {2023})}\BibitemShut {NoStop}%
\bibitem [{\citenamefont {Sivak}\ \emph {et~al.}(2023)\citenamefont {Sivak}, \citenamefont {Eickbusch}, \citenamefont {Royer}, \citenamefont {Singh}, \citenamefont {Tsioutsios}, \citenamefont {Ganjam}, \citenamefont {Miano}, \citenamefont {Brock}, \citenamefont {Ding}, \citenamefont {Frunzio} \emph {et~al.}}]{sivak2023real}%
  \BibitemOpen
  \bibfield  {author} {\bibinfo {author} {\bibfnamefont {V.}~\bibnamefont {Sivak}}, \bibinfo {author} {\bibfnamefont {A.}~\bibnamefont {Eickbusch}}, \bibinfo {author} {\bibfnamefont {B.}~\bibnamefont {Royer}}, \bibinfo {author} {\bibfnamefont {S.}~\bibnamefont {Singh}}, \bibinfo {author} {\bibfnamefont {I.}~\bibnamefont {Tsioutsios}}, \bibinfo {author} {\bibfnamefont {S.}~\bibnamefont {Ganjam}}, \bibinfo {author} {\bibfnamefont {A.}~\bibnamefont {Miano}}, \bibinfo {author} {\bibfnamefont {B.}~\bibnamefont {Brock}}, \bibinfo {author} {\bibfnamefont {A.}~\bibnamefont {Ding}}, \bibinfo {author} {\bibfnamefont {L.}~\bibnamefont {Frunzio}}, \emph {et~al.},\ }\bibfield  {title} {\bibinfo {title} {Real-time quantum error correction beyond break-even},\ }\href {https://doi.org/https://doi.org/10.1038/s41586-023-05782-6} {\bibfield  {journal} {\bibinfo  {journal} {Nature}\ }\textbf {\bibinfo {volume} {616}},\ \bibinfo {pages} {50} (\bibinfo {year} {2023})}\BibitemShut {NoStop}%
\bibitem [{\citenamefont {Bluvstein}\ \emph {et~al.}(2024)\citenamefont {Bluvstein}, \citenamefont {Evered}, \citenamefont {Geim}, \citenamefont {Li}, \citenamefont {Zhou}, \citenamefont {Manovitz}, \citenamefont {Ebadi}, \citenamefont {Cain}, \citenamefont {Kalinowski}, \citenamefont {Hangleiter} \emph {et~al.}}]{bluvstein2024logical}%
  \BibitemOpen
  \bibfield  {author} {\bibinfo {author} {\bibfnamefont {D.}~\bibnamefont {Bluvstein}}, \bibinfo {author} {\bibfnamefont {S.~J.}\ \bibnamefont {Evered}}, \bibinfo {author} {\bibfnamefont {A.~A.}\ \bibnamefont {Geim}}, \bibinfo {author} {\bibfnamefont {S.~H.}\ \bibnamefont {Li}}, \bibinfo {author} {\bibfnamefont {H.}~\bibnamefont {Zhou}}, \bibinfo {author} {\bibfnamefont {T.}~\bibnamefont {Manovitz}}, \bibinfo {author} {\bibfnamefont {S.}~\bibnamefont {Ebadi}}, \bibinfo {author} {\bibfnamefont {M.}~\bibnamefont {Cain}}, \bibinfo {author} {\bibfnamefont {M.}~\bibnamefont {Kalinowski}}, \bibinfo {author} {\bibfnamefont {D.}~\bibnamefont {Hangleiter}}, \emph {et~al.},\ }\bibfield  {title} {\bibinfo {title} {Logical quantum processor based on reconfigurable atom arrays},\ }\href {https://doi.org/https://doi.org/10.1038/s41586-023-06927-3} {\bibfield  {journal} {\bibinfo  {journal} {Nature}\ }\textbf {\bibinfo {volume} {626}},\ \bibinfo {pages} {58} (\bibinfo {year} {2024})}\BibitemShut {NoStop}%
\bibitem [{\citenamefont {{Google Quantum AI and Collaborators}}(2024)}]{ai2024quantum}%
  \BibitemOpen
  \bibfield  {author} {\bibinfo {author} {\bibnamefont {{Google Quantum AI and Collaborators}}},\ }\bibfield  {title} {\bibinfo {title} {Quantum error correction below the surface code threshold},\ }\href {https://doi.org/https://doi.org/10.1038/s41586-024-08449-y} {\bibfield  {journal} {\bibinfo  {journal} {Nature}\ }\textbf {\bibinfo {volume} {638}},\ \bibinfo {pages} {920} (\bibinfo {year} {2024})}\BibitemShut {NoStop}%
\bibitem [{\citenamefont {Temme}\ \emph {et~al.}(2017)\citenamefont {Temme}, \citenamefont {Bravyi},\ and\ \citenamefont {Gambetta}}]{temme2017error}%
  \BibitemOpen
  \bibfield  {author} {\bibinfo {author} {\bibfnamefont {K.}~\bibnamefont {Temme}}, \bibinfo {author} {\bibfnamefont {S.}~\bibnamefont {Bravyi}},\ and\ \bibinfo {author} {\bibfnamefont {J.~M.}\ \bibnamefont {Gambetta}},\ }\bibfield  {title} {\bibinfo {title} {Error mitigation for short-depth quantum circuits},\ }\href {https://doi.org/10.1103/PhysRevLett.119.180509} {\bibfield  {journal} {\bibinfo  {journal} {Phys. Rev. Lett.}\ }\textbf {\bibinfo {volume} {119}},\ \bibinfo {pages} {180509} (\bibinfo {year} {2017})}\BibitemShut {NoStop}%
\bibitem [{\citenamefont {Endo}\ \emph {et~al.}(2021)\citenamefont {Endo}, \citenamefont {Cai}, \citenamefont {Benjamin},\ and\ \citenamefont {Yuan}}]{endo2021hybrid}%
  \BibitemOpen
  \bibfield  {author} {\bibinfo {author} {\bibfnamefont {S.}~\bibnamefont {Endo}}, \bibinfo {author} {\bibfnamefont {Z.}~\bibnamefont {Cai}}, \bibinfo {author} {\bibfnamefont {S.~C.}\ \bibnamefont {Benjamin}},\ and\ \bibinfo {author} {\bibfnamefont {X.}~\bibnamefont {Yuan}},\ }\bibfield  {title} {\bibinfo {title} {Hybrid quantum-classical algorithms and quantum error mitigation},\ }\href {https://doi.org/https://doi.org/10.7566/JPSJ.90.032001} {\bibfield  {journal} {\bibinfo  {journal} {J. Phys. Soc. Japan}\ }\textbf {\bibinfo {volume} {90}},\ \bibinfo {pages} {032001} (\bibinfo {year} {2021})}\BibitemShut {NoStop}%
\bibitem [{\citenamefont {Cai}\ \emph {et~al.}(2023)\citenamefont {Cai}, \citenamefont {Babbush}, \citenamefont {Benjamin}, \citenamefont {Endo}, \citenamefont {Huggins}, \citenamefont {Li}, \citenamefont {McClean},\ and\ \citenamefont {O'Brien}}]{cai2023quantum}%
  \BibitemOpen
  \bibfield  {author} {\bibinfo {author} {\bibfnamefont {Z.}~\bibnamefont {Cai}}, \bibinfo {author} {\bibfnamefont {R.}~\bibnamefont {Babbush}}, \bibinfo {author} {\bibfnamefont {S.~C.}\ \bibnamefont {Benjamin}}, \bibinfo {author} {\bibfnamefont {S.}~\bibnamefont {Endo}}, \bibinfo {author} {\bibfnamefont {W.~J.}\ \bibnamefont {Huggins}}, \bibinfo {author} {\bibfnamefont {Y.}~\bibnamefont {Li}}, \bibinfo {author} {\bibfnamefont {J.~R.}\ \bibnamefont {McClean}},\ and\ \bibinfo {author} {\bibfnamefont {T.~E.}\ \bibnamefont {O'Brien}},\ }\bibfield  {title} {\bibinfo {title} {Quantum error mitigation},\ }\href {https://doi.org/10.1103/RevModPhys.95.045005} {\bibfield  {journal} {\bibinfo  {journal} {Rev. Mod. Phys.}\ }\textbf {\bibinfo {volume} {95}},\ \bibinfo {pages} {045005} (\bibinfo {year} {2023})}\BibitemShut {NoStop}%
\bibitem [{\citenamefont {Kim}\ \emph {et~al.}(2023)\citenamefont {Kim}, \citenamefont {Eddins}, \citenamefont {Anand}, \citenamefont {Wei}, \citenamefont {Van Den~Berg}, \citenamefont {Rosenblatt}, \citenamefont {Nayfeh}, \citenamefont {Wu}, \citenamefont {Zaletel}, \citenamefont {Temme} \emph {et~al.}}]{kim2023evidence}%
  \BibitemOpen
  \bibfield  {author} {\bibinfo {author} {\bibfnamefont {Y.}~\bibnamefont {Kim}}, \bibinfo {author} {\bibfnamefont {A.}~\bibnamefont {Eddins}}, \bibinfo {author} {\bibfnamefont {S.}~\bibnamefont {Anand}}, \bibinfo {author} {\bibfnamefont {K.~X.}\ \bibnamefont {Wei}}, \bibinfo {author} {\bibfnamefont {E.}~\bibnamefont {Van Den~Berg}}, \bibinfo {author} {\bibfnamefont {S.}~\bibnamefont {Rosenblatt}}, \bibinfo {author} {\bibfnamefont {H.}~\bibnamefont {Nayfeh}}, \bibinfo {author} {\bibfnamefont {Y.}~\bibnamefont {Wu}}, \bibinfo {author} {\bibfnamefont {M.}~\bibnamefont {Zaletel}}, \bibinfo {author} {\bibfnamefont {K.}~\bibnamefont {Temme}}, \emph {et~al.},\ }\bibfield  {title} {\bibinfo {title} {Evidence for the utility of quantum computing before fault tolerance},\ }\href {https://doi.org/https://doi.org/10.1038/s41586-023-06096-3} {\bibfield  {journal} {\bibinfo  {journal} {Nature}\ }\textbf {\bibinfo {volume} {618}},\ \bibinfo {pages} {500} (\bibinfo {year} {2023})}\BibitemShut {NoStop}%
\bibitem [{\citenamefont {Piveteau}\ \emph {et~al.}(2021)\citenamefont {Piveteau}, \citenamefont {Sutter}, \citenamefont {Bravyi}, \citenamefont {Gambetta},\ and\ \citenamefont {Temme}}]{piveteau2021error}%
  \BibitemOpen
  \bibfield  {author} {\bibinfo {author} {\bibfnamefont {C.}~\bibnamefont {Piveteau}}, \bibinfo {author} {\bibfnamefont {D.}~\bibnamefont {Sutter}}, \bibinfo {author} {\bibfnamefont {S.}~\bibnamefont {Bravyi}}, \bibinfo {author} {\bibfnamefont {J.~M.}\ \bibnamefont {Gambetta}},\ and\ \bibinfo {author} {\bibfnamefont {K.}~\bibnamefont {Temme}},\ }\bibfield  {title} {\bibinfo {title} {Error mitigation for universal gates on encoded qubits},\ }\href {https://doi.org/10.1103/PhysRevLett.127.200505} {\bibfield  {journal} {\bibinfo  {journal} {Phys. Rev. Lett.}\ }\textbf {\bibinfo {volume} {127}},\ \bibinfo {pages} {200505} (\bibinfo {year} {2021})}\BibitemShut {NoStop}%
\bibitem [{\citenamefont {Lostaglio}\ and\ \citenamefont {Ciani}(2021)}]{lostaglio2021error}%
  \BibitemOpen
  \bibfield  {author} {\bibinfo {author} {\bibfnamefont {M.}~\bibnamefont {Lostaglio}}\ and\ \bibinfo {author} {\bibfnamefont {A.}~\bibnamefont {Ciani}},\ }\bibfield  {title} {\bibinfo {title} {Error mitigation and quantum-assisted simulation in the error corrected regime},\ }\href {https://doi.org/10.1103/PhysRevLett.127.200506} {\bibfield  {journal} {\bibinfo  {journal} {Phys. Rev. Lett.}\ }\textbf {\bibinfo {volume} {127}},\ \bibinfo {pages} {200506} (\bibinfo {year} {2021})}\BibitemShut {NoStop}%
\bibitem [{\citenamefont {Suzuki}\ \emph {et~al.}(2022)\citenamefont {Suzuki}, \citenamefont {Endo}, \citenamefont {Fujii},\ and\ \citenamefont {Tokunaga}}]{suzuki2022quantum}%
  \BibitemOpen
  \bibfield  {author} {\bibinfo {author} {\bibfnamefont {Y.}~\bibnamefont {Suzuki}}, \bibinfo {author} {\bibfnamefont {S.}~\bibnamefont {Endo}}, \bibinfo {author} {\bibfnamefont {K.}~\bibnamefont {Fujii}},\ and\ \bibinfo {author} {\bibfnamefont {Y.}~\bibnamefont {Tokunaga}},\ }\bibfield  {title} {\bibinfo {title} {Quantum error mitigation as a universal error reduction technique: Applications from the nisq to the fault-tolerant quantum computing eras},\ }\href {https://doi.org/10.1103/PRXQuantum.3.010345} {\bibfield  {journal} {\bibinfo  {journal} {PRX Quantum}\ }\textbf {\bibinfo {volume} {3}},\ \bibinfo {pages} {010345} (\bibinfo {year} {2022})}\BibitemShut {NoStop}%
\bibitem [{\citenamefont {Tsubouchi}\ \emph {et~al.}(2025{\natexlab{a}})\citenamefont {Tsubouchi}, \citenamefont {Mitsuhashi}, \citenamefont {Sharma},\ and\ \citenamefont {Yoshioka}}]{tsubouchi2024symmetric}%
  \BibitemOpen
  \bibfield  {author} {\bibinfo {author} {\bibfnamefont {K.}~\bibnamefont {Tsubouchi}}, \bibinfo {author} {\bibfnamefont {Y.}~\bibnamefont {Mitsuhashi}}, \bibinfo {author} {\bibfnamefont {K.}~\bibnamefont {Sharma}},\ and\ \bibinfo {author} {\bibfnamefont {N.}~\bibnamefont {Yoshioka}},\ }\bibfield  {title} {\bibinfo {title} {Symmetric clifford twirling for cost-optimal quantum error mitigation in early ftqc regime},\ }\href {https://doi.org/https://doi.org/10.1038/s41534-025-01050-9} {\bibfield  {journal} {\bibinfo  {journal} {npj Quantum Information}\ }\textbf {\bibinfo {volume} {11}},\ \bibinfo {pages} {104} (\bibinfo {year} {2025}{\natexlab{a}})}\BibitemShut {NoStop}%
\bibitem [{\citenamefont {Bomb\'{\i}n}\ \emph {et~al.}(2024)\citenamefont {Bomb\'{\i}n}, \citenamefont {Pant}, \citenamefont {Roberts},\ and\ \citenamefont {Seetharam}}]{bombin2024fault}%
  \BibitemOpen
  \bibfield  {author} {\bibinfo {author} {\bibfnamefont {H.}~\bibnamefont {Bomb\'{\i}n}}, \bibinfo {author} {\bibfnamefont {M.}~\bibnamefont {Pant}}, \bibinfo {author} {\bibfnamefont {S.}~\bibnamefont {Roberts}},\ and\ \bibinfo {author} {\bibfnamefont {K.~I.}\ \bibnamefont {Seetharam}},\ }\bibfield  {title} {\bibinfo {title} {Fault-tolerant postselection for low-overhead magic state preparation},\ }\href {https://doi.org/10.1103/PRXQuantum.5.010302} {\bibfield  {journal} {\bibinfo  {journal} {PRX Quantum}\ }\textbf {\bibinfo {volume} {5}},\ \bibinfo {pages} {010302} (\bibinfo {year} {2024})}\BibitemShut {NoStop}%
\bibitem [{\citenamefont {Meister}\ \emph {et~al.}(2024)\citenamefont {Meister}, \citenamefont {Pattison},\ and\ \citenamefont {Preskill}}]{meister2024efficient}%
  \BibitemOpen
  \bibfield  {author} {\bibinfo {author} {\bibfnamefont {N.}~\bibnamefont {Meister}}, \bibinfo {author} {\bibfnamefont {C.~A.}\ \bibnamefont {Pattison}},\ and\ \bibinfo {author} {\bibfnamefont {J.}~\bibnamefont {Preskill}},\ }\bibfield  {title} {\bibinfo {title} {Efficient soft-output decoders for the surface code},\ }\href {https://doi.org/10.48550/arXiv.2405.07433} {\bibfield  {journal} {\bibinfo  {journal} {arXiv:2405.07433}\ } (\bibinfo {year} {2024})}\BibitemShut {NoStop}%
\bibitem [{\citenamefont {Smith}\ \emph {et~al.}(2024)\citenamefont {Smith}, \citenamefont {Brown},\ and\ \citenamefont {Bartlett}}]{smith2024mitigating}%
  \BibitemOpen
  \bibfield  {author} {\bibinfo {author} {\bibfnamefont {S.~C.}\ \bibnamefont {Smith}}, \bibinfo {author} {\bibfnamefont {B.~J.}\ \bibnamefont {Brown}},\ and\ \bibinfo {author} {\bibfnamefont {S.~D.}\ \bibnamefont {Bartlett}},\ }\bibfield  {title} {\bibinfo {title} {Mitigating errors in logical qubits},\ }\href {https://doi.org/10.1038/s42005-024-01883-4} {\bibfield  {journal} {\bibinfo  {journal} {Commun. Phys.}\ }\textbf {\bibinfo {volume} {7}},\ \bibinfo {pages} {386} (\bibinfo {year} {2024})}\BibitemShut {NoStop}%
\bibitem [{\citenamefont {Gidney}\ \emph {et~al.}(2025)\citenamefont {Gidney}, \citenamefont {Newman}, \citenamefont {Brooks},\ and\ \citenamefont {Jones}}]{gidney2025yoked}%
  \BibitemOpen
  \bibfield  {author} {\bibinfo {author} {\bibfnamefont {C.}~\bibnamefont {Gidney}}, \bibinfo {author} {\bibfnamefont {M.}~\bibnamefont {Newman}}, \bibinfo {author} {\bibfnamefont {P.}~\bibnamefont {Brooks}},\ and\ \bibinfo {author} {\bibfnamefont {C.}~\bibnamefont {Jones}},\ }\bibfield  {title} {\bibinfo {title} {Yoked surface codes},\ }\href {https://doi.org/https://doi.org/10.1038/s41467-025-59714-1} {\bibfield  {journal} {\bibinfo  {journal} {Nat. Commun.}\ }\textbf {\bibinfo {volume} {16}},\ \bibinfo {pages} {4498} (\bibinfo {year} {2025})}\BibitemShut {NoStop}%
\bibitem [{\citenamefont {Lee}\ \emph {et~al.}(2025)\citenamefont {Lee}, \citenamefont {English},\ and\ \citenamefont {Bartlett}}]{lee2025efficient}%
  \BibitemOpen
  \bibfield  {author} {\bibinfo {author} {\bibfnamefont {S.-H.}\ \bibnamefont {Lee}}, \bibinfo {author} {\bibfnamefont {L.}~\bibnamefont {English}},\ and\ \bibinfo {author} {\bibfnamefont {S.~D.}\ \bibnamefont {Bartlett}},\ }\bibfield  {title} {\bibinfo {title} {Efficient post-selection for general quantum ldpc codes},\ }\href {https://doi.org/10.48550/arXiv.2510.05795} {\bibfield  {journal} {\bibinfo  {journal} {arXiv:2510.05795}\ } (\bibinfo {year} {2025})}\BibitemShut {NoStop}%
\bibitem [{\citenamefont {Xie}\ \emph {et~al.}(2026)\citenamefont {Xie}, \citenamefont {Yoshioka}, \citenamefont {Tsubouchi},\ and\ \citenamefont {Li}}]{xie2026simple}%
  \BibitemOpen
  \bibfield  {author} {\bibinfo {author} {\bibfnamefont {H.}~\bibnamefont {Xie}}, \bibinfo {author} {\bibfnamefont {N.}~\bibnamefont {Yoshioka}}, \bibinfo {author} {\bibfnamefont {K.}~\bibnamefont {Tsubouchi}},\ and\ \bibinfo {author} {\bibfnamefont {Y.}~\bibnamefont {Li}},\ }\bibfield  {title} {\bibinfo {title} {Simple, efficient, and generic post-selection decoding for qldpc codes},\ }\href {https://doi.org/10.48550/arXiv.2601.17757} {\bibfield  {journal} {\bibinfo  {journal} {arXiv:2601.17757}\ } (\bibinfo {year} {2026})}\BibitemShut {NoStop}%
\bibitem [{\citenamefont {Kishi}\ \emph {et~al.}(2026)\citenamefont {Kishi}, \citenamefont {Toshio}, \citenamefont {Fujisaki}, \citenamefont {Oshima}, \citenamefont {Sato},\ and\ \citenamefont {Fujii}}]{kishi2026even}%
  \BibitemOpen
  \bibfield  {author} {\bibinfo {author} {\bibfnamefont {K.}~\bibnamefont {Kishi}}, \bibinfo {author} {\bibfnamefont {R.}~\bibnamefont {Toshio}}, \bibinfo {author} {\bibfnamefont {J.}~\bibnamefont {Fujisaki}}, \bibinfo {author} {\bibfnamefont {H.}~\bibnamefont {Oshima}}, \bibinfo {author} {\bibfnamefont {S.}~\bibnamefont {Sato}},\ and\ \bibinfo {author} {\bibfnamefont {K.}~\bibnamefont {Fujii}},\ }\bibfield  {title} {\bibinfo {title} {Even more efficient soft-output decoding with extra-cluster growth and early stopping},\ }\href {https://doi.org/10.48550/arXiv.2602.03336} {\bibfield  {journal} {\bibinfo  {journal} {arXiv preprint arXiv:2602.03336}\ } (\bibinfo {year} {2026})}\BibitemShut {NoStop}%
\bibitem [{\citenamefont {Zhou}\ \emph {et~al.}(2025)\citenamefont {Zhou}, \citenamefont {Pexton}, \citenamefont {Kubica},\ and\ \citenamefont {Ding}}]{zhou2025error}%
  \BibitemOpen
  \bibfield  {author} {\bibinfo {author} {\bibfnamefont {Z.}~\bibnamefont {Zhou}}, \bibinfo {author} {\bibfnamefont {S.}~\bibnamefont {Pexton}}, \bibinfo {author} {\bibfnamefont {A.}~\bibnamefont {Kubica}},\ and\ \bibinfo {author} {\bibfnamefont {Y.}~\bibnamefont {Ding}},\ }\bibfield  {title} {\bibinfo {title} {Error mitigation of fault-tolerant quantum circuits with soft information},\ }\href {https://doi.org/10.48550/arXiv.2512.09863} {\bibfield  {journal} {\bibinfo  {journal} {arXiv:2512.09863}\ } (\bibinfo {year} {2025})}\BibitemShut {NoStop}%
\bibitem [{\citenamefont {Dinc{\u{a}}}\ \emph {et~al.}(2025)\citenamefont {Dinc{\u{a}}}, \citenamefont {Chan},\ and\ \citenamefont {Benjamin}}]{dincua2025error}%
  \BibitemOpen
  \bibfield  {author} {\bibinfo {author} {\bibfnamefont {M.}~\bibnamefont {Dinc{\u{a}}}}, \bibinfo {author} {\bibfnamefont {T.}~\bibnamefont {Chan}},\ and\ \bibinfo {author} {\bibfnamefont {S.~C.}\ \bibnamefont {Benjamin}},\ }\bibfield  {title} {\bibinfo {title} {Error mitigation for logical circuits using decoder confidence},\ }\href {https://doi.org/10.48550/arXiv.2512.15689} {\bibfield  {journal} {\bibinfo  {journal} {arXiv:2512.15689}\ } (\bibinfo {year} {2025})}\BibitemShut {NoStop}%
\bibitem [{\citenamefont {Aharonov}\ \emph {et~al.}(2025)\citenamefont {Aharonov}, \citenamefont {Atia}, \citenamefont {Bairey}, \citenamefont {Brakerski}, \citenamefont {Cohen}, \citenamefont {Golan}, \citenamefont {Gurwich}, \citenamefont {Lindner},\ and\ \citenamefont {Shutman}}]{aharonov2025syndrome}%
  \BibitemOpen
  \bibfield  {author} {\bibinfo {author} {\bibfnamefont {D.}~\bibnamefont {Aharonov}}, \bibinfo {author} {\bibfnamefont {Y.}~\bibnamefont {Atia}}, \bibinfo {author} {\bibfnamefont {E.}~\bibnamefont {Bairey}}, \bibinfo {author} {\bibfnamefont {Z.}~\bibnamefont {Brakerski}}, \bibinfo {author} {\bibfnamefont {I.}~\bibnamefont {Cohen}}, \bibinfo {author} {\bibfnamefont {O.}~\bibnamefont {Golan}}, \bibinfo {author} {\bibfnamefont {I.}~\bibnamefont {Gurwich}}, \bibinfo {author} {\bibfnamefont {N.~H.}\ \bibnamefont {Lindner}},\ and\ \bibinfo {author} {\bibfnamefont {M.}~\bibnamefont {Shutman}},\ }\bibfield  {title} {\bibinfo {title} {Syndrome aware mitigation of logical errors},\ }\href {https://doi.org/10.48550/arXiv.2512.23810} {\bibfield  {journal} {\bibinfo  {journal} {arXiv:2512.23810}\ } (\bibinfo {year} {2025})}\BibitemShut {NoStop}%
\bibitem [{\citenamefont {Takagi}\ \emph {et~al.}(2022)\citenamefont {Takagi}, \citenamefont {Endo}, \citenamefont {Minagawa},\ and\ \citenamefont {Gu}}]{takagi2022fundamental}%
  \BibitemOpen
  \bibfield  {author} {\bibinfo {author} {\bibfnamefont {R.}~\bibnamefont {Takagi}}, \bibinfo {author} {\bibfnamefont {S.}~\bibnamefont {Endo}}, \bibinfo {author} {\bibfnamefont {S.}~\bibnamefont {Minagawa}},\ and\ \bibinfo {author} {\bibfnamefont {M.}~\bibnamefont {Gu}},\ }\bibfield  {title} {\bibinfo {title} {Fundamental limits of quantum error mitigation},\ }\href {https://doi.org/https://doi.org/10.1038/s41534-022-00618-z} {\bibfield  {journal} {\bibinfo  {journal} {npj Quantum Inf.}\ }\textbf {\bibinfo {volume} {8}},\ \bibinfo {pages} {114} (\bibinfo {year} {2022})}\BibitemShut {NoStop}%
\bibitem [{\citenamefont {Takagi}\ \emph {et~al.}(2023)\citenamefont {Takagi}, \citenamefont {Tajima},\ and\ \citenamefont {Gu}}]{takagi2023universal}%
  \BibitemOpen
  \bibfield  {author} {\bibinfo {author} {\bibfnamefont {R.}~\bibnamefont {Takagi}}, \bibinfo {author} {\bibfnamefont {H.}~\bibnamefont {Tajima}},\ and\ \bibinfo {author} {\bibfnamefont {M.}~\bibnamefont {Gu}},\ }\bibfield  {title} {\bibinfo {title} {Universal sampling lower bounds for quantum error mitigation},\ }\href {https://doi.org/10.1103/PhysRevLett.131.210602} {\bibfield  {journal} {\bibinfo  {journal} {Phys. Rev. Lett.}\ }\textbf {\bibinfo {volume} {131}},\ \bibinfo {pages} {210602} (\bibinfo {year} {2023})}\BibitemShut {NoStop}%
\bibitem [{\citenamefont {Tsubouchi}\ \emph {et~al.}(2023)\citenamefont {Tsubouchi}, \citenamefont {Sagawa},\ and\ \citenamefont {Yoshioka}}]{tsubouchi2023universal}%
  \BibitemOpen
  \bibfield  {author} {\bibinfo {author} {\bibfnamefont {K.}~\bibnamefont {Tsubouchi}}, \bibinfo {author} {\bibfnamefont {T.}~\bibnamefont {Sagawa}},\ and\ \bibinfo {author} {\bibfnamefont {N.}~\bibnamefont {Yoshioka}},\ }\bibfield  {title} {\bibinfo {title} {Universal cost bound of quantum error mitigation based on quantum estimation theory},\ }\href {https://doi.org/10.1103/PhysRevLett.131.210601} {\bibfield  {journal} {\bibinfo  {journal} {Phys. Rev. Lett.}\ }\textbf {\bibinfo {volume} {131}},\ \bibinfo {pages} {210601} (\bibinfo {year} {2023})}\BibitemShut {NoStop}%
\bibitem [{\citenamefont {Quek}\ \emph {et~al.}(2024)\citenamefont {Quek}, \citenamefont {Stilck~Fran{\c{c}}a}, \citenamefont {Khatri}, \citenamefont {Meyer},\ and\ \citenamefont {Eisert}}]{quek2024exponentially}%
  \BibitemOpen
  \bibfield  {author} {\bibinfo {author} {\bibfnamefont {Y.}~\bibnamefont {Quek}}, \bibinfo {author} {\bibfnamefont {D.}~\bibnamefont {Stilck~Fran{\c{c}}a}}, \bibinfo {author} {\bibfnamefont {S.}~\bibnamefont {Khatri}}, \bibinfo {author} {\bibfnamefont {J.~J.}\ \bibnamefont {Meyer}},\ and\ \bibinfo {author} {\bibfnamefont {J.}~\bibnamefont {Eisert}},\ }\bibfield  {title} {\bibinfo {title} {Exponentially tighter bounds on limitations of quantum error mitigation},\ }\href {https://doi.org/https://doi.org/10.1038/s41567-024-02536-7} {\bibfield  {journal} {\bibinfo  {journal} {Nat. Phys.}\ }\textbf {\bibinfo {volume} {20}},\ \bibinfo {pages} {1648} (\bibinfo {year} {2024})}\BibitemShut {NoStop}%
\bibitem [{\citenamefont {Helstrom}(1969)}]{helstrom1969quantum}%
  \BibitemOpen
  \bibfield  {author} {\bibinfo {author} {\bibfnamefont {C.~W.}\ \bibnamefont {Helstrom}},\ }\bibfield  {title} {\bibinfo {title} {Quantum detection and estimation theory},\ }\href {https://doi.org/https://doi.org/10.1007/BF01007479} {\bibfield  {journal} {\bibinfo  {journal} {J. Stat. Phys.}\ }\textbf {\bibinfo {volume} {1}},\ \bibinfo {pages} {231} (\bibinfo {year} {1969})}\BibitemShut {NoStop}%
\bibitem [{\citenamefont {Holevo}(2011)}]{holevo2011probabilistic}%
  \BibitemOpen
  \bibfield  {author} {\bibinfo {author} {\bibfnamefont {A.~S.}\ \bibnamefont {Holevo}},\ }\href {https://doi.org/https://doi.org/10.1007/978-88-7642-378-9} {\emph {\bibinfo {title} {Probabilistic and statistical aspects of quantum theory}}},\ Vol.~\bibinfo {volume} {1}\ (\bibinfo  {publisher} {Springer Science \& Business Media},\ \bibinfo {year} {2011})\BibitemShut {NoStop}%
\bibitem [{\citenamefont {Hayashi}(2006)}]{hayashi2006quantum}%
  \BibitemOpen
  \bibfield  {author} {\bibinfo {author} {\bibfnamefont {M.}~\bibnamefont {Hayashi}},\ }\href {https://doi.org/https://doi.org/10.1007/3-540-30266-2} {\emph {\bibinfo {title} {Quantum information}}}\ (\bibinfo  {publisher} {Springer},\ \bibinfo {year} {2006})\BibitemShut {NoStop}%
\bibitem [{\citenamefont {Suzuki}\ \emph {et~al.}(2020)\citenamefont {Suzuki}, \citenamefont {Yang},\ and\ \citenamefont {Hayashi}}]{suzuki2020quantum}%
  \BibitemOpen
  \bibfield  {author} {\bibinfo {author} {\bibfnamefont {J.}~\bibnamefont {Suzuki}}, \bibinfo {author} {\bibfnamefont {Y.}~\bibnamefont {Yang}},\ and\ \bibinfo {author} {\bibfnamefont {M.}~\bibnamefont {Hayashi}},\ }\bibfield  {title} {\bibinfo {title} {Quantum state estimation with nuisance parameters},\ }\href {https://doi.org/https://doi.org/10.1088/1751-8121/ab8b78} {\bibfield  {journal} {\bibinfo  {journal} {Journal of Physics A: Mathematical and Theoretical}\ }\textbf {\bibinfo {volume} {53}},\ \bibinfo {pages} {453001} (\bibinfo {year} {2020})}\BibitemShut {NoStop}%
\bibitem [{\citenamefont {Yang}\ \emph {et~al.}(2019)\citenamefont {Yang}, \citenamefont {Chiribella},\ and\ \citenamefont {Hayashi}}]{yang2019attaining}%
  \BibitemOpen
  \bibfield  {author} {\bibinfo {author} {\bibfnamefont {Y.}~\bibnamefont {Yang}}, \bibinfo {author} {\bibfnamefont {G.}~\bibnamefont {Chiribella}},\ and\ \bibinfo {author} {\bibfnamefont {M.}~\bibnamefont {Hayashi}},\ }\bibfield  {title} {\bibinfo {title} {Attaining the ultimate precision limit in quantum state estimation},\ }\href {https://doi.org/https://doi.org/10.1007/s00220-019-03433-4} {\bibfield  {journal} {\bibinfo  {journal} {Commun. Math. Phys.}\ }\textbf {\bibinfo {volume} {368}},\ \bibinfo {pages} {223} (\bibinfo {year} {2019})}\BibitemShut {NoStop}%
\bibitem [{\citenamefont {Zhou}\ \emph {et~al.}(2020)\citenamefont {Zhou}, \citenamefont {Zou},\ and\ \citenamefont {Jiang}}]{zhou2020saturating}%
  \BibitemOpen
  \bibfield  {author} {\bibinfo {author} {\bibfnamefont {S.}~\bibnamefont {Zhou}}, \bibinfo {author} {\bibfnamefont {C.-L.}\ \bibnamefont {Zou}},\ and\ \bibinfo {author} {\bibfnamefont {L.}~\bibnamefont {Jiang}},\ }\bibfield  {title} {\bibinfo {title} {Saturating the quantum cram{\'e}r--rao bound using locc},\ }\href {https://doi.org/https://doi.org/10.1088/2058-9565/ab71f8} {\bibfield  {journal} {\bibinfo  {journal} {Quantum Sci. Technol.}\ }\textbf {\bibinfo {volume} {5}},\ \bibinfo {pages} {025005} (\bibinfo {year} {2020})}\BibitemShut {NoStop}%
\bibitem [{\citenamefont {Kimura}(2003)}]{kimura2003bloch}%
  \BibitemOpen
  \bibfield  {author} {\bibinfo {author} {\bibfnamefont {G.}~\bibnamefont {Kimura}},\ }\bibfield  {title} {\bibinfo {title} {The bloch vector for n-level systems},\ }\href {https://doi.org/https://doi.org/10.1016/S0375-9601(03)00941-1} {\bibfield  {journal} {\bibinfo  {journal} {Phys. Lett. A}\ }\textbf {\bibinfo {volume} {314}},\ \bibinfo {pages} {339} (\bibinfo {year} {2003})}\BibitemShut {NoStop}%
\bibitem [{\citenamefont {Iyer}\ and\ \citenamefont {Poulin}(2015)}]{iyer2015hardness}%
  \BibitemOpen
  \bibfield  {author} {\bibinfo {author} {\bibfnamefont {P.}~\bibnamefont {Iyer}}\ and\ \bibinfo {author} {\bibfnamefont {D.}~\bibnamefont {Poulin}},\ }\bibfield  {title} {\bibinfo {title} {Hardness of decoding quantum stabilizer codes},\ }\href {https://doi.org/https://doi.org/10.1109/TIT.2015.2422294} {\bibfield  {journal} {\bibinfo  {journal} {IEEE Trans. Inf. Theory}\ }\textbf {\bibinfo {volume} {61}},\ \bibinfo {pages} {5209} (\bibinfo {year} {2015})}\BibitemShut {NoStop}%
\bibitem [{\citenamefont {Fuentes}\ \emph {et~al.}(2021)\citenamefont {Fuentes}, \citenamefont {Martinez}, \citenamefont {Crespo},\ and\ \citenamefont {Garcia-Fr{\'\i}as}}]{fuentes2021degeneracy}%
  \BibitemOpen
  \bibfield  {author} {\bibinfo {author} {\bibfnamefont {P.}~\bibnamefont {Fuentes}}, \bibinfo {author} {\bibfnamefont {J.~E.}\ \bibnamefont {Martinez}}, \bibinfo {author} {\bibfnamefont {P.~M.}\ \bibnamefont {Crespo}},\ and\ \bibinfo {author} {\bibfnamefont {J.}~\bibnamefont {Garcia-Fr{\'\i}as}},\ }\bibfield  {title} {\bibinfo {title} {Degeneracy and its impact on the decoding of sparse quantum codes},\ }\href {https://doi.org/https://doi.org/10.1109/ACCESS.2021.3089829} {\bibfield  {journal} {\bibinfo  {journal} {IEEE Access}\ }\textbf {\bibinfo {volume} {9}},\ \bibinfo {pages} {89093} (\bibinfo {year} {2021})}\BibitemShut {NoStop}%
\bibitem [{\citenamefont {deMarti iOlius}\ \emph {et~al.}(2024)\citenamefont {deMarti iOlius}, \citenamefont {Fuentes}, \citenamefont {Or{\'u}s}, \citenamefont {Crespo},\ and\ \citenamefont {Martinez}}]{demarti2024decoding}%
  \BibitemOpen
  \bibfield  {author} {\bibinfo {author} {\bibfnamefont {A.}~\bibnamefont {deMarti iOlius}}, \bibinfo {author} {\bibfnamefont {P.}~\bibnamefont {Fuentes}}, \bibinfo {author} {\bibfnamefont {R.}~\bibnamefont {Or{\'u}s}}, \bibinfo {author} {\bibfnamefont {P.~M.}\ \bibnamefont {Crespo}},\ and\ \bibinfo {author} {\bibfnamefont {J.~E.}\ \bibnamefont {Martinez}},\ }\bibfield  {title} {\bibinfo {title} {Decoding algorithms for surface codes},\ }\href {https://doi.org/https://doi.org/10.22331/q-2024-10-10-1498} {\bibfield  {journal} {\bibinfo  {journal} {Quantum}\ }\textbf {\bibinfo {volume} {8}},\ \bibinfo {pages} {1498} (\bibinfo {year} {2024})}\BibitemShut {NoStop}%
\bibitem [{\citenamefont {Watanabe}\ \emph {et~al.}(2010)\citenamefont {Watanabe}, \citenamefont {Sagawa},\ and\ \citenamefont {Ueda}}]{watanabe2010optimal}%
  \BibitemOpen
  \bibfield  {author} {\bibinfo {author} {\bibfnamefont {Y.}~\bibnamefont {Watanabe}}, \bibinfo {author} {\bibfnamefont {T.}~\bibnamefont {Sagawa}},\ and\ \bibinfo {author} {\bibfnamefont {M.}~\bibnamefont {Ueda}},\ }\bibfield  {title} {\bibinfo {title} {Optimal measurement on noisy quantum systems},\ }\href {https://doi.org/10.1103/PhysRevLett.104.020401} {\bibfield  {journal} {\bibinfo  {journal} {Phys. Rev. Lett.}\ }\textbf {\bibinfo {volume} {104}},\ \bibinfo {pages} {020401} (\bibinfo {year} {2010})}\BibitemShut {NoStop}%
\bibitem [{\citenamefont {Li}\ and\ \citenamefont {Benjamin}(2017)}]{li2017efficient}%
  \BibitemOpen
  \bibfield  {author} {\bibinfo {author} {\bibfnamefont {Y.}~\bibnamefont {Li}}\ and\ \bibinfo {author} {\bibfnamefont {S.~C.}\ \bibnamefont {Benjamin}},\ }\bibfield  {title} {\bibinfo {title} {Efficient variational quantum simulator incorporating active error minimization},\ }\href {https://doi.org/10.1103/PhysRevX.7.021050} {\bibfield  {journal} {\bibinfo  {journal} {Phys. Rev. X}\ }\textbf {\bibinfo {volume} {7}},\ \bibinfo {pages} {021050} (\bibinfo {year} {2017})}\BibitemShut {NoStop}%
\bibitem [{\citenamefont {Endo}\ \emph {et~al.}(2018)\citenamefont {Endo}, \citenamefont {Benjamin},\ and\ \citenamefont {Li}}]{endo2018practical}%
  \BibitemOpen
  \bibfield  {author} {\bibinfo {author} {\bibfnamefont {S.}~\bibnamefont {Endo}}, \bibinfo {author} {\bibfnamefont {S.~C.}\ \bibnamefont {Benjamin}},\ and\ \bibinfo {author} {\bibfnamefont {Y.}~\bibnamefont {Li}},\ }\bibfield  {title} {\bibinfo {title} {Practical quantum error mitigation for near-future applications},\ }\href {https://doi.org/10.1103/PhysRevX.8.031027} {\bibfield  {journal} {\bibinfo  {journal} {Phys. Rev. X}\ }\textbf {\bibinfo {volume} {8}},\ \bibinfo {pages} {031027} (\bibinfo {year} {2018})}\BibitemShut {NoStop}%
\bibitem [{\citenamefont {Van Den~Berg}\ \emph {et~al.}(2023)\citenamefont {Van Den~Berg}, \citenamefont {Minev}, \citenamefont {Kandala},\ and\ \citenamefont {Temme}}]{van2023probabilistic}%
  \BibitemOpen
  \bibfield  {author} {\bibinfo {author} {\bibfnamefont {E.}~\bibnamefont {Van Den~Berg}}, \bibinfo {author} {\bibfnamefont {Z.~K.}\ \bibnamefont {Minev}}, \bibinfo {author} {\bibfnamefont {A.}~\bibnamefont {Kandala}},\ and\ \bibinfo {author} {\bibfnamefont {K.}~\bibnamefont {Temme}},\ }\bibfield  {title} {\bibinfo {title} {Probabilistic error cancellation with sparse pauli--lindblad models on noisy quantum processors},\ }\href {https://doi.org/https://doi.org/10.1038/s41567-023-02042-2} {\bibfield  {journal} {\bibinfo  {journal} {Nat. Phys.}\ }\textbf {\bibinfo {volume} {19}},\ \bibinfo {pages} {1116} (\bibinfo {year} {2023})}\BibitemShut {NoStop}%
\bibitem [{\citenamefont {Bonet-Monroig}\ \emph {et~al.}(2018)\citenamefont {Bonet-Monroig}, \citenamefont {Sagastizabal}, \citenamefont {Singh},\ and\ \citenamefont {O'Brien}}]{bonet2018low}%
  \BibitemOpen
  \bibfield  {author} {\bibinfo {author} {\bibfnamefont {X.}~\bibnamefont {Bonet-Monroig}}, \bibinfo {author} {\bibfnamefont {R.}~\bibnamefont {Sagastizabal}}, \bibinfo {author} {\bibfnamefont {M.}~\bibnamefont {Singh}},\ and\ \bibinfo {author} {\bibfnamefont {T.~E.}\ \bibnamefont {O'Brien}},\ }\bibfield  {title} {\bibinfo {title} {Low-cost error mitigation by symmetry verification},\ }\href {https://doi.org/10.1103/PhysRevA.98.062339} {\bibfield  {journal} {\bibinfo  {journal} {Phys. Rev. A}\ }\textbf {\bibinfo {volume} {98}},\ \bibinfo {pages} {062339} (\bibinfo {year} {2018})}\BibitemShut {NoStop}%
\bibitem [{\citenamefont {McArdle}\ \emph {et~al.}(2019)\citenamefont {McArdle}, \citenamefont {Yuan},\ and\ \citenamefont {Benjamin}}]{mcardle2019error}%
  \BibitemOpen
  \bibfield  {author} {\bibinfo {author} {\bibfnamefont {S.}~\bibnamefont {McArdle}}, \bibinfo {author} {\bibfnamefont {X.}~\bibnamefont {Yuan}},\ and\ \bibinfo {author} {\bibfnamefont {S.}~\bibnamefont {Benjamin}},\ }\bibfield  {title} {\bibinfo {title} {Error-mitigated digital quantum simulation},\ }\href {https://doi.org/10.1103/PhysRevLett.122.180501} {\bibfield  {journal} {\bibinfo  {journal} {Phys. Rev. Lett.}\ }\textbf {\bibinfo {volume} {122}},\ \bibinfo {pages} {180501} (\bibinfo {year} {2019})}\BibitemShut {NoStop}%
\bibitem [{\citenamefont {Czarnik}\ \emph {et~al.}(2021)\citenamefont {Czarnik}, \citenamefont {Arrasmith}, \citenamefont {Coles},\ and\ \citenamefont {Cincio}}]{czarnik2021error}%
  \BibitemOpen
  \bibfield  {author} {\bibinfo {author} {\bibfnamefont {P.}~\bibnamefont {Czarnik}}, \bibinfo {author} {\bibfnamefont {A.}~\bibnamefont {Arrasmith}}, \bibinfo {author} {\bibfnamefont {P.~J.}\ \bibnamefont {Coles}},\ and\ \bibinfo {author} {\bibfnamefont {L.}~\bibnamefont {Cincio}},\ }\bibfield  {title} {\bibinfo {title} {Error mitigation with clifford quantum-circuit data},\ }\href {https://doi.org/https://doi.org/10.22331/q-2021-11-26-592} {\bibfield  {journal} {\bibinfo  {journal} {Quantum}\ }\textbf {\bibinfo {volume} {5}},\ \bibinfo {pages} {592} (\bibinfo {year} {2021})}\BibitemShut {NoStop}%
\bibitem [{\citenamefont {Strikis}\ \emph {et~al.}(2021)\citenamefont {Strikis}, \citenamefont {Qin}, \citenamefont {Chen}, \citenamefont {Benjamin},\ and\ \citenamefont {Li}}]{strinkis2021learning}%
  \BibitemOpen
  \bibfield  {author} {\bibinfo {author} {\bibfnamefont {A.}~\bibnamefont {Strikis}}, \bibinfo {author} {\bibfnamefont {D.}~\bibnamefont {Qin}}, \bibinfo {author} {\bibfnamefont {Y.}~\bibnamefont {Chen}}, \bibinfo {author} {\bibfnamefont {S.~C.}\ \bibnamefont {Benjamin}},\ and\ \bibinfo {author} {\bibfnamefont {Y.}~\bibnamefont {Li}},\ }\bibfield  {title} {\bibinfo {title} {Learning-based quantum error mitigation},\ }\href {https://doi.org/10.1103/PRXQuantum.2.040330} {\bibfield  {journal} {\bibinfo  {journal} {PRX Quantum}\ }\textbf {\bibinfo {volume} {2}},\ \bibinfo {pages} {040330} (\bibinfo {year} {2021})}\BibitemShut {NoStop}%
\bibitem [{\citenamefont {Paetznick}\ \emph {et~al.}(2024)\citenamefont {Paetznick}, \citenamefont {Da~Silva}, \citenamefont {Ryan-Anderson}, \citenamefont {Bello-Rivas}, \citenamefont {Campora~III}, \citenamefont {Chernoguzov}, \citenamefont {Dreiling}, \citenamefont {Foltz}, \citenamefont {Frachon}, \citenamefont {Gaebler} \emph {et~al.}}]{paetznick2024demonstration}%
  \BibitemOpen
  \bibfield  {author} {\bibinfo {author} {\bibfnamefont {A.}~\bibnamefont {Paetznick}}, \bibinfo {author} {\bibfnamefont {M.}~\bibnamefont {Da~Silva}}, \bibinfo {author} {\bibfnamefont {C.}~\bibnamefont {Ryan-Anderson}}, \bibinfo {author} {\bibfnamefont {J.}~\bibnamefont {Bello-Rivas}}, \bibinfo {author} {\bibfnamefont {J.}~\bibnamefont {Campora~III}}, \bibinfo {author} {\bibfnamefont {A.}~\bibnamefont {Chernoguzov}}, \bibinfo {author} {\bibfnamefont {J.}~\bibnamefont {Dreiling}}, \bibinfo {author} {\bibfnamefont {C.}~\bibnamefont {Foltz}}, \bibinfo {author} {\bibfnamefont {F.}~\bibnamefont {Frachon}}, \bibinfo {author} {\bibfnamefont {J.}~\bibnamefont {Gaebler}}, \emph {et~al.},\ }\bibfield  {title} {\bibinfo {title} {Demonstration of logical qubits and repeated error correction with better-than-physical error rates},\ }\href {https://doi.org/10.48550/arXiv.2404.02280} {\bibfield  {journal} {\bibinfo  {journal} {arXiv:2404.02280}\ } (\bibinfo {year} {2024})}\BibitemShut {NoStop}%
\bibitem [{\citenamefont {Dennis}\ \emph {et~al.}(2002)\citenamefont {Dennis}, \citenamefont {Kitaev}, \citenamefont {Landahl},\ and\ \citenamefont {Preskill}}]{dennis2002topological}%
  \BibitemOpen
  \bibfield  {author} {\bibinfo {author} {\bibfnamefont {E.}~\bibnamefont {Dennis}}, \bibinfo {author} {\bibfnamefont {A.}~\bibnamefont {Kitaev}}, \bibinfo {author} {\bibfnamefont {A.}~\bibnamefont {Landahl}},\ and\ \bibinfo {author} {\bibfnamefont {J.}~\bibnamefont {Preskill}},\ }\bibfield  {title} {\bibinfo {title} {Topological quantum memory},\ }\href {https://doi.org/https://doi.org/10.1063/1.1499754} {\bibfield  {journal} {\bibinfo  {journal} {J. Math. Phys.}\ }\textbf {\bibinfo {volume} {43}},\ \bibinfo {pages} {4452} (\bibinfo {year} {2002})}\BibitemShut {NoStop}%
\bibitem [{\citenamefont {Higgott}\ and\ \citenamefont {Gidney}(2025)}]{higgott2025sparse}%
  \BibitemOpen
  \bibfield  {author} {\bibinfo {author} {\bibfnamefont {O.}~\bibnamefont {Higgott}}\ and\ \bibinfo {author} {\bibfnamefont {C.}~\bibnamefont {Gidney}},\ }\bibfield  {title} {\bibinfo {title} {Sparse blossom: correcting a million errors per core second with minimum-weight matching},\ }\href {https://doi.org/https://doi.org/10.22331/q-2025-01-20-1600} {\bibfield  {journal} {\bibinfo  {journal} {Quantum}\ }\textbf {\bibinfo {volume} {9}},\ \bibinfo {pages} {1600} (\bibinfo {year} {2025})}\BibitemShut {NoStop}%
\bibitem [{\citenamefont {Panteleev}\ and\ \citenamefont {Kalachev}(2021)}]{panteleev2021degenerate}%
  \BibitemOpen
  \bibfield  {author} {\bibinfo {author} {\bibfnamefont {P.}~\bibnamefont {Panteleev}}\ and\ \bibinfo {author} {\bibfnamefont {G.}~\bibnamefont {Kalachev}},\ }\bibfield  {title} {\bibinfo {title} {Degenerate quantum ldpc codes with good finite length performance},\ }\href {https://doi.org/https://doi.org/10.22331/q-2021-11-22-585} {\bibfield  {journal} {\bibinfo  {journal} {Quantum}\ }\textbf {\bibinfo {volume} {5}},\ \bibinfo {pages} {585} (\bibinfo {year} {2021})}\BibitemShut {NoStop}%
\bibitem [{\citenamefont {Roffe}\ \emph {et~al.}(2020)\citenamefont {Roffe}, \citenamefont {White}, \citenamefont {Burton},\ and\ \citenamefont {Campbell}}]{roffe2020decoding}%
  \BibitemOpen
  \bibfield  {author} {\bibinfo {author} {\bibfnamefont {J.}~\bibnamefont {Roffe}}, \bibinfo {author} {\bibfnamefont {D.~R.}\ \bibnamefont {White}}, \bibinfo {author} {\bibfnamefont {S.}~\bibnamefont {Burton}},\ and\ \bibinfo {author} {\bibfnamefont {E.}~\bibnamefont {Campbell}},\ }\bibfield  {title} {\bibinfo {title} {Decoding across the quantum low-density parity-check code landscape},\ }\href {https://doi.org/10.1103/PhysRevResearch.2.043423} {\bibfield  {journal} {\bibinfo  {journal} {Phys. Rev. Res.}\ }\textbf {\bibinfo {volume} {2}},\ \bibinfo {pages} {043423} (\bibinfo {year} {2020})}\BibitemShut {NoStop}%
\bibitem [{\citenamefont {M{\"u}ller}\ \emph {et~al.}(2025)\citenamefont {M{\"u}ller}, \citenamefont {Alexander}, \citenamefont {Beverland}, \citenamefont {B{\"u}hler}, \citenamefont {Johnson}, \citenamefont {Maurer},\ and\ \citenamefont {Vandeth}}]{muller2025improved}%
  \BibitemOpen
  \bibfield  {author} {\bibinfo {author} {\bibfnamefont {T.}~\bibnamefont {M{\"u}ller}}, \bibinfo {author} {\bibfnamefont {T.}~\bibnamefont {Alexander}}, \bibinfo {author} {\bibfnamefont {M.~E.}\ \bibnamefont {Beverland}}, \bibinfo {author} {\bibfnamefont {M.}~\bibnamefont {B{\"u}hler}}, \bibinfo {author} {\bibfnamefont {B.~R.}\ \bibnamefont {Johnson}}, \bibinfo {author} {\bibfnamefont {T.}~\bibnamefont {Maurer}},\ and\ \bibinfo {author} {\bibfnamefont {D.}~\bibnamefont {Vandeth}},\ }\bibfield  {title} {\bibinfo {title} {Improved belief propagation is sufficient for real-time decoding of quantum memory},\ }\href {https://doi.org/10.48550/arXiv.2506.01779} {\bibfield  {journal} {\bibinfo  {journal} {arXiv:2506.01779}\ } (\bibinfo {year} {2025})}\BibitemShut {NoStop}%
\bibitem [{\citenamefont {Tsubouchi}\ \emph {et~al.}(2025{\natexlab{b}})\citenamefont {Tsubouchi}, \citenamefont {Yamasaki},\ and\ \citenamefont {Tamiya}}]{tsubouchi2025degeneracy}%
  \BibitemOpen
  \bibfield  {author} {\bibinfo {author} {\bibfnamefont {K.}~\bibnamefont {Tsubouchi}}, \bibinfo {author} {\bibfnamefont {H.}~\bibnamefont {Yamasaki}},\ and\ \bibinfo {author} {\bibfnamefont {S.}~\bibnamefont {Tamiya}},\ }\bibfield  {title} {\bibinfo {title} {Degeneracy cutting: A local and efficient post-processing for belief propagation decoding of quantum low-density parity-check codes},\ }\href {https://doi.org/10.48550/arXiv.2510.08695} {\bibfield  {journal} {\bibinfo  {journal} {arXiv:2510.08695}\ } (\bibinfo {year} {2025}{\natexlab{b}})}\BibitemShut {NoStop}%
\bibitem [{\citenamefont {Hutter}\ \emph {et~al.}(2014)\citenamefont {Hutter}, \citenamefont {Wootton},\ and\ \citenamefont {Loss}}]{hutter2014efficient}%
  \BibitemOpen
  \bibfield  {author} {\bibinfo {author} {\bibfnamefont {A.}~\bibnamefont {Hutter}}, \bibinfo {author} {\bibfnamefont {J.~R.}\ \bibnamefont {Wootton}},\ and\ \bibinfo {author} {\bibfnamefont {D.}~\bibnamefont {Loss}},\ }\bibfield  {title} {\bibinfo {title} {Efficient markov chain monte carlo algorithm for the surface code},\ }\href {https://doi.org/10.1103/PhysRevA.89.022326} {\bibfield  {journal} {\bibinfo  {journal} {Phys. Rev. A}\ }\textbf {\bibinfo {volume} {89}},\ \bibinfo {pages} {022326} (\bibinfo {year} {2014})}\BibitemShut {NoStop}%
\bibitem [{\citenamefont {Gidney}(2021)}]{gidney2021stim}%
  \BibitemOpen
  \bibfield  {author} {\bibinfo {author} {\bibfnamefont {C.}~\bibnamefont {Gidney}},\ }\bibfield  {title} {\bibinfo {title} {Stim: a fast stabilizer circuit simulator},\ }\href {https://doi.org/10.22331/q-2021-07-06-497} {\bibfield  {journal} {\bibinfo  {journal} {{Quantum}}\ }\textbf {\bibinfo {volume} {5}},\ \bibinfo {pages} {497} (\bibinfo {year} {2021})}\BibitemShut {NoStop}%
\bibitem [{\citenamefont {Kueng}\ and\ \citenamefont {Gross}(2015)}]{kueng2015qubit}%
  \BibitemOpen
  \bibfield  {author} {\bibinfo {author} {\bibfnamefont {R.}~\bibnamefont {Kueng}}\ and\ \bibinfo {author} {\bibfnamefont {D.}~\bibnamefont {Gross}},\ }\bibfield  {title} {\bibinfo {title} {Qubit stabilizer states are complex projective 3-designs},\ }\href@noop {} {\bibfield  {journal} {\bibinfo  {journal} {arXiv preprint arXiv:1510.02767}\ } (\bibinfo {year} {2015})}\BibitemShut {NoStop}%
\bibitem [{\citenamefont {Zhou}\ \emph {et~al.}(2018)\citenamefont {Zhou}, \citenamefont {Zhang}, \citenamefont {Preskill},\ and\ \citenamefont {Jiang}}]{zhou2018achieving}%
  \BibitemOpen
  \bibfield  {author} {\bibinfo {author} {\bibfnamefont {S.}~\bibnamefont {Zhou}}, \bibinfo {author} {\bibfnamefont {M.}~\bibnamefont {Zhang}}, \bibinfo {author} {\bibfnamefont {J.}~\bibnamefont {Preskill}},\ and\ \bibinfo {author} {\bibfnamefont {L.}~\bibnamefont {Jiang}},\ }\bibfield  {title} {\bibinfo {title} {Achieving the heisenberg limit in quantum metrology using quantum error correction},\ }\href {https://doi.org/https://doi.org/10.1038/s41467-017-02510-3} {\bibfield  {journal} {\bibinfo  {journal} {Nature communications}\ }\textbf {\bibinfo {volume} {9}},\ \bibinfo {pages} {78} (\bibinfo {year} {2018})}\BibitemShut {NoStop}%
\bibitem [{\citenamefont {Zhou}\ and\ \citenamefont {Jiang}(2021)}]{zhou2021asymptotic}%
  \BibitemOpen
  \bibfield  {author} {\bibinfo {author} {\bibfnamefont {S.}~\bibnamefont {Zhou}}\ and\ \bibinfo {author} {\bibfnamefont {L.}~\bibnamefont {Jiang}},\ }\bibfield  {title} {\bibinfo {title} {Asymptotic theory of quantum channel estimation},\ }\href {https://doi.org/10.1103/PRXQuantum.2.010343} {\bibfield  {journal} {\bibinfo  {journal} {PRX Quantum}\ }\textbf {\bibinfo {volume} {2}},\ \bibinfo {pages} {010343} (\bibinfo {year} {2021})}\BibitemShut {NoStop}%
\bibitem [{\citenamefont {Jeon}\ and\ \citenamefont {Cai}(2026)}]{jeon2026quantum}%
  \BibitemOpen
  \bibfield  {author} {\bibinfo {author} {\bibfnamefont {M.}~\bibnamefont {Jeon}}\ and\ \bibinfo {author} {\bibfnamefont {Z.}~\bibnamefont {Cai}},\ }\bibfield  {title} {\bibinfo {title} {Quantum error correction on error-mitigated physical qubits},\ }\href {https://doi.org/10.48550/arXiv.2601.18384} {\bibfield  {journal} {\bibinfo  {journal} {arXiv:2601.18384}\ } (\bibinfo {year} {2026})}\BibitemShut {NoStop}%
\bibitem [{\citenamefont {Mele}(2024)}]{mele2024introduction}%
  \BibitemOpen
  \bibfield  {author} {\bibinfo {author} {\bibfnamefont {A.~A.}\ \bibnamefont {Mele}},\ }\bibfield  {title} {\bibinfo {title} {Introduction to haar measure tools in quantum information: A beginner's tutorial},\ }\href {https://doi.org/https://doi.org/10.22331/q-2024-05-08-1340} {\bibfield  {journal} {\bibinfo  {journal} {Quantum}\ }\textbf {\bibinfo {volume} {8}},\ \bibinfo {pages} {1340} (\bibinfo {year} {2024})}\BibitemShut {NoStop}%
\end{thebibliography}%
\end{document}